\documentclass[a4p,11pt]{article}

\usepackage{amsfonts,amsmath,amssymb,amsthm,ascmac,mathrsfs,mathtools,caption,comment,enumerate,bm,multicol,multirow,setspace,lmodern,microtype,physics,paralist,graphicx,hyperref,blindtext,booktabs}

\usepackage[dvipsnames]{xcolor}
\usepackage{tikz}
\usetikzlibrary{calc, intersections}

\usepackage[right=1in, left=1in, bottom=1in, top=1in]{geometry}
\hypersetup{colorlinks=true, linkcolor=black, citecolor=black}

\usepackage{natbib}
\setcitestyle{authoryear,open={(},close={)}}

\newtheorem{theorem}{Theorem}
\newtheorem{proposition}{Proposition}
\newtheorem{corollary}{Corollary}
\newtheorem{lemma}{Lemma}

\theoremstyle{definition}

\newtheorem*{assumption*}{Assumption}

\newtheorem{example}{Example}

\theoremstyle{remark}
\newtheorem{remark}{Remark}

\newcommand{\N}{\mathbb{N}}
\newcommand{\R}{\mathbb{R}}

\newcommand{\1}{\mathbf{1}}
\newcommand{\6}{\partial}

\newcommand{\calA}{\mathcal{A}}

\newcommand{\calD}{\mathcal{D}}

\newcommand{\calL}{\mathcal{L}}

\newcommand{\argmax}{\mathop{\rm arg~max}\limits}

\DeclareMathOperator{\Var}{\mathbb{V}ar}

\DeclareMathOperator{\Diag}{Diag}

\DeclareMathOperator{\Span}{span}

\begin{document}

\title{\bf Characteristics Design: \\ 
\Large A Hedonic Approach to Optimal Product Differentiation}
\author{Masaki Miyashita\thanks{The University of Hong Kong, {\tt masaki11@hku.hk}}}
\date{\today}
\maketitle

\begin{abstract}
Building on the generalized hedonic-linear model of \cite{pellegrino2025}, this paper studies optimal product differentiation when a representative consumer has preferences over product characteristics.
Under multiproduct monopoly, the monopolist's choice of product characteristics is always aligned with the social planner's optimum, despite underproduction.
By contrast, under oligopoly, multiple equilibria can arise that differ qualitatively in their patterns of characteristics design.
We show that, while oligopoly equilibria exhibiting product differentiation yield higher welfare than those with product concentration, the degree of product differentiation under oligopoly remains below the socially optimal level.
As a result, social welfare under oligopoly is typically lower than under monopoly, highlighting a key advantage of coordination in product differentiation.
We extend the analysis to settings with overlapping ownership structures and show that common ownership can improve welfare by inducing firms to soften competition through increased product differentiation rather than output reduction.

\end{abstract}


\hypersetup{colorlinks=true, linkcolor=Red, citecolor=Blue}
\section{Introduction}

Product differentiation lies at the core of industrial organization.
Firms operating in competitive markets can alter their profitability through the design of product characteristics \citep{hotelling1929}, which in turn shapes consumer welfare \citep{lancaster1975, spence1976}.
This naturally raises several fundamental questions.
What patterns of product differentiation arise from firms' strategic decisions?
What are the welfare consequences of the resulting market outcomes?
How does market structure affect these implications?

This paper addresses these questions by extending the generalized hedonic-linear (GHL) model of \cite{pellegrino2025}.
Following the tradition of hedonic demand literature \citep{lancaster1966}, each firm's product is perceived as a bundle of characteristics in the GHL model.
These characteristics are categorized into \emph{common characteristics}, which can be shared across all products, and \emph{idiosyncratic characteristics}, which are firm-specific and cannot be copied by rivals.
A representative consumer has quadratic preferences over linear aggregates of these characteristics.
The key departure of this paper is that, while common characteristics are exogenously given in \cite{pellegrino2025}, we allow firms to choose them endogenously as part of their strategic decisions.

Two market structures are considered in our baseline analysis.
The first is a (multiproduct) monopoly, where a single firm designs the characteristics and quantities of all $n$ products.
Not only serving as a benchmark for comparison, this case is of independent interest, as it can be interpreted as a model of collusion \citep{deneckere1983, majerus1988}, a merger counterfactual \citep{farrell1990shapiro}, or an extreme form of common ownership \citep{rotemberg1984}.
The second is an oligopoly, where $n$ firms each design and produce one product and compete in a Cournot market.
We offer a complete characterization of equilibria under each market structure and conduct a systematic comparison that yields sharp welfare implications, which highlight the advantage of coordination in characteristics design.

Our first set of results concerns the social planner's and monopolist's benchmarks.
While monopoly entails underproduction, we show that the monopolist's choice of product characteristics is always aligned with the social planner's optimum.
As a result, both common and idiosyncratic characteristics provided under monopoly are proportional to those at the social optimum, albeit scaled down by one half.
The welfare implication mirrors that in multiproduct monopoly models without characteristics design \citep{amiretal2016}, in the sense that the deadweight loss induced by monopoly is a fixed proportion of the social optimum in our setting as well.

Turning to oligopoly, Theorem~\ref{thm_eq} reveals that multiple equilibria can emerge that differ qualitatively in how firms position their products in the characteristics space.
We identify three equilibrium patterns, which may coexist under the same parameter specification.
A key determinant of equilibrium structure is what we call firms' \emph{standalone values}, defined as the value of each firm's idiosyncratic characteristic for the consumer, net of the firm's marginal cost of production.
These can be interpreted as each firm's brand value that cannot be copied by rivals and is determined independently of common-characteristics design.

First, \emph{product differentiation} arises when firms' standalone values are collectively large and there is no single firm whose standalone value far exceeds those of others.
This case is arguably the most empirically relevant, given the pervasive coexistence of differentiated products in real-world markets.
The equilibrium is characterized explicitly in terms of primitives, allowing us to quantify the equilibrium degree of product differentiation.
We show that firms collectively supply the socially optimal aggregate level of common characteristics, yet individual firms' choices are closer to one another than under the social planner's optimum.
In this sense, our results echo Hotelling's classical insight that firms tend to locate too closely in equilibrium, in contrast to the principle of maximal differentiation as in \cite{daspremont1979}.
This distortion in characteristics design generates welfare losses under oligopoly, which are absent under monopoly, where characteristics design is fully internalized.
Specifically, Theorem~\ref{thm_welfare_D} shows that when product differentiation is socially optimal, oligopoly typically yields lower welfare than monopoly.

Second, when standalone values are collectively weak, \emph{product concentration} can arise, in which case all firms choose identical common characteristics that are perfectly aligned with the consumer's ideal.
In this equilibrium, firms' products are highly similar, though not perfectly homogeneous due to the presence of idiosyncratic characteristics.
Building on an intuition analogous to a standard Cournot model with homogeneous goods, the welfare ranking relative to monopoly can reverse in this scenario.
Specifically, Theorem~\ref{thm_welfare_C} shows if the distribution of standalone values is sufficiently balanced, then oligopoly yields higher welfare than monopoly.

Third, when there is sufficient asymmetry in firms' standalone values, \emph{product polarization} can arise, in which case firms split into two groups: one group aligns its common characteristics perfectly with the consumer's ideal, while the other group chooses the exact opposite direction in order to hedge against competition.
Multiple product-polarization equilibria can exist, differing in which firms choose to align positively or negatively with the consumer's ideal.

Across those multiple equilibria, Theorem~\ref{thm_welfare_DCP} shows that equilibria exhibiting product differentiation yield higher welfare than those exhibiting product concentration or ``inefficient'' forms of product polarization, where firms with relatively low standalone values nevertheless dominate the provision of common characteristics.
Therefore, even though oligopoly typically performs worse than monopoly in our model, product differentiation emerges as the most desirable equilibrium configuration within oligopoly.

We extend the analysis to incorporate common ownership, under which firms are exposed to overlapping ownership structures and therefore internalize rivals' profits in their strategic decisions.
We formulate this setting by adopting the modeling strategy of \cite{ederer2025pellegrino}, which nests monopoly and oligopoly as two polar cases.
The rise of common ownership has been prominently observed over recent decades in the United States \citep{backusetal2021}, as well as in other countries \citep{antonetal2025}, and it has attracted increasing attention regarding its anticompetitive effects; see, e.g., the survey by \cite{schmalz2018}.
The conventional policy concern is that common ownership softens firms' incentives to compete aggressively, for example by reducing output or raising prices, to the detriment of welfare.

Our analysis sheds new light on the effects of common ownership through the channel of product differentiation.
Specifically, unlike existing work \citep{rotemberg1984}, firms in our model have an additional channel through which they can soften competition, namely by differentiating their products more from those of other firms.
A higher degree of product differentiation then allows firms to sustain higher markups, which in turn strengthens their incentives to produce.
We show that, in a simplified symmetric setting, equilibrium output levels increase with the strength of overlapping ownership links, indicating that this differentiation-induced effect outweighs the direct output-reducing effect stemming from profit internalization.
Moreover, we show that social welfare is also increasing in the degree of common ownership, up to mild qualifications, implying that common ownership can contribute to welfare improvements once endogenous product differentiation is taken into account.

\paragraph{Related Literature.}
The economics of product differentiation has a long history and has been studied in a wide range of models; see, e.g., \cite{tirole1988} or the survey by \cite{lancaster1990}.
One of the most widely used frameworks is the location model \`a la \cite{hotelling1929}, among which the circular model \`a la \cite{salop1979} and \cite{vickrey1999} bears particular similarity to our model.
Compared to the location model, we do not restrict consumers to unit demand, resulting in a smooth demand system without kinks or discontinuities.
This substantially improves tractability, enabling us to derive equilibria based on a standard first-order approach, and our characterization holds with full generality for arbitrary parameter values, without imposing ad hoc assumptions such as symmetry across firms.
Moreover, product characteristics map directly into the Slutsky matrix of consumer demand, providing an economically meaningful interpretation of firms' strategic variables.

The GHL model is proposed by \cite{pellegrino2025} and has generated fruitful insights by exploiting its close theoretical connection to network games and its empirical advantage in fitting the model to rich firm-level datasets, such as those developed by \cite{hoberg2016phillips}.
A growing literature builds on the GHL framework to study a range of economic implications in differentiated product markets, for example, in the contexts of common ownership \citep{ederer2025pellegrino}, production networks \citep{bizzarri2024}, and dynamic R\&D investment \citep{hopenhayn2025okumura, okumura2025}.

In this strand, one of the most relevant papers is the recent independent work by \cite{voelkening2025}, who studies a symmetric Cournot duopoly with endogenous product design, whereas there are two notable differences.
First, while \cite{voelkening2025} allows firms to choose product characteristics prior to output, we assume that these decisions are made simultaneously.\footnote{The sequential assumption of \cite{voelkening2025} may be more realistic and better aligned with actual firm behavior; for example, \cite{tirole1988} distinguishes between different adjustment horizons for firms' decision variables by viewing output or prices as short-run choices and product characteristics as somewhat long-run decisions. However, we adopt the simultaneous assumption because it offers technical advantages in obtaining sharp equilibrium and welfare characterizations in our general setting.}
Second, while the equilibrium in \cite{voelkening2025} is unique and exhibits product concentration, our model admits multiple equilibria that exhibit different patterns of characteristics design.
The driving force behind this difference is the presence of idiosyncratic characteristics, which are absent in \cite{voelkening2025}; indeed, when firms' standalone values are collectively weak, product concentration is also the unique equilibrium in our model.
Moreover, even in equilibria that exhibit product differentiation, Corollary~\ref{cor_cosine_diff} shows that the degree of equilibrium product differentiation is lower than the socially optimal level.
Hence, the two papers offer complementary insights by highlighting the robustness of agglomeration-type outcomes in different settings.

From a theoretical perspective, this paper is related to the literature on network games \citep{bcz2006, bramoulleetal2014}.
A key departure from the standard linear-quadratic framework is that each player's action in our model is multi-dimensional, which turns out to be central to the coexistence of multiple equilibria.
In the scalar-action case, players' best responses are given by first-order conditions, and equilibrium is characterized as a solution to the resulting linear system.
As shown in \cite{bcz2006}, generic invertibility of the relevant matrices then implies that almost every network game admits a unique equilibrium.
In our setting, best responses are likewise characterized by a system of linear equations, but these take the form of orthogonality conditions across ``vectors'' of actions.
Unlike the scalar-action case, such conditions can be satisfied by multiple action profiles, giving rise to multiple patterns of equilibrium product differentiation.
Beyond this equilibrium analysis with multi-dimensional actions, we also contribute to the literature by offering a systematic welfare analysis, leveraging spectral methods in a manner related to \cite{bko2015} and \cite{galeottietal2020, galeottietal2025}.

\paragraph{Organization.}
The paper proceeds as follows.
Section~\ref{sec_model} introduces the model.
Section~\ref{sec_sp} studies the social planner's and monopolist's problems.
Section~\ref{sec_eq} characterizes all equilibria under oligopoly.
Section~\ref{sec_welfare} derives welfare implications through comparisons across market structures and across equilibria.
Section~\ref{sec_ext} extends the analysis by incorporating network effects and common ownership.
Omitted proofs and additional results are presented in the appendix.


\section{Model}
\label{sec_model}

Our model builds on the GHL model of \cite{pellegrino2025}.
The key departure is that product characteristics are endogenously determined, with each firm choosing both its output level and the characteristics of its product.

The following notation is used throughout the paper.
Let $[n] \coloneqq \{1,\ldots,n\}$ for any integer $n \in \N$.
Any vector $\bm{x} \in \R^n$ is treated as a column vector unless otherwise stated.
For $1 \le p < \infty$, let $\|\bm{x}\|_p \coloneqq (\sum_{i=1}^n \qty|x_i|^p)^{1/p}$ denote the $\ell_p$ norm.
Let $\|\bm{x}\|_{\infty} \coloneqq \max_{i \in [n]} \qty|x_i|$ denote the max norm.
For $\bm{x},\bm{y} \in \R^n$, we write $\bm{x} \ge \bm{y}$ (resp.~$\bm{x} \gg \bm{y}$) when $x_i \ge y_i$ (resp.~$x_i > y_i$) for all $i \in [n]$.

\subsection{Generalized Hedonic-Linear Demand}
There are $n$ firms, indexed by $i \in [n]$.
Each firm produces a single differentiated product.
There are $n+m$ characteristics.
The first $n$ characteristics are idiosyncratic and firm-specific, while the remaining $m$ characteristics are common across all products.
Each product is therefore described by a vector $(\bm{e}_i,\bm{a}_i)$, where $\bm{e}_i \in \R^n$ is the unit vector pointing in the $i$-th direction, and $\bm{a}_i = [a_{1i},\ldots,a_{mi}]^\top \in \R^m$ is a vector of unit $\ell_2$ norm, i.e., $\|\bm{a}_i\|_2 = (\bm{a}_i^\top \bm{a}_i)^{1/2} = 1$.
A \emph{characteristics profile} refers to an $m$-by-$n$ matrix $\bm{A}$ obtained by listing all firms' common-characteristics vectors,
\[
\bm{A} \coloneqq \mqty[\bm{a}_1 & \cdots & \bm{a}_i & \cdots & \bm{a}_n] = \mqty[a_{11} & & a_{1i} & & a_{1n} \\
\vdots & \cdots & \vdots & \cdots & \vdots \\
a_{m1} & & a_{mi} & & a_{mn}].
\]
Let $\calA$ be the set of all $m$-by-$n$ matrices such that the norm of each column is one.
Denote by $q_i$ the number of units produced by firm $i$.
An \emph{output profile} refers to a vector $\bm{q} = [q_1,\ldots,q_n]^\top$ that specifies the quantity produced by each firm.

There is a representative agent who serves as a consumer, a worker, and an owner.
Following the hedonic demand literature \citep{lancaster1966}, the agent linearly aggregates the characteristics of different products, and preferences are defined over the aggregated characteristics.
The total units of common characteristic $k$ are given by $x_k = \sum_{i=1}^n a_{ki} q_i$.
In matrix form, the aggregate vector of common characteristics is given by
\[
\bm{x} = \bm{A}\bm{q}.
\]
Also, the aggregate vector of idiosyncratic characteristics is given by
\[
\bm{y} = \bm{q}.
\]

The agent has a utility function that is quadratic in both common characteristics $\bm{x}$ and idiosyncratic characteristics $\bm{y}$, and linear in disutility from labor supply $H$:
\[
U(\bm{x}, \bm{y}, H)
= \alpha \qty(\bm{x}^\top \bm{\beta} - \frac{1}{2}\bm{x}^\top \bm{x})
+ \qty(\bm{y}^\top \bm{b} - \frac{1}{2}\bm{y}^\top \bm{y})
- H,
\]
where $\bm{\beta}$ and $\bm{b}$ are given vectors, and $\alpha>0$ is the utility weight on common characteristics.
Note that this utility function features satiation with respect to $\bm{x}$ and $\bm{y}$, as the utility from common and idiosyncratic characteristics is maximized at $\bm{\beta}$ and $\bm{b}$, respectively.

The agent chooses a consumption bundle $\bm{q}$ taking the characteristics profile $\bm{A}$ and the price vector $\bm{p} = [p_1, \ldots, p_n]^\top$ as given.
Their income is given by the sum of labor income $H$ and ownership dividend $\Pi$, where the wage per unit of $H$ is normalized to be one. The budget constraint is therefore
\[
\bm{q}^\top \bm{p} \le H + \Pi.
\]
Plugging $\bm{x} = \bm{A}\bm{q}$ and $\bm{y} = \bm{q}$ into $U(\bm{x}, \bm{y}, H)$, we obtain the following Lagrangian for the agent's utility maximization problem:
\begin{align*}
\calL(\bm{q}, H) &= 
\alpha \qty(\bm{q}^\top \bm{A}^\top \bm{\beta} - \frac{1}{2} \bm{q}^\top \bm{A}^\top \bm{A}\bm{q}) + \qty(\bm{q}^\top \bm{b} - \frac{1}{2} \bm{q}^\top \bm{q}) - H + \lambda \qty(H+\Pi - \bm{q}^\top \bm{p}) \\
&= \bm{q}^\top \underbrace{\qty(\alpha \bm{A}^\top \bm{\beta} + \bm{b})}_{\eqqcolon \; \bm{b}_{\bm{A}}} - \frac{1}{2} \bm{q}^\top \underbrace{\qty(\alpha \bm{A}^\top \bm{A} + \bm{I})}_{\eqqcolon \; \bm{\Sigma}_{\bm{A}}} \bm{q} - \lambda \bm{q}^\top \bm{p} + (\lambda-1)H + \lambda \Pi.
\end{align*}
The wage normalization immediately pins down the Lagrange multiplier $\lambda = 1$.
Thus, the first-order condition with respect to $\bm{q}$ yields the following demand and inverse demand functions:
\begin{align*}
\text{Aggregate demand:} \quad &\bm{q} = \bm{\Sigma}^{-1}_{\bm{A}} \qty(\bm{b}_{\bm{A}} - \bm{p}), \\
\text{Inverse demand:} \quad &\bm{p} = \bm{b}_{\bm{A}} - \bm{\Sigma}_{\bm{A}}\bm{q}.
\end{align*}
Also, the surplus is given by
\[
\bm{q}^\top \bm{b}_{\bm{A}} - \frac{1}{2} \bm{q}^\top \bm{\Sigma}_{\bm{A}}\bm{q} - H.
\]
Note that the negative inverse of $\bm{\Sigma}_{\bm{A}}$ represents the Slutsky matrix, since $\6 \bm{q}/\6 \bm{p} = -\bm{\Sigma}_{\bm{A}}^{-1}$, which is negative definite by construction.

We remark that, when the characteristics profile $\bm{A}$ is taken as given, the GHL demand system is mathematically isomorphic to a familiar linear demand system (LDS) for differentiated products, which can be micro-founded by a quasilinear quadratic utility model (QQUM) that closely resembles our representative consumer's utility.\footnote{\cite{chone2020linnemer} provides an excellent survey of the history and basic properties of LDS and QQUM, with the modern formulations commonly used in industrial organization attributed to \cite{spence1976} and \cite{shubik1980levitan}. To draw a close connection to this strand of the literature, one can, for example, establish a one-to-one correspondence between our demand system and that in \cite{amiretal2016}}
However, the key difference is that while the demand intercept $\bm{b}_{\bm A}$ and the substitution matrix $\bm{\Sigma}_{\bm A}$ are treated as exogenous in those works, they are endogenously shaped through the design of $\bm{A}$ in our model, which creates a nontrivial interplay between firms' strategic choices over product design and the quantities they supply.

One technical yet important difference from the original GHL model is that, while \cite{pellegrino2025} restricts each characteristics vector $\bm{a}_i$ to lie in the positive orthant, we allow $\bm{a}_i$ to be any unit vector, possibly with negative coordinates.
As a result, the cosine similarity $\bm{a}_i^\top \bm{a}_j$ can take any value between $-1$ and $1$, allowing us to capture not only substitutability but also complementarity across firms' products.\footnote{As noted in Section~V.I of \cite{pellegrino2025}, restricting $\bm{\Sigma}_{\bm A}$ to be nonnegative implies that products are ``physical substitutes'' in utility, but they need not be ``strategic substitutes'' since the associated Slutsky matrix $-\bm{\Sigma}_{\bm A}^{-1}$ can still contain positive entries.}
We adopt this more general choice set for firms because the focus of this paper is theoretical, and thus we do not face data availability constraints in estimating the model.\footnote{\cite{pellegrino2025} estimates $\bm{a}_i^\top \bm{a}_j$ using text-based data on product similarity developed by \cite{hoberg2016phillips}, which by construction yield nonnegative values.}
More substantially, this modeling choice is necessary for the tractability of our analysis; without it, a clean characterization of aggregate characteristics, as in Lemma~\ref{lem_donut}, would no longer be available.


\subsection{Labor Demand and Product Supply}

Each firm's production technology is represented by a linear cost function $h_i(q_i)=c_i q_i$, where $c_i>0$ is a constant marginal cost and $h_i$ is the labor input used by firm $i$.
The labor market clearing condition requires
\[
H=\sum_{i=1}^n h_i(q_i)=\bm{q}^\top \bm{c}.
\] 
Firm $i$'s profit is defined as
\begin{equation} \label{eq_profit}
\Pi_i(\bm{A},\bm{q})
= \qty(p_i(\bm{A},\bm{q})-c_i)\cdot q_i,
\end{equation}
where $p_i(\bm{A},\bm{q})$ is the $i$-th coordinate of the market-clearing price vector $\bm{b}_{\bm{A}}-\bm{\Sigma}_{\bm{A}}\bm{q}$, and $p_i(\bm{A},\bm{q})-c_i$ represents firm $i$'s markup.

Recalling that $\bm{a}_i^\top \bm{a}_i=1$, we can explicitly express firm $i$'s markup as
\begin{equation} \label{eq_markup}
p_i(\bm{A},\bm{q})-c_i
= \alpha \bm{a}_i^\top \qty(\bm{\beta}-\sum_{j\neq i} q_j\bm{a}_j) - (1+\alpha) q_i + \gamma_i,
\end{equation}
where
\[
\gamma_i \coloneqq b_i-c_i.
\]
Thus, firm $i$'s markup depends on both its common characteristics vector  $\bm{a}_i$ and its output level $q_i$.
The first term in \eqref{eq_markup} captures the dependence of markup on product characteristics.
Specifically, as firm $i$ chooses $\bm{a}_i$ closer to the consumer's ideal vector $\bm{\beta}$, its markup increases, whereas greater similarity between $\bm{a}_i$ and opponents' characteristics $\bm{a}_j$ reduces markup.
The second term captures the negative dependency of firm $i$'s markup on its own output $q_i$, reflecting the standard nature of a downward sloping demand curve.

The final term $\gamma_i$, which we refer to as firm $i$'s \emph{standalone value}, is constant in the firm's choices.
It represents the value of its product's idiosyncratic characteristic for the representative consumer, net of marginal cost.
Economically, $\gamma_i$ can be interpreted as a baseline or ``brand'' value of firm $i$'s product that is orthogonal to common characteristics and cannot be replicated by competitors.
Let $\bm{\gamma}=[\gamma_1,\ldots,\gamma_n]^\top$ denote the vector collecting these standalone values.


\subsection{Parametric Assumptions}

The model is parameterized by $(\alpha,\bm{\beta},\bm{\gamma})$, which will play a decisive role in shaping both socially optimal and equilibrium allocations.
Throughout, we assume $\alpha > 0$ and normalize $\bm{\beta}$ so that
\[
\|\bm{\beta}\|_2=1.
\]
This normalization is innocuously obtained by an appropriate scaling of the agent's utility function.
In addition, we assume that
\[
\bm{\gamma}\gg \bm{0}.
\]
This assumption ensures that every firm produces a strictly positive quantity both at the social planner's optimum and in equilibrium.
We remark that it may fail when a firm's marginal cost is sufficiently high relative to the value of its idiosyncratic characteristic.


\section{Social Planner's and Monopolist's Benchmarks}
\label{sec_sp}

Plugging total labor supply $H$ into the representative agent's utility function yields social surplus.
In addition, aggregate profit is obtained by summing $\Pi_i$ across all firms.
These are explicitly given as follows:
\begin{alignat*}{2}
\text{Total surplus:} \quad \Omega(\bm{A},\bm{q}) &\coloneqq \bm{q}^\top \qty(\bm{b}_{\bm{A}}-\bm{c}) - \frac{1}{2} \bm{q}^\top \bm{\Sigma}_{\bm{A}}\bm{q}, \\
\text{Aggregate profit:} \quad \Pi(\bm{A},\bm{q}) &\coloneqq \bm{q}^\top \qty(\bm{b}_{\bm{A}}-\bm{c}) - \bm{q}^\top \bm{\Sigma}_{\bm{A}}\bm{q}.
\end{alignat*}
This section studies two benchmark cases; namely, a benevolent social planner's optimization of total surplus $\Omega$ and a monopolistic firm's optimization of aggregate profit $\Pi$.

If a characteristic profile $\bm{A}$---and thus $\bm{b}_{\bm{A}}$ and $\bm{\Sigma}_{\bm{A}}$---is fixed, maximizing each objective is a simple quadratic optimization problem in $\bm{q}$.
What if $\bm{A}$ is also controllable?
We address this joint choice of $\bm{A}$ and $\bm{q}$ in two steps.
First, we take any fixed production profile $\bm{q}$ and characterize the optimal $\bm{A}$ given $\bm{q}$.
We then determine the optimal $\bm{q}$ given the matrix $\bm{A}$ identified in the first step.

\subsection{Optimal Characteristics Design}

We begin by rewriting the social planner's objective as follows:
\[
\Omega(\bm{A},\bm{q})
= \alpha \underbrace{\qty((\bm{A}\bm{q})^\top \bm{\beta} - \frac{1}{2} (\bm{A}\bm{q})^\top \bm{A}\bm{q})}_{\eqqcolon \; \omega(\bm{A}\bm{q})} \;+\; \bm{q}^\top \bm{\gamma} - \frac{\bm{q}^\top \bm{q}}{2}.
\]
Similarly, we rewrite the monopolist's objective as follows:
\[
\Pi(\bm{A},\bm{q})
= \alpha \underbrace{\qty((\bm{A}\bm{q})^\top \bm{\beta} - (\bm{A}\bm{q})^\top \bm{A}\bm{q})}_{\eqqcolon \; \pi(\bm{A}\bm{q})} \;+\; \bm{q}^\top \bm{\gamma} - \bm{q}^\top \bm{q}.
\]
These expressions suggest that $\Omega$ and $\Pi$ depend on the choice of $\bm{A}$ only through $\omega(\bm{A}\bm{q})$ and $\pi(\bm{A}\bm{q})$, respectively, which are functions of the aggregate characteristics vector $\bm{x} = \bm{A}\bm{q}$.
This observation raises the question of which vectors $\bm{x}$ are attainable through an appropriate choice of characteristic vectors $\bm{a}_1,\ldots,\bm{a}_n$, each of which is required to have unit length by definition.
The next lemma provides an answer to this question.


\begin{figure}[t!]
\centering
\begin{tikzpicture}[>=stealth]
\usetikzlibrary{calc}
\def\Rin{0.5}
\def\Rout{2}
\def\AxisLength{2.5} 

\coordinate (O) at (0, 0);

\draw[fill=brown!25, draw=none, even odd rule]
    (0, 0) circle (\Rout)
    (0, 0) circle (\Rin);

\draw[red] (0, 0) circle (\Rin);
\draw[blue] (0, 0) circle (\Rout);

\draw[->, thick] (-\AxisLength, 0) -- (\AxisLength, 0);
\draw[->, thick] (0, -\AxisLength) -- (0, \AxisLength);

\node[red, anchor=south west, xshift=5pt, yshift=-5pt] at (0, \Rin) {$r(\bm{q})$};
\node[blue, anchor=north west, xshift=27pt, yshift=8pt] at (0, \Rout) {$R(\bm{q})$};

\node[anchor=north] at ({-0.7*(\Rin+\Rout)/2}, {-0.5*(\Rin+\Rout)/2}) {$\calD(\bm{q})$};

\end{tikzpicture}
\caption{Given $\bm{q} \ge \bm{0}$, the region $\calD(\bm{q})$ specifies the set of vectors $\bm{x} = \bm{A}\bm{q}$, where the outer radius $R(\bm{q}) = \sum_{i=1}^n q_i$ is the sum of production levels, and the inner radius $r(\bm{q}) = \max_{i \in [n]} q_i - \sum_{j\neq i} q_j$ is the difference between the highest $q_i$ and the sum of the rest.}
\label{fig_donut1}
\end{figure}
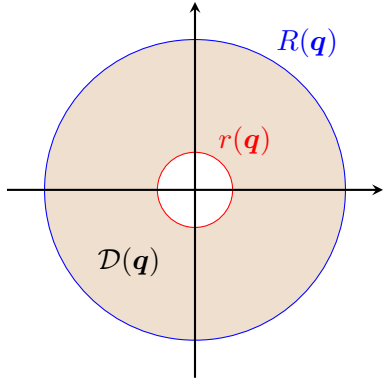


\begin{lemma} \label{lem_donut}
Given any $\bm{q} \ge \bm{0}$, there exists $\bm{A} \in \calA$ such that $\bm{x} = \bm{A}\bm{q}$ if and only if $\bm{x}$ satisfies
\begin{equation} \label{feasible_x}
r(\bm{q}) \le \|\bm{x}\|_2 \le R(\bm{q}), \quad \text{where} \quad r(\bm{q}) \coloneqq 2 \|\bm{q}\|_\infty - \|\bm{q}\|_1, \quad R(\bm{q}) \coloneqq \|\bm{q}\|_1.
\end{equation}
\end{lemma}

This lemma shows that any feasible $\bm{x}$ must lie in a ``donut'' region,
\[
\calD(\bm q) \coloneqq \qty{\bm{x} \in \R^m: r(\bm{q}) \le \|\bm{x}\|_2 \le R(\bm{q})},
\]
i.e., the area between the inner circle of radius $r(\bm{q})$ and the outer circle of radius $R(\bm{q})$.\footnote{Note that $r(\bm{q})$ can be negative so that the lower bound condition in \eqref{feasible_x} is trivially satisfied. In such cases, $\calD(\bm{q})$ reduces to a ball of radius $R(\bm{q})$.}
To gain a geometric intuition behind the condition, consider any feasible $\bm{x}$ such that
\[
q_1\bm{a}_1 + \cdots + q_n\bm{a}_n = \bm{x}.
\]
This equation describes a closed polygon with $n+1$ edges in $\R^m$, whose respective lengths are given by $q_1,\ldots,q_n$, and $\|\bm{x}\|_2$.
For such a polygon to exist, no edge may exceed the sum of the remaining edges.
In particular, the edge corresponding to $\bm{x}$ should not be longer than the sum of the edges $q_1\bm{a}_1,\ldots,q_n\bm{a}_n$, producing the upper bound in \eqref{feasible_x}.  
In addition, each edge $q_i\bm{a}_i$ should not be longer than the sum of all other edges, so
\[
q_i = \|q_i\bm{a}_i\|_2 \le \sum_{j \neq i} \|q_j\bm{a}_j\|_2 + \|\bm{x}\|_2 = \sum_{j \neq i} q_j + \|\bm{x}\|_2, \quad \forall i \in [n],
\]
whose tightest restriction arises when $q_i = \|\bm{q}\|_{\infty}$ and gives the lower bound in \eqref{feasible_x}.

Now, observe that $\omega(\bm{x})$ is globally maximized at the interior solution $\bm{x} = \bm{\beta}$, which is feasible precisely when $\bm{\beta} \in \calD(\bm{q})$.
Otherwise, if $\bm{\beta} \notin \calD(\bm{q})$ for a given $\bm{q}$, then $\bm{x}$ must instead be chosen as the projection of $\bm{\beta}$ onto $\calD(\bm{q})$.
Specifically, since $\|\bm{\beta}\|_2 = 1$ by normalization, if $R(\bm{q}) < 1$---i.e., when total production $\sum_{i=1}^n q_i$ is insufficient to reach $1$---the optimal $\bm{x}$ is given by the vector parallel to $\bm{\beta}$ whose length is truncated to $R(\bm{q})$.
Conversely, if $r(\bm{q}) > 1$---which occurs when a single dominant firm $i$ has $q_i$ far exceeding all others---$\bm{x}$ is again chosen parallel to $\bm{\beta}$, but with its length extended to $r(\bm{q})$.
A parallel argument applies to the monopolist's conditionally optimal design of $\bm{A}$ given $\bm{q}$, with $\bm{\beta}$ replaced by $\bm{\beta}/2$, since $\pi(\bm{x})$ is uniquely maximized at $\bm{x} = \bm{\beta}/2$.
We summarize the discussion so far in the next lemma.


\begin{figure}[t]
\centering
\begin{tabular}{ccc}


\begin{minipage}[t]{0.32\textwidth}
\begin{tikzpicture}[>=stealth]
\def\xmax{4}
\def\ymax{4}
\def\Rin{0.7}
\def\Rout{2.3}
\def\startangle{-5}
\def\endangle{95}
\def\thetaone{45}
\def\thetab{45}
\def\DeltaY{5pt}

\pgfmathsetmacro{\qone}{(\Rout+\Rin)/2}
\pgfmathsetmacro{\qtwo}{(\Rout-\Rin)/2}
\pgfmathsetmacro{\costh}{cos(\thetaone-\thetab)}
\pgfmathsetmacro{\disc}{(2*\qone*\costh)^2 - 4*(\qone*\qone-\qtwo*\qtwo)}
\pgfmathsetmacro{\tplus}{(\qone*\costh + 0.5*sqrt(\disc))}
\pgfmathsetmacro{\tminus}{(\qone*\costh - 0.5*sqrt(\disc))}
\pgfmathsetmacro{\t}{max(\tplus,\tminus)}

\coordinate (O) at (0,0);
\coordinate (P) at ({\qone*cos(\thetaone)},{\qone*sin(\thetaone)});
\coordinate (X) at ({\t*cos(\thetab)},{\t*sin(\thetab)});

\coordinate (Oshift) at ($(O)+(0,\DeltaY)$);
\coordinate (Pshift) at ($(P)+(0,\DeltaY)$);
\coordinate (XshiftL) at ($(X)+(-2.5pt,\DeltaY-2.5pt)$);
\coordinate (XshiftR) at ($(X)+(-2pt,\DeltaY-2pt)$);

\coordinate (Bend) at ({3.3*cos(\thetab)},{3.3*sin(\thetab)});
\coordinate (Bshift) at ($(Bend)+(-3pt,\DeltaY-3pt)$);

\draw[fill=brown!25,draw=none]
  ({\Rin*cos(\startangle)},{\Rin*sin(\startangle)}) --
  ({\Rout*cos(\startangle)},{\Rout*sin(\startangle)})
  arc (\startangle:\endangle:\Rout) --
  ({\Rin*cos(\endangle)},{\Rin*sin(\endangle)})
  arc (\endangle:\startangle:\Rin);

\draw[->,thick] (0,0) -- (\xmax,0);
\draw[->,thick] (0,0) -- (0,\ymax);

\node[anchor=east,xshift=-5pt] at (0,\Rin) {$r(\bm{q})$};
\node[anchor=east,xshift=-5pt] at (0,\Rout) {$R(\bm{q})$};

\draw[->,very thick,red] (Oshift) -- (Pshift)
  node[left, xshift=-2pt] {$q_1\bm{a}_1$};

\draw[->,very thick,blue] (Pshift) -- (XshiftL)
  node[above, xshift=-12pt, yshift=-3pt] {$q_2\bm{a}_2$};

\draw[<-,dotted,very thick,black] (XshiftR) -- (Bshift);

\draw[->,very thick,black] (O) -- (Bend)
  node[below right] {$\bm{\beta}$};

\fill (X) circle (2.5pt)
  node[anchor=north west,xshift=-6pt,yshift=-2pt] {$\bm{x}^\dagger$};
\end{tikzpicture}
\end{minipage}%
\hfill

\begin{minipage}[t]{0.32\textwidth}
\begin{tikzpicture}[>=stealth]
\def\xmax{4}
\def\ymax{4}
\def\Rin{0.9}
\def\Rout{2.9}
\def\startangle{-5}
\def\endangle{95}
\def\thetaone{63}
\def\thetab{45}

\pgfmathsetmacro{\qone}{(\Rout+\Rin)/2}
\pgfmathsetmacro{\qtwo}{(\Rout-\Rin)/2}
\pgfmathsetmacro{\costh}{cos(\thetaone-\thetab)}
\pgfmathsetmacro{\disc}{(2*\qone*\costh)^2 - 4*(\qone*\qone-\qtwo*\qtwo)}
\pgfmathsetmacro{\tplus}{(\qone*\costh + 0.5*sqrt(\disc))}
\pgfmathsetmacro{\tminus}{(\qone*\costh - 0.5*sqrt(\disc))}
\pgfmathsetmacro{\t}{max(\tplus,\tminus)}

\coordinate (O) at (0,0);
\coordinate (P) at ({\qone*cos(\thetaone)},{\qone*sin(\thetaone)});
\coordinate (X) at ({\t*cos(\thetab)},{\t*sin(\thetab)});
\coordinate (Oshift) at ($(O)+(0,1pt)$);
\coordinate (Pshift) at ($(P)+(0,1.5pt)$);
\coordinate (Xshift) at ($(X)+(-2.5pt,1pt)$);
\coordinate (XshiftB) at ($(X)+(-1.5pt,-1.5pt)$);

\draw[fill=brown!25,draw=none]
  ({\Rin*cos(\startangle)},{\Rin*sin(\startangle)}) --
  ({\Rout*cos(\startangle)},{\Rout*sin(\startangle)})
  arc (\startangle:\endangle:\Rout) --
  ({\Rin*cos(\endangle)},{\Rin*sin(\endangle)})
  arc (\endangle:\startangle:\Rin);

\draw[->,thick] (0,0) -- (\xmax,0);
\draw[->,thick] (0,0) -- (0,\ymax);

\node[anchor=east,xshift=-5pt] at (0,\Rin) {$r(\bm{q})$};
\node[anchor=east,xshift=-5pt] at (0,\Rout) {$R(\bm{q})$};

\draw[->,very thick,black] (O) -- (XshiftB);

\draw[->,very thick,red] (Oshift) -- (P)
  node[above, xshift=-12pt, yshift=-4pt] {$q_1\bm{a}_1$};

\draw[->,very thick,blue] (Pshift) -- (Xshift)
  node[above, xshift=-10pt, yshift=0pt] {$q_2\bm{a}_2$};

\fill (X) circle (2.5pt)
  node[anchor=north west,xshift=-6pt,yshift=-2pt] {$\bm{x}^\dagger=\bm{\beta}$};
\end{tikzpicture}
\end{minipage}%
\hfill

\begin{minipage}[t]{0.32\textwidth}
\begin{tikzpicture}[>=stealth]
\def\xmax{4}
\def\ymax{4}
\def\Rin{2.6}
\def\Rout{3.6}
\def\startangle{-5}
\def\endangle{95}
\def\thetaone{45}
\def\thetab{45}
\def\DeltaY{5pt}

\pgfmathsetmacro{\qone}{(\Rout+\Rin)/2}
\pgfmathsetmacro{\qtwo}{(\Rout-\Rin)/2}
\pgfmathsetmacro{\costh}{cos(\thetaone-\thetab)}
\pgfmathsetmacro{\disc}{(2*\qone*\costh)^2 - 4*(\qone*\qone-\qtwo*\qtwo)}
\pgfmathsetmacro{\tplus}{(\qone*\costh + 0.5*sqrt(\disc))}
\pgfmathsetmacro{\tminus}{(\qone*\costh - 0.5*sqrt(\disc))}
\pgfmathsetmacro{\t}{min(\tplus,\tminus)}

\coordinate (O) at (0,0);
\coordinate (P) at ({\qone*cos(\thetaone)},{\qone*sin(\thetaone)});
\coordinate (X) at ({\t*cos(\thetab)},{\t*sin(\thetab)});

\coordinate (Oshift) at ($(O)+(0,\DeltaY)$);
\coordinate (Pshift) at ($(P)+(0,\DeltaY)$);
\coordinate (PshiftB) at ($(P)+(2pt,2pt)$);
\coordinate (XshiftR) at ($(X)+(2pt,2pt)$);
\coordinate (XshiftL) at ($(X)-(2pt,2pt)$);

\draw[fill=brown!25,draw=none]
  ({\Rin*cos(\startangle)},{\Rin*sin(\startangle)}) --
  ({\Rout*cos(\startangle)},{\Rout*sin(\startangle)})
  arc (\startangle:\endangle:\Rout) --
  ({\Rin*cos(\endangle)},{\Rin*sin(\endangle)})
  arc (\endangle:\startangle:\Rin);

\draw[->,thick] (0,0) -- (\xmax,0);
\draw[->,thick] (0,0) -- (0,\ymax);

\node[anchor=east,xshift=-5pt] at (0,\Rin) {$r(\bm{q})$};
\node[anchor=east,xshift=-5pt] at (0,\Rout) {$R(\bm{q})$};

\draw[->,very thick,red] (Oshift) -- (Pshift)
  node[midway,left,xshift=-2pt] {$q_1\bm{a}_1$};

\draw[->,very thick,blue] (PshiftB) -- (XshiftR)
  node[midway,right,xshift=3pt] {$q_2\bm{a}_2$};

\draw[->,dotted,very thick,black] ({1.8*cos(\thetab)},{1.8*sin(\thetab)}) -- (XshiftL);

\draw[->,very thick,black] (O) -- ({1.8*cos(\thetab)},{1.8*sin(\thetab)})
  node[midway,below right] {$\bm{\beta}$};

\fill (X) circle (2.5pt)
  node[anchor=north west,xshift=-6pt,yshift=-2pt] {$\bm{x}^\dagger$};
\end{tikzpicture}
\end{minipage}

\end{tabular}
\caption{
Left panel: When $R(\bm{q}) < \|\bm{\beta}\|_2$, $\bm{\beta}$ is truncated to length $R(\bm{q})$, and product concentration arises.
Middle panel: When $r(\bm{q}) < \|\bm{\beta}\|_2 < R(\bm{q})$, the interior solution $\bm{x}^\dagger = \bm{\beta}$ obtains, and product differentiation arises.
Right panel: When $\|\bm{\beta}\|_2 < r(\bm{q})$, $\bm{\beta}$ is extended to length $r(\bm{q})$, and product polarization arises.
}
\label{fig_donut2}
\end{figure}
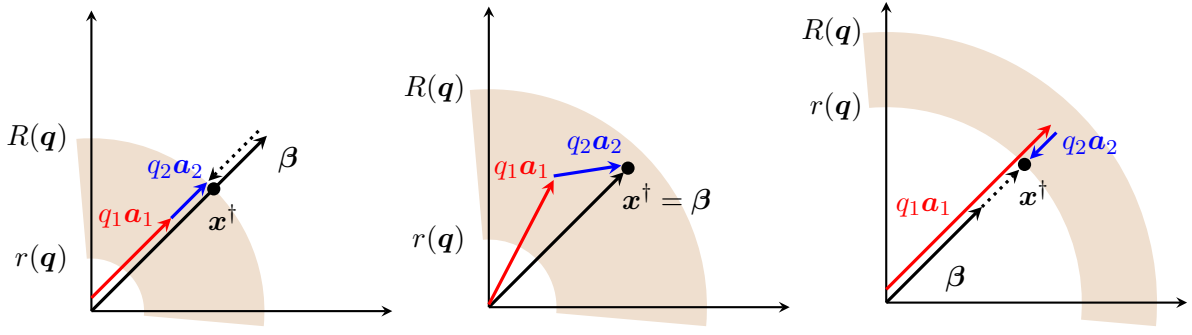


\begin{lemma} \label{lem_opt_A}
Given any $\bm{q} \ge \bm{0}$, $\omega(\bm{x})$ and $\pi(\bm{x})$ are uniquely maximized over $\calD(\bm{q})$ at $\bm{x}^\dagger(\bm{q}) = \rho^\dagger(\bm{q}) \cdot \bm{\beta}$ and $\bm{x}^\ddagger(\bm{q}) = \rho^\ddagger(\bm{q}) \cdot \bm{\beta}$, respectively, where
\[
\rho^\dagger(\bm{q}) = 
\max \qty{r(\bm{q}),\, \min\qty{R(\bm{q}),\, 1}}, \qquad \rho^\ddagger(\bm{q}) = 
\max \qty{r(\bm{q}),\, \min\qty{R(\bm{q}),\, 1/2}}.
\]
\end{lemma}

Therefore, for both the monopolist and the social planner, the conditionally optimal characteristics profile for a given output profile $\bm{q}$ yields an aggregate characteristics vector that is proportional to the consumer's ideal vector $\bm{\beta}$, where the proportionality coefficient is determined by the norms of $\bm{q}$.
Notably, the monopolist's choice of product differentiation is always aligned with the social planner's optimum, even though the output profiles they select may differ.
This implies that any inefficiency under monopoly arises solely from distortions in production quantities rather than from misaligned product design.

We can further characterize each characteristics vector $\bm{a}_i$ when the aggregate vector $\bm{x} = \bm{A}\bm{q}$ lies on the boundary of $\calD(\bm{q})$, i.e., $\|\bm{x}\|_2$ is equal to either $R(\bm{q})$ or $r(\bm{q})$.
First, when $\|\bm{x}\|_2 = R(\bm{q})$, the left panel of Figure~\ref{fig_donut2} suggests that $\bm{x}$ can be generated by aligning all characteristics vectors $\bm{a}_i$ in the same direction as $\bm{\beta}$.
We say that a characteristics profile $\bm{A}$ exhibits \emph{product concentration} if $\bm{a}_i = \bm{a}_j$ for all $i,j \in [n]$.
The next lemma shows that product concentration is not only sufficient but also necessary to generate $\bm{x}$ on the outer boundary of $\calD(\bm{q})$.

\begin{lemma} \label{lem_R}
Given any $\bm{q} \gg \bm{0}$, let $\bm{x} \in \calD(\bm{q})$ be such that $\|\bm{x}\|_2 = R(\bm{q})$.
Then, $\bm{A}\bm{q}=\bm{x}$ holds if and only if $\bm{A}$ exhibits product concentration, with $\bm{a}_i = \frac{\bm{x}}{\|\bm{x}\|_2}$ for all $i \in [n]$.
\end{lemma}

Next, we turn to the opposite boundary case, where $\|\bm{x}\|_2 = r(\bm{q})$.
Note that this case can arise only when $r(\bm{q}) \ge 0$, which requires the existence of a unique firm $i$ whose output $q_i$ is excessively large relative to the others.
As illustrated in the right panel of Figure~\ref{fig_donut2}, the vector $\bm{x}$ can then be generated by orienting the characteristics vector of the largest firm in the direction of $\bm{x}$, while orienting those of all remaining firms in the opposite direction.
Generalizing this case, we say that a characteristics profile $\bm{A}$ exhibits \emph{product polarization} if $\bm{a}_i = \pm \bm{u}$ for all $i \in [n]$, for some unit vector $\bm{u}$, with both signs occurring.
The next lemma shows that product polarization is necessary and sufficient to generate $\bm{x}$ on the inner boundary of $\calD(\bm{q})$, and, in particular, that all non-dominant firms orient their characteristics vectors in the direction opposite to that of the largest firm.

\begin{lemma} \label{lem_r}
Given any $\bm{q} \gg \bm{0}$ with $r(\bm{q}) > 0$, let $\bm{x} \in \calD(\bm{q})$ be such that $\|\bm{x}\|_2 = r(\bm{q})$.
Then, $\bm{A}\bm{q}=\bm{x}$ holds if and only if $\bm{A}$ exhibits product differentiation, with $\bm{a}_i = \frac{\bm{x}}{\|\bm{x}\|_2}$ for the unique firm $i$ with the highest output $q_i$, and $\bm{a}_j = -\frac{\bm{x}}{\|\bm{x}\|_2}$ for all $j \neq i$.
\end{lemma}

These lemmas imply that the characteristics profile $\bm{A}$ is uniquely determined once we require the targeted vector $\bm{x}$ to have length equal to either boundary value $R(\bm{q})$ or $r(\bm{q})$.
In both cases, the resulting characteristics profile satisfies $\rank(\bm{A})=1$.
By contrast, we say that $\bm{A}$ exhibits \emph{product differentiation} when $\rank(\bm{A})\ge 2$,\footnote{As the construction in Lemma~\ref{lem_donut} shows, any target vector $\bm{x}\in\calD(\bm{q})$ can be generated in such a way that the vectors $q_1\bm{a}_1,\ldots,q_n\bm{a}_n$, and $\bm{x}$ together form a ``triangle'' in a space isomorphic to $\R^2$. Hence, both the social planner's and equilibrium characteristics profiles may be chosen with $\rank(\bm{A})\le 2$, although there can be other higher-rank matrices $\bm{A}$ that achieve the same aggregate vector $\bm{x}$.}
in which case there exist firms $i\neq j$ such that the cosine similarity between $\bm{a}_i$ and $\bm{a}_j$ lies strictly between the extreme values $-1$ and $1$.
The next lemma shows that product differentiation arises for generic targeted vector $\bm{x}$ in $\calD(\bm{q})$.

\begin{lemma} \label{lem_b}
Given any $\bm{q}\gg 0$, let $\bm{x}\in \calD(\bm{q})$ be such that there exists no $\bm{\sigma}\in\{-1,1\}^n$ satisfying $\bm{\sigma}^\top \bm{q}=\|\bm{x}\|_2$.
If $\bm{A}\bm{q}=\bm{x}$, then $\bm{A}$ exhibits product differentiation.
In particular, if $n=2$, $\bm{A}$ exhibits product differentiation whenever $|q_1-q_2|<\|\bm{x}\|_2<q_1+q_2$.
\end{lemma}

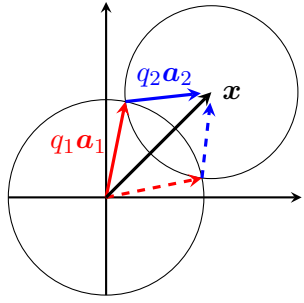
\begin{figure}[t]
\centering
\begin{tikzpicture}[>=stealth]
\usetikzlibrary{calc,intersections}

\def\qone{1.3}
\def\qtwo{1.15}
\def\AxisPos{2.6}
\def\AxisNeg{1.3}
\def\ShortenBlue{4pt}

\coordinate (O) at (0,0);
\coordinate (X) at (1.4,1.4);

\path[name path=Cone] (O) circle (\qone);
\path[name path=Ctwo] (X) circle (\qtwo);

\path[name intersections={of=Cone and Ctwo, by={P,Q}}];

\draw[->, thick] (-\AxisNeg,0) -- (\AxisPos,0);
\draw[->, thick] (0,-\AxisNeg) -- (0,\AxisPos);

\draw (O) circle (\qone);
\draw (X) circle (\qtwo);

\draw[->, very thick, red] (O) -- (P);
\draw[->, very thick, blue, shorten >=\ShortenBlue] (P) -- (X);

\draw[->, very thick, red, dashed] (O) -- (Q);
\draw[->, very thick, blue, dashed, shorten >=\ShortenBlue] (Q) -- (X);

\draw[->, very thick] (O) -- (X)
  node[right] {$\bm{x}$};

\node[red, anchor=east] at ($(O)!0.55!(P)$) {$q_1\bm{a}_1$};
\node[blue, anchor=south, xshift=-3pt] at ($(P)!0.55!(X)$) {$q_2\bm{a}_2$};

\end{tikzpicture}
\caption{Geometric illustration of multiplicity of $(\bm{a}_1,\bm{a}_2)$ when \eqref{feasible_x} is slack.}
\label{fig_remark1}
\end{figure}

\begin{remark} \label{remark_multiplicity}
When $\bm{x}$ belongs to the interior of $\calD(\bm{q})$, it can be generated by multiple characteristics profiles.
To illustrate, suppose that $n=m=2$.
As depicted in Figure~\ref{fig_remark1}, consider a circle of radius $q_1$ centered at the origin and another circle of radius $q_2$ centered at $\bm{x}$.
When \eqref{feasible_x} is slack, these two circles intersect at exactly two points, each of which corresponds to a pair of characteristics vectors $(\bm{a}_1,\bm{a}_2)$ satisfying $\bm{x}=q_1\bm{a}_1+q_2\bm{a}_2$.
More generally, when $m\ge 3$, there are infinitely many such pairs $(\bm{a}_1,\bm{a}_2)$ that generate $\bm{x}$; geometrically, they form a sphere embedded in a space isomorphic to $\R^{m-1}$.
Also, the degree of flexibility in $\bm{A}$ increases with $n$.
\end{remark}


\subsection{Optimal Output Levels}

Given any $\bm{q} \ge \bm{0}$, let $\bm{A}^\dagger(\bm{q})$ and $\bm{A}^\ddagger(\bm{q})$ be any characteristics profiles that maximize $\omega$ and $\pi$, respectively.
Incorporating $\bm{A}^\dagger(\bm{q})$ into $\Omega$, Lemma~\ref{lem_opt_A} implies that the social planner's problem reduces to the following optimization over $\bm{q}$ alone:
\[
\max_{\bm{q}\ge\bm{0}} \; \Omega(\bm{A}^\dagger(\bm{q}), \bm{q}) 
= \alpha \qty(\rho^\dagger(\bm{q}) - \frac{\rho^\dagger(\bm{q})^2}{2}) 
+ \bm{q}^\top \bm{\gamma} - \frac{\bm{q}^\top \bm{q}}{2}.
\]
Similarly, substituting $\bm{A}^\ddagger(\bm{q})$ into $\Pi$, the monopolist's problem reduces to
\[
\max_{\bm{q}\ge\bm{0}} \; \Pi(\bm{A}^\ddagger(\bm{q}), \bm{q})
= \alpha \qty(\rho^\ddagger(\bm{q}) - \rho^\ddagger(\bm{q})^2) 
+ \bm{q}^\top \bm{\gamma} - \bm{q}^\top \bm{q}.
\]

These problems share a close similarity, as they differ only in the coefficients on the quadratic terms.
Indeed, the next proposition shows that the same characteristics profile is optimal for both the social planner and the monopolist, whereas the monopolist's output is scaled down by one half.
As a result, welfare under monopoly amounts to $75\%$ of the social optimum.\footnote{Once the coincidence of $\bm{A}^{\dagger}$ and $\bm{A}^\ddagger$ is established, the remaining implications of Proposition~\ref{prop_sp_mono} come with no surprise given the result of \cite{amiretal2016}; see also Section~3.3 of \cite{chone2020linnemer}.}

\begin{proposition} \label{prop_sp_mono}
Any $(\bm{A}^\ddagger, \bm{q}^\ddagger)$ is optimal for the monopolist if and only if there is $(\bm{A}^\dagger, \bm{q}^\dagger)$ that is optimal for the social planner such that $\bm{A}^\ddagger = \bm{A}^\dagger$ and $\bm{q}^\ddagger = \bm{q}^\dagger/2$.
Consequently, $\frac{\Omega(\bm{A}^\ddagger, \bm{q}^\ddagger)}{\Omega(\bm{A}^\dagger, \bm{q}^\dagger)} = \frac{3}{4}$.
\end{proposition}

The next proposition characterizes the social planner's optimal output profile $\bm{q}^\dagger$---and hence, the monopolist's optimum $\bm{q}^\dagger$ via Proposition~\ref{prop_sp_mono}---which in turn determines the configuration of the optimal characteristics profile $\bm{A}^\dagger(\bm{q}^\dagger)$ via Lemmas~\ref{lem_R}--\ref{lem_b}.

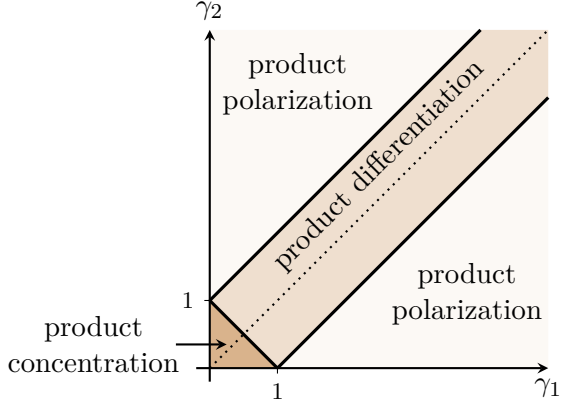
\begin{figure}[t]
\centering
\begin{tikzpicture}[>=stealth, scale=0.9]
\def\max{5}
\useasboundingbox (0,0) rectangle (\max,\max);

\fill[brown!25]
  (0,1) --
  (\max-1,\max) --
  (\max,\max) --
  (\max,\max-1) --
  (1,0) --
  cycle;
  
\fill[brown!60]
  (0,0) --
  (0,1) --
  (1,0) --
  cycle;
  
\fill[brown!5]
  (0,1) --
  (0,\max) --
  (\max-1,\max) --
  cycle;
  
\fill[brown!5]
  (1,0) --
  (\max,0) --
  (\max,\max-1) --
  cycle;

\draw[black, very thick]
  (0,1) -- (1,0);

\draw[black, very thick]
  (1,0) -- (\max,\max-1);

\draw[black, very thick]
  (0,1) -- (\max-1,\max);

\node[rotate=45] at (2.5,3) {product differentiation};

\node[align=center] at (1.3,4.2) {product\\polarization};
\node[align=center] at (3.8,1.1) {product\\polarization};

\node[align=center, anchor=east] at (-0.3,0.35) {product\\concentration};
\draw[->, thick] (-0.5,0.35) -- (0.25,0.35);

\draw[->, thick] (-.2,0) -- (\max,0) node[below] {$\gamma_1$};
\draw[->, thick] (0,-.2) -- (0,\max) node[above] {$\gamma_2$};
\draw (1,0) -- (1,-0.08) node[below, font=\footnotesize] {$1$};
\draw (0,1) -- (-0.08,1) node[left, font=\footnotesize] {$1$};
\draw[dotted, thick] (0,0) -- (\max,\max);

\end{tikzpicture}

\vspace{1em}

\caption{In the case with $n=2$ and $\alpha=1$, each shaded area depicts the set of $(\gamma_1,\gamma_2)$ under which the socially optimal allocation exhibits the corresponding property.}
\label{fig_sp}
\end{figure}

\begin{proposition} \label{prop_sp}
The social planner's optimal output profile $\bm{q}^\dagger$, characteristics profile $\bm{A}^\dagger$, and resulting $\bm{x}^\dagger = \bm{A}^\dagger\bm{q}^\dagger $ are characterized as follows:
\begin{enumerate}[\rm i).]
\item \label{prop_sp1}
If $r(\bm{\gamma}) \le 1 \le R(\bm{\gamma})$, then $\bm{q}^\dagger = \bm{\gamma}$.
In this case, $\bm{x}^\dagger = \bm{\beta}$ holds, and any $\bm{A}^\dagger$ exhibits product differentiation unless there exist $\bm{\sigma} \in \{-1,1\}^n$ such that $\bm{\sigma}^\top \bm{\gamma} = 1$. 

\item \label{prop_sp2}
If $R(\bm{\gamma}) \le 1$, then
\[
q^\dagger_i = \gamma_i + \frac{\alpha \qty(1-R(\bm{\gamma}))}{1+n\alpha}, \quad \forall i \in [n].
\]
In this case, $\bm{x}^\dagger = R(\bm{q}^\dagger)\bm{\beta}$ holds, and $\bm{A}^\dagger$ is unique and exhibits product concentration, with $\bm{a}_i^\dagger = \bm{\beta}$ for all $i \in [n]$.

\item \label{prop_sp3}
If $r(\bm{\gamma}) \ge 1$ with $\gamma_i = \|\bm{\gamma}\|_\infty$, then
\[
q_i = \gamma_i - \frac{\alpha(r(\bm{\gamma})-1)}{1+n\alpha}, \qquad
q_j = \gamma_j + \frac{\alpha(r(\bm{\gamma})-1)}{1+n\alpha}, \quad \forall j \neq i.
\]
In this case, $\bm{x}^\dagger = r(\bm{q}^\dagger)\bm{\beta}$ holds, and $\bm{A}^\dagger$ is unique and exhibits product polarization, with $\bm{a}^\dagger_i = \bm{\beta}$ and $\bm{a}^\dagger_j = -\bm{\beta}$ for all $j \neq i$.
\end{enumerate}
\end{proposition}

There are three distinct patterns of characteristics design that can arise at the social planner's optimum.
These are determined by the summary statistics $R(\bm{\gamma})$ and $r(\bm{\gamma})$ of the distribution of firms' standalone values.
A graphical illustration in the duopoly case is provided in Figure~\ref{fig_sp}.

First, product differentiation arises when $R(\bm{\gamma})$ is sufficiently large while $r(\bm{\gamma})$ is not too large, as described in case~\eqref{prop_sp1}.
The requirement on $R(\bm{\gamma})$ implies that firms' standalone values are sufficiently large in the aggregate.
By contrast, the requirement on $r(\bm{\gamma})$ rules out situations in which a single firm has an excessively large standalone value relative to the others.
The latter condition is automatically satisfied in the symmetric case where $\gamma_i$ is constant across firms, or more generally, if there are at least two firms with nearly equal and largest standalone values.
In such environments, no single firm dominates the supply of common characteristics; instead, the planner designate firms to differentiate common characteristics.

Second, product concentration arises when $R(\bm{\gamma})$ is small, as in case~\eqref{prop_sp2}.
This situation corresponds to an economy in which firms' standalone values are collectively weak, so that idiosyncratic characteristics do not provide substantial appeal to the representative consumer.
In this case, products are relatively plain in terms of brand-specific attributes, and the social planner assigns all firms to supply the same ideal common characteristics.
As a result, product concentration emerges, and firms perfectly aligned with the consumer's ideal.

Third, product polarization arises when $r(\bm{\gamma})$ is large, as in case~\eqref{prop_sp3}.
This is an extreme case in which a single firm has a standalone value that is excessively large relative to others.
The dominant firm is designated to produce common characteristics aligned with the consumer's ideal, while the remaining firms are directed to supply common characteristics pointing in the opposite direction.
Nevertheless, these marginal firms continue to generate value through their idiosyncratic characteristics, which remain appreciated by the consumer.
Thus, the dominant firm monopolizes the provision of desirable common characteristics, and the remaining firms specialize in supplying idiosyncratic attributes.


\section{Cournot Oligopoly with Characteristics Design}
\label{sec_eq}

We now turn to analyzing Cournot oligopoly in the present model, in which each firm $i$ simultaneously chooses both a characteristics vector $\bm{a}_i$ and an output level $q_i$ to maximize its profit.
Denote by $\bm{A}_{-i} = (\bm{a}_1,\ldots,\bm{a}_{i-1}, \bm{a}_{i+1},\ldots,\bm{a}_n)$ and $\bm{q}_{-i} = (q_1,\ldots,q_{i-1},q_{i+1},\ldots,q_n)$ generic choice variables of firm $i$'s rival firms.
Substituting the expression for $i$'s markup \eqref{eq_markup} into its profit function \eqref{eq_profit},
\begin{equation} \label{eq_profit_sep}
\Pi_i \qty(\bm{a}_i, q_i; \bm{A}_{-i}, \bm{q}_{-i}) = \alpha q_i \bm{a}_i^\top \qty(\bm{\beta} - \sum_{j \neq i} q_j \bm{a}_j) - (1+\alpha)q_i^2 +  \gamma_i q_i.
\end{equation}
We say that $(\bm{A}^*,\bm{q}^*)$ constitutes a \emph{Cournot-Nash equilibrium}, or an \emph{equilibrium} in short, if
\[
(\bm{a}_i^*, q_i^*) \in \argmax_{\|\bm{a}_i\|_2 = 1,\, q_i \ge 0} \; \Pi_i \qty(\bm{a}_i, q_i; \bm{A}^*_{-i}, \bm{q}^*_{-i}), \quad \forall i \in [n].
\]

As can be seen from \eqref{eq_profit_sep}, the firm $i$'s total profit $\Pi_i$ depends on its common characteristics vector $\bm{a}_i$ solely through the first term,
\[
\pi_i\qty(q_i \bm{a}_i; \textstyle \sum_{j \neq i} q_j \bm{a}_j) \coloneqq q_i \bm{a}_i^\top \qty(\bm{\beta} - \sum_{j \neq i} q_j \bm{a}_j),
\]
which is a function of $q_i \bm{a}_i$ and $\sum_{j \neq i} q_j \bm{a}_j$.
Taking $\sum_{j \neq i} q_j \bm{a}_j$ as given, the value of $\pi_i$ is maximized by choosing $q_i \bm{a}_i$ to be proportional to $\bm{\beta} - \sum_{j \neq i} q_j \bm{a}_j$, which represents the gap between the consumer's ideal vector and the aggregate common characteristics supplied from other firms.
Specifically, firm $i$'s optimal $\bm{a}_i$ is obtained by scaling $\bm{\beta} - \sum_{j \neq i} q_j \bm{a}_j$ to have unit length if $\sum_{j \neq i} q_j \bm{a}_j \neq \bm{\beta}$, while it is arbitrary if $\sum_{j \neq i} q_j \bm{a}_j = \bm{\beta}$.

In both cases, the maximized value of firm $i$'s payoff from common characteristics design is given as the following linear function of $q_i$:
\begin{equation} \label{opt_value_pi}
\max_{\|\bm{a}_i\|_2 = 1} \; \pi_i\qty(q_i \bm{a}_i; \textstyle \sum_{j \neq i} q_j \bm{a}_j) = \left\|\bm{\beta} - \textstyle \sum_{j \neq i} q_j \bm{a}_j \right\|_2 \cdot q_i.
\end{equation}
Moreover, substituting the maximized value \eqref{opt_value_pi} back into \eqref{eq_profit_sep}, we can also determine the firm $i$'s optimal output level as a function of $\sum_{j \neq i} q_j \bm{a}_j$.
This yields a full characterization of equilibrium conditions, recorded as the following lemma.

\begin{lemma} \label{lem_br}
Given any profile $(\bm{A}_{-i}, \bm{q}_{-i})$ of opponents' strategies, firm $i$'s best response strategy $(\bm{a}_i,q_i)$ is characterized as follows:
\[
\delta_i \bm{a}_i = \bm{\beta} - \sum_{j \neq i} q_j \bm{a}_j, \qquad
q_i = \frac{\alpha \delta_i + \gamma_i}{2(1+\alpha)}, \qquad \text{where} \quad \delta_i = \left\|\bm{\beta} - \textstyle \sum_{j \neq i} q_j \bm{a}_j \right\|_2.
\]
Moreover, $(\bm{A}^*,\bm{q}^*)$ constitutes an equilibrium if and only if the following are satisfied:
\begin{align}
2(1+\alpha) \bm{q}^* &\ge \bm{\gamma}; \quad \text{and} \label{eq_cond1} \\
\qty((2+\alpha)q_i^* - \gamma_i) \bm{a}_i^* &= \alpha \qty(\bm{\beta} - \bm{A}^* \bm{q}^*), \quad \forall i \in [n]. \label{eq_cond2}
\end{align}
Finally, for any equilibrium $(\bm{A}^*,\bm{q}^*)$, each firm~$i$'s markup and profit depend solely on $q_i^*$ and are given by $(1+\alpha)q_i^*$ and $(1+\alpha)q_i^{*2}$, respectively.
\end{lemma}

The equilibrium characterization in this lemma is obtained by eliminating $\delta_i$ from each firm's best-response conditions, where the inequality condition~\eqref{eq_cond1} follows directly from the nonnegativity of $\delta_i$, as it represents the norm of a vector.
The key restriction is instead condition~\eqref{eq_cond2}, which is informative about the equilibrium pattern of product differentiation.
To see this, since the right-hand side of \eqref{eq_cond2} is identical across firms, any equilibrium must satisfy
\begin{equation} \label{eq_delta}
\bigl((2+\alpha)q_i^* - \gamma_i\bigr)\bm{a}_i^*
=
\bigl((2+\alpha)q_j^* - \gamma_j\bigr)\bm{a}_j^*,
\quad \forall i,j \in [n].
\end{equation}
In particular, since each $\bm{a}_i^*$ is a nonzero vector, $(2+\alpha)q_i^*=\gamma_i$ holds for some $i$ if and only if it holds for every $i$.
This implies that if $(2+\alpha)q_i^*\neq\gamma_i$, the common-characteristics vectors chosen by any two firms must be parallel, i.e., $\bm{a}_i^{*\top}\bm{a}_j^*=\pm1$ for all $i,j\in[n]$.
As such, if all inner products equal $1$, all firms choose identical characteristics vectors and the equilibrium exhibits product concentration.
If instead there exist $i\neq j$ such that $\bm{a}_i^{*\top}\bm{a}_j^*=-1$, the equilibrium exhibits product polarization.

On the flip side, \eqref{eq_delta} also implies that the condition $(2+\alpha)q_i^*=\gamma_i$ is necessary for an equilibrium to exhibit product differentiation, and when this condition holds, the equilibrium output profile is uniquely determined as $\bm{q}^*=\frac{\bm{\gamma}}{2+\alpha}$.
Condition~\eqref{eq_cond2} then dictates the characteristics profile to satisfy $\bm{A}^*\bm{q}^*=\bm{\beta}$, so that the aggregate supply of common characteristics coincides with the consumer's ideal.
Nevertheless, this allocation need not be socially optimal, since the equilibrium output profile differs from the planner's choice.


\begin{figure}[t]
\centering
\begin{minipage}[t]{0.24\textwidth}
\centering
\begin{tikzpicture}[>=stealth, scale=3/5]
\def\xmax{5}
\def\ymax{5}

\fill[brown!25]
  (0,3) --
  (2,\ymax) --
  (\xmax,\ymax) --
  (\xmax,2) --
  (3,0) --
  cycle;

\draw[black, very thick] (0,3) -- (3,0);
\draw[black, very thick] (3,0) -- (\xmax,2);
\draw[black, very thick] (0,3) -- (2,\ymax);

\node[align=center, font=\footnotesize] at (2.75,2.75) {product \\ differentiation};

\draw[->, thick] (-.2,0) -- (\xmax,0) node[below] {$\gamma_1$};
\draw[->, thick] (0,-.2) -- (0,\ymax) node[above] {$\gamma_2$};
\draw (3,0) -- (3,-0.08) node[below, font=\footnotesize] {$3$};
\draw (0,3) -- (-0.08,3) node[left, font=\footnotesize] {$3$};
\draw[dotted] (0,0) -- (\xmax,\ymax);
\end{tikzpicture}
\end{minipage}%
\hfill
\begin{minipage}[t]{0.24\textwidth}
\centering
\begin{tikzpicture}[>=stealth, scale=3/5]
\def\xmax{5}
\def\ymax{5}
\pgfmathsetmacro{\xTop}{(\ymax+12)/4}
\pgfmathsetmacro{\yRight}{(\xmax+12)/4}

\fill[brown!25]
  (0,0) --
  (3,0) --
  (4,4) --
  cycle;

\fill[brown!25]
  (0,0) --
  (4,4) --
  (0,3) --
  cycle;

\draw[black, very thick] (3,0) -- (\xTop,\ymax);
\draw[black, very thick] (0,3) -- (\xmax,\yRight);

\node[align=center, font=\footnotesize] at (1.75,2.25) {product \\ concentration};

\draw[->, thick] (-.2,0) -- (\xmax,0) node[below] {$\gamma_1$};
\draw[->, thick] (0,-.2) -- (0,\ymax) node[above] {$\gamma_2$};
\draw (3,0) -- (3,-0.08) node[below, font=\footnotesize] {$3$};
\draw (0,3) -- (-0.08,3) node[left, font=\footnotesize] {$3$};
\draw[dotted] (0,0) -- (\xmax,\ymax);
\end{tikzpicture}
\end{minipage}%
\hfill
\begin{minipage}[t]{0.24\textwidth}
\centering
\begin{tikzpicture}[>=stealth, scale=3/5]
\def\xmax{5}
\def\ymax{5}
\pgfmathsetmacro{\xTopRed}{(12-\ymax)/4}

\fill[brown!25]
  (0,\ymax) --
  (\xTopRed,\ymax) --
  (3,0) --
  (\xmax,0) --
  (\xmax,\ymax) --
  cycle;

\draw[black, very thick] (\xTopRed,\ymax) -- (3,0);

\node[align=center, font=\footnotesize] at (3.6,4)
  {product \\ polarization \\ (firm 1)};

\draw[->, thick] (-.2,0) -- (\xmax,0) node[below] {$\gamma_1$};
\draw[->, thick] (0,-.2) -- (0,\ymax) node[above] {$\gamma_2$};
\draw (3,0) -- (3,-0.08) node[below, font=\footnotesize] {$3$};
\draw (0,3) -- (-0.08,3) node[left, font=\footnotesize] {$3$};
\draw[dotted] (0,0) -- (\xmax,\ymax);
\end{tikzpicture}
\end{minipage}%
\hfill
\begin{minipage}[t]{0.24\textwidth}
\centering
\begin{tikzpicture}[>=stealth, scale=3/5]
\def\xmax{5}
\def\ymax{5}

\fill[brown!25]
  (0,3) --
  (0,\ymax) --
  (\xmax,\ymax) --
  (\xmax,3-\xmax/4) --
  cycle;

\draw[black, very thick] (0,3) -- (\xmax,3-\xmax/4);

\node[align=center, font=\footnotesize] at (2.5,4)
  {product \\ polarization \\ (firm 2)};

\draw[->, thick] (-.2,0) -- (\xmax,0) node[below] {$\gamma_1$};
\draw[->, thick] (0,-.2) -- (0,\ymax) node[above] {$\gamma_2$};
\draw (3,0) -- (3,-0.08) node[below, font=\footnotesize] {$3$};
\draw (0,3) -- (-0.08,3) node[left, font=\footnotesize] {$3$};
\draw[dotted] (0,0) -- (\xmax,\ymax);
\end{tikzpicture}
\end{minipage}%
\hfill

\caption{
In the case with $n=2$ and $\alpha=1$, each shaded area depicts the set of $(\gamma_1,\gamma_2)$ under which an equilibrium of the corresponding type exists.
}
\label{fig_duopoly_exist}
\end{figure}
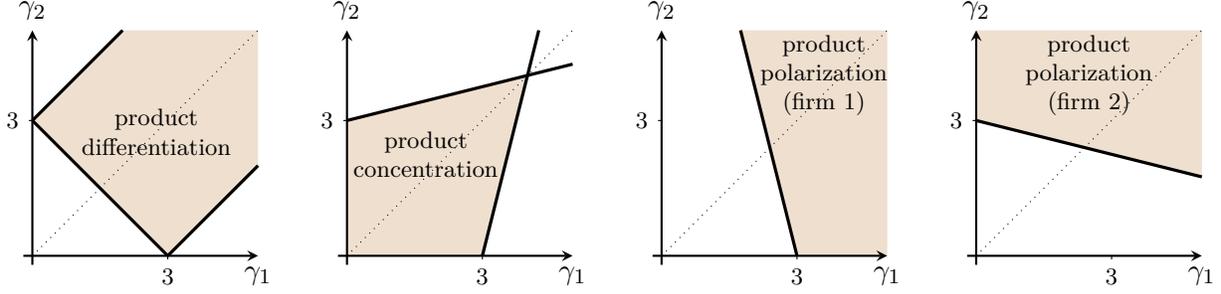

\subsection{Equilibrium with Product Differentiation}

The next proposition characterizes equilibria that exhibit product differentiation.

\begin{proposition} \label{prop_eq_D}
If and only if $r(\bm{\gamma}) \le 2+\alpha \le R(\bm{\gamma})$, there exists an equilibrium $(\bm{A}^{\rm d},\bm{q}^{\rm d})$ such that
\begin{equation} \label{profile_diff}
\bm{A}^{\rm d} \bm{q}^{\rm d} = \bm{\beta}, \qquad \bm{q}^{\rm d} = \frac{\bm{\gamma}}{2+\alpha}.
\end{equation}
\end{proposition}

\begin{remark}
By Lemma~\ref{lem_b}, this equilibrium exhibits product differentiation generically, unless there exists $\bm{\sigma} \in \{-1,1\}^n$ such that $\bm{\sigma}^\top \bm{\gamma} = 2+\alpha$, so we simply call it a product-differentiation equilibrium.
Moreover, there exists no other equilibrium exhibiting product differentiation since Lemma~\ref{lem_br} implies that any such equilibrium must satisfy \eqref{profile_diff}.
\end{remark}

While there remains some flexibility in the choice of each individual $\bm{a}_i^{\rm d}$, as illustrated in Remark~\ref{remark_multiplicity}, note that \eqref{profile_diff} pins down both the output profile $\bm{q}^{\rm d}$ and the aggregate supply of common characteristics $\bm{A}^{\rm d}\bm{q}^{\rm d}$.
As a result, social welfare is identical across all equilibria that exhibit product differentiation.
Moreover, together with Lemma~\ref{lem_br}, the uniqueness of $\bm{q}^{\rm d}$ implies that each firm's markup and profit are also invariant across all such equilibria.

The parameter region for product differentiation to emerge in oligopoly is related to, but different from, the corresponding region in the social planner's problem.
Specifically, the relevant cutoff shifts from $1$ to $2+\alpha$, which tightens the restriction on $R(\bm{\gamma})$ while relaxing the restriction on $r(\bm{\gamma})$.
To facilitate comparison among the three allocations in the social planner's problem, monopoly, and oligopoly, we focus on the parameter region
\begin{equation} \label{gamma_diff}
r(\bm{\gamma}) \le 1 < 2+\alpha \le R(\bm{\gamma}), \tag{D}
\end{equation}
under which product differentiation arises in all three cases.
Recalling the explicit forms of the output profiles characterized in Propositions~\ref{prop_sp_mono}--\ref{prop_eq_D}, we obtain the following ranking:
\[
\bm{q}^\dagger = \bm{\gamma}
\gg
\bm{q}^\ddagger = \frac{\bm{\gamma}}{2}
\gg
\bm{q}^{\rm d} = \frac{\bm{\gamma}}{2+\alpha}.
\]
That is, the social planner supplies the largest output, followed by the monopolist, while output under oligopoly is the smallest.

The rationale behind this output ranking is well understood by relating it with the equilibrium pattern of product differentiation.
At this point, it is not immediately clear how firms' characteristics choices are positioned relative to $\bm{\beta}$ and to one another since each firm faces a trade-off in its characteristics design.
On the one hand, choosing $\bm{a}_i$ closer to the consumer's ideal vector $\bm{\beta}$ raises demand.
On the other hand, doing so can intensify competition and reduce markups if opponents also attempt to move their characteristics closer to $\bm{\beta}$.

To quantify the overall effect of this trade-off, we consider the weighted average of cosine similarities between $\bm{a}_i$ and $\bm{a}_j$, defined for any characteristics profile $\bm{A}$ as follows:\footnote{The summation is taken over all unordered pairs of distinct firms, since own-firm cosine similarity is always one by normalization. Here the weighted index is considered for technical purposes to prove Corollary~\ref{cor_cosine_diff}, while $\bar{s}_{\bm{\gamma}}$ reduces to the simple average when $\gamma_i$ is identical across all firms.}
\[
\bar{s}_{\bm{\gamma}}(A) \coloneqq \frac{1}{\binom{n}{2}} \sum_{\{i,j\} \subseteq [n]} \gamma_i \gamma_j \bm{a}_i^\top \bm{a}_j .
\]
The next result shows that under condition \eqref{gamma_diff}, the value of $\bar{s}_{\bm{\gamma}}$ is higher in oligopoly than in the social planner's and monopolist's benchmarks, reflecting stronger alignment among firms' characteristics on average.

\begin{corollary} \label{cor_cosine_diff}
Under condition~\eqref{gamma_diff}, we have
\[
\bar{s}_{\bm{\gamma}}(\bm{A}^\dagger) = \bar{s}_{\bm{\gamma}}(\bm{A}^\ddagger) = \frac{1 - \|\bm{\gamma}\|_2^2}{n(n-1)}
< \frac{(2+\alpha)^2 - \|\bm{\gamma}\|_2^2}{n(n-1)} = \bar{s}_{\bm{\gamma}}(\bm{A}^{\rm d}).
\]
\end{corollary}

Thus, common characteristics are more concentrated under monopoly than under oligopoly.
This difference is driven by the presence of idiosyncratic characteristics, which are firm-specific and cannot be replicated by rival firms, thereby mitigating firms' concerns about increased competition.
Specifically, condition~\eqref{gamma_diff} implies that these characteristics are sufficiently valuable to the representative consumer.
As a result, each firm is motivated to exploit its standalone value to raise markup by lowering its output level, resulting in the aforementioned output ranking relative to the monopolist benchmark.
At the same time, the high standalone values allow firms to compete aggressively to let them go where the demand is, i.e., they locate common characteristics closer to the consumer's ideal.

To sum up, while condition~\eqref{gamma_diff} ensures that product differentiation arises in equilibrium, the resulting degree of differentiation remains smaller than what is socially optimal.
This naturally raises the question of welfare implications.
Strikingly, we will see that when product differentiation is socially optimal and a weak additional assumption is satisfied, social welfare under monopoly is strictly higher than under oligopoly; see Theorem~\ref{thm_welfare_D}

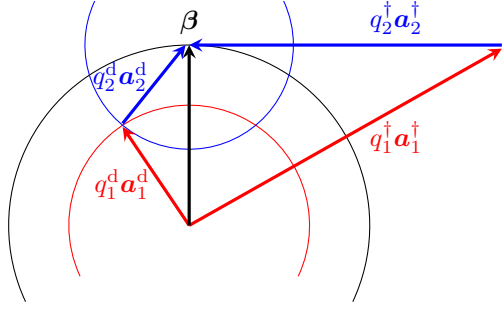
\begin{figure}[t]
\centering
\begin{tikzpicture}[>=stealth, scale=0.8]
\usetikzlibrary{calc,intersections}

\def\R{3}
\def\alpha{1}
\def\gOne{2}
\def\gTwo{sqrt(3)}
\pgfmathsetmacro{\Rred}{\gOne/(2+\alpha) * \R}
\pgfmathsetmacro{\Rblue}{\gTwo/(2+\alpha) * \R}
\def\angUpA{-25}
\def\angUpB{205}
\def\angBlueA{155}
\def\angBlueB{385}

\useasboundingbox (-4.2,-3.8) rectangle (4.2,4.2);

\coordinate (O) at (0,0);
\coordinate (B) at (0,\R);
\coordinate (P) at ({sqrt(3)*\R},\R);

\draw (O) ++(\angUpA:\R)
  arc[start angle=\angUpA, end angle=\angUpB, radius=\R];
\draw[red] (O) ++(\angUpA:\Rred)
  arc[start angle=\angUpA, end angle=\angUpB, radius=\Rred];
\draw[blue] (B) ++(\angBlueA:\Rblue)
  arc[start angle=\angBlueA, end angle=\angBlueB, radius=\Rblue];

\path[name path=rc] (O) circle (\Rred);
\path[name path=bc] (B) circle (\Rblue);
\path[name intersections={of=rc and bc, by={I1,I2}}];

\makeatletter
\newdimen\Xone \newdimen\Xtwo
\pgfextractx{\Xone}{\pgfpointanchor{I1}{center}}
\pgfextractx{\Xtwo}{\pgfpointanchor{I2}{center}}
\ifdim\Xone<\Xtwo \coordinate (I) at (I1); \else \coordinate (I) at (I2); \fi
\makeatother

\coordinate (Ired)  at ($(I)+(0,-0.03)$);
\coordinate (Iblue) at ($(I)+(0,0.03)$);
\coordinate (Bblue) at ($(B)-(0.06,0)$);
\coordinate (Pred)  at ($(P)-(0,0.06)$);

\draw[->, very thick, red]
  (O) -- (Pred) node[midway, right, xshift=5pt, font=\small] {$q_1^\dagger \bm{a}_1^\dagger$};
\draw[->, very thick, blue]
  (P) -- (B) node[midway, above right, xshift=5pt, font=\small] {$q_2^\dagger \bm{a}_2^\dagger$};
\draw[->, very thick, red]
  (O) -- (Ired) node[midway, left, font=\small, xshift=2pt, yshift=-3pt] {$q_1^{\rm d} \bm{a}_1^{\rm d}$};
\draw[->, very thick, blue]
  (Iblue) -- (Bblue) node[midway, left, font=\small, xshift=2pt, yshift=3pt] {$q_2^{\rm d} \bm{a}_2^{\rm d}$};
\draw[->, very thick]
  (O) -- (B) node[above, font=\small] {$\bm{\beta}$};

\end{tikzpicture}
\vspace{-5em}
\caption{Socially optimal and equilibrium product differentiation in the duopoly case with $m=n=2$, $\alpha=1$, $\bm{\beta}=(0,1)$, $\bm{q}^\dagger = \bm{\gamma} = (2,\sqrt{3})$, and $\bm{q}^{\rm d} = \frac{\bm{\gamma}}{2+\alpha} = (\frac{2}{3},\frac{1}{\sqrt{3}})$.}
\label{fig_duopoly_D}
\end{figure}

\begin{example} \label{ex_duopoly}
Consider the duopoly case with two common characteristics, $m=n=2$, with $\alpha=1$ and $\bm{\beta}=(0,1)$.
Suppose that the firms' standalone values are $\gamma_1=2$ and $\gamma_2=\sqrt{3}$.
Then $R(\bm{\gamma})>1$, and Proposition~\ref{prop_sp} implies that the social planner's optimum exhibits product differentiation with $\bm{q}^\dagger=\bm{\gamma}$.
The associated characteristics profile is such that $\bm{a}_1$ and $\bm{a}_2$ form angles $\frac{\pi}{6}$ and $\pi$ on the unit circle, respectively, or their mirror image with respect to the vertical axis.
In either case, the resulting cosine similarity equals $-\frac{\sqrt{3}}{2}$, and the aggregate provision of common characteristics is $q_1^\dagger \bm{a}_1^\dagger + q_2^\dagger \bm{a}_2^\dagger = \bm{\beta}$; see Figure~\ref{fig_duopoly_D}.
The monopolist selects the same characteristics profile but produces at a lower scale, with $\bm{q}^\ddagger=\frac{\bm{\gamma}}{2}$.

Moreover, since $R(\bm{\gamma})>2+\alpha=3$, there exists an equilibrium that also exhibits product differentiation.
By Proposition~\ref{prop_eq_D}, the equilibrium output profile is $\bm{q}^{\rm d}=\frac{\bm{\gamma}}{3}$, which confirms the ranking $\bm{q}^\dagger \gg \bm{q}^\ddagger \gg \bm{q}^{\rm d}$.
In addition, Corollary~\ref{cor_cosine_diff} implies that, in this equilibrium, the cosine similarity between the two firms' characteristics vectors equals $\frac{1}{2\sqrt{3}}$, which is substantially larger than that under the social planner's optimum.
One such equilibrium is depicted in Figure~\ref{fig_duopoly_D}, whereas, again, another equilibrium exists as the mirror image with respect to the vertical axis.
\end{example}


\subsection{Equilibrium with Product Concentration/Polarization}

Next, we consider equilibria that exhibit either product concentration or product polarization, in which each firm's common-characteristics vector is aligned with $\bm{\beta}$ or with its negation $-\bm{\beta}$.
The next proposition characterizes all such equilibria.

\begin{proposition} \label{prop_eq_CP}
Given any $\bm{\sigma} \in \{-1,1\}^n$, let $N_+ = \{i: \sigma_i = 1\}$ and $N_- = \{i: \sigma_i = -1\}$.
There exists an equilibrium such that $\bm{a}^*_i = \sigma_i \bm{\beta}$ for each $i \in [n]$ if and only if
\begin{equation} \label{cond_rank1}
(2+\alpha) - \qty(\frac{2+(n+1)\alpha}{2(1+\alpha)}) \cdot \min_{j \in N_-} \gamma_j \le \bm{\sigma}^\top \bm{\gamma} \le (2+\alpha) + \qty(\frac{2+(n+1)\alpha}{2(1+\alpha)}) \cdot \min_{j \in N_+} \gamma_j.
\end{equation}
In this equilibrium, the output profile is uniquely determined as
\begin{equation} \label{eq_rank1}
\bm{q}^* =
\frac{\bm{\gamma}}{2+\alpha} - \phi(\bm{\sigma}) \cdot \bm{\sigma}, \quad \text{where} \quad \phi(\bm{\sigma}) \coloneqq \frac{\alpha(\bm{\sigma}^\top \bm{\gamma} - (2+\alpha))}{(2+\alpha)(2+(n+1)\alpha)}.
\end{equation}
\end{proposition}

\begin{remark} \label{remark_min_empty}
Throughout, we adopt the convention that the minimum over an empty set is $+\infty$.
\end{remark}

\begin{remark}
The lower-bound condition in \eqref{cond_rank1} is violated for $\bm{\sigma}=(-1,\ldots,-1)$.
Hence, in any equilibrium characterized in Proposition~\ref{prop_eq_CP}, there must exist at least one firm $i$ such that $\bm{a}_i=\bm{\beta}$.
\end{remark}

This proposition encompasses product concentration as the special case in which $\bm{\sigma}=\bm{1}$.
It also includes the extreme form of product polarization that appears in Proposition~\ref{prop_sp};
let us refer to an equilibrium $(\bm{A}^*,\bm{q}^*)$ as exhibiting \emph{dominant-firm polarization} if $\bm{a}_i^*=\bm{\beta}$ for one and only one firm $i$ such that $\gamma_i=\|\bm{\gamma}\|_\infty$, and $\bm{a}_j^*=-\bm{\beta}$ for all $j\neq i$.
Beyond these cases, however, the proposition reveals additional forms of product polarization that can arise in oligopoly but never occur at the social optimum.

Let $(\bm{A}^{\bm{\sigma}},\bm{q}^{\bm{\sigma}})$ be an equilibrium defined by some $\bm{\sigma}\in\{-1,1\}^n$.
Together with Proposition~\ref{prop_eq_D}, equation~\eqref{eq_rank1} implies that $\bm{q}^{\bm{\sigma}}$ can be expressed as
\[
\bm{q}^{\bm{\sigma}} = \bm{q}^{\rm d} - \phi(\bm{\sigma}) \cdot \bm{\sigma}.
\]
Thus, relative to the product-differentiation equilibrium, each firm~$i$'s output level is shifted by a common magnitude $\phi(\bm{\sigma})$, with the direction of the shift determined by the sign of $\sigma_i$.
Moreover, Lemma~\ref{lem_br} implies that the sign of $\sigma_i\phi(\bm{\sigma})$ determines whether firm~$i$ earns a lower or higher markup and profit relative to the product-differentiation equilibrium.
We will address the welfare implications of this relationship in Theorem~\ref{thm_welfare_DCP}.

Let $(\bm{A}^{\rm c}, \bm{q}^{\rm c})$ denote the special case of product concentration with $\bm{\sigma}=\bm{1}$.
This equilibrium exists if and only if
\begin{equation} \label{cond_conc}
R(\bm{\gamma}) \le (2+\alpha) + \qty(\frac{2+(n+1)\alpha}{2(1+\alpha)}) \cdot \min_{j \in [n]} \gamma_j.
\end{equation}
This condition is markedly milder than the corresponding condition in the social planner's problem since the threshold for $R(\bm{\gamma})$ shifts to $2+\alpha$, rather than $1$, plus an additional slack term, which becomes material when the minimum standalone value is large; for illustration, compare the second panel of Figure~\ref{fig_duopoly_exist} with Figure~\ref{fig_sp}.
This observation is consistent with Corollary~\ref{cor_cosine_diff}, which shows that even when product differentiation arises in oligopoly, the equilibrium characteristics profile is more concentrated than that in the social planner's optimum.

Let us compare the product-concentration equilibrium with the social planner's and monopolist's benchmarks by focusing on the parameter region
\begin{equation} \label{gamma_conc}
R(\bm{\gamma}) \le 1, \tag{C}
\end{equation}
under which product concentration is socially optimal as well.
Under this condition, we obtain the following ranking of aggregate output levels across firms:
\[
R(\bm{q}^\dagger) = \frac{n\alpha+R(\bm{\gamma})}{n\alpha+1}
> R(\bm{q}^{\rm c}) = \frac{n\alpha+R(\bm{\gamma})}{n\alpha + 2 + \alpha}
> R(\bm{q}^\ddagger) = \frac{n\alpha+R(\bm{\gamma})}{2n\alpha+2}.
\]
Namely, the ranking between oligopoly and monopoly is reversed relative to the case of product differentiation.
The degree of underproduction is therefore less severe under oligopoly than under monopoly in the present case, and this will lead to a potential reversal of the social-welfare ranking; see Theorem~\ref{thm_welfare_C}.

Beyond condition~\eqref{gamma_conc}, the product-concentration equilibrium can coexist with the product-differentiation equilibrium when parameter values fall in the intersection of condition~\eqref{cond_conc} and the condition in Proposition~\ref{prop_eq_D}, i.e.,
\begin{equation} \label{gamma_CD}
r(\bm{\gamma}) \le 2+\alpha \le R(\bm{\gamma}) \le (2+\alpha) + \qty(\frac{2+(n+1)\alpha}{2(1+\alpha)}) \cdot \min_{j \in [n]} \gamma_j,
\end{equation}
Moreover, as long as the second inequality is strict, the corresponding output profiles are ranked as
\[
\bm{q}^{\rm d} = \frac{\bm{\gamma}}{2+\alpha} \gg
\bm{q}^{\rm c} = \frac{\bm{\gamma}}{2+\alpha} - \phi(\bm{1}) \cdot \1,
\]
since $\phi(\bm{1}) > 0$ whenever $R(\bm{\gamma}) > 2+\alpha$.
Hence, when both types of equilibria coexist, firms underproduce in the product-concentration equilibrium relative to the product-differentiation equilibrium, resulting in lower markups and profit levels.
This relationship further determines the welfare ranking between these equilibria; see Theorem~\ref{thm_welfare_DCP}~\eqref{thm_welfare_DC}.
We remark that the above comparison between $\bm{q}^{\rm d}$ and $\bm{q}^{\rm c}$ does not contradict the ranking relative to the monopolist benchmark discussed early, since condition~\eqref{gamma_CD} is disjoint from~\eqref{gamma_conc}.

\setcounter{example}{0}
\begin{example}[continued]
For the specified parameter values, the right-hand side of condition~\eqref{cond_conc} becomes $3+\frac{5\sqrt{3}}{4}$, which is greater than $R(\bm{\gamma}) = 2 + \sqrt{3}$.
Therefore, Proposition~\ref{prop_eq_CP} implies that there exists an equilibrium exhibiting product concentration.
The associated output profile is given by $q_1^{\rm c}=\frac{11}{15}-\frac{1}{5\sqrt{3}}$ and $q_2^{\rm c}=\frac{1}{15}+\frac{4}{5\sqrt{3}}$, which is strictly smaller than that in the equilibrium with product differentiation.
Lastly, we note that no equilibrium exhibiting product polarization exists for the current parameter values, since the lower-bound condition in~\eqref{cond_rank1} is violated whenever $\bm{\sigma}\neq(1,1)$.
\end{example}

We turn to another special case of Proposition~\ref{prop_eq_CP}, in which one and only one firm $i$ chooses $\bm{a}_i=\bm{\beta}$.
Unlike in the social planner's problem, such polarization can arise in oligopoly even when firm~$i$ does not have the highest standalone value.
Indeed, substituting the vector $\bm{\sigma}$ with $\sigma_i=1$ and $\sigma_j=-1$ for all $j\neq i$ into \eqref{cond_rank1} yields
\begin{equation} \label{cond_pola}
\tilde{r}_i(\bm{\gamma}) \ge 2+\alpha - \qty(\frac{2+(n+1)\alpha}{2+2\alpha}) \cdot \max_{j\neq i}\gamma_j,
\quad \text{where} \quad
\tilde{r}_i(\bm{\gamma}) \coloneqq \gamma_i - \sum_{j\neq i}\gamma_j .
\end{equation}
This condition imposes a restriction on $\tilde{r}_i(\bm{\gamma})$---i.e., the difference between firm~$i$'s standalone value and the sum of its opponents'---rather than on $r(\bm{\gamma})$.
In addition, the threshold for $\tilde{r}_i(\bm{\gamma})$ shifts to $2+\alpha$, rather than $1$, and is further relaxed by an additional term proportional to the largest standalone value among the opponents.
This term is positive and typically sizable, so that the overall threshold for $\tilde{r}_i(\bm{\gamma})$ can be negative.
Taken together, these imply that multiple equilibria exhibiting product polarization can exist, differing in which firm aligns its common-characteristics choice with the consumer's ideal.
The right panels of Figure~\ref{fig_duopoly_exist} illustrate the relevant parameter regions in the duopoly case.

More generally, there exist multiple product-polarization equilibria that differ in the set of firms aligning positively with $\bm{\beta}$.
This multiplicity points to a source of inefficiency that is specific to oligopoly and absent under monopoly.
In oligopoly, firms with relatively low standalone values---potentially reflecting higher marginal costs---may nonetheless dominate the provision of common characteristics.
Once such firms choose to align positively with $\bm{\beta}$, rival firms may have no incentive to realign their characteristics choices toward $\bm{\beta}$, even when doing so would be socially desirable.
Indeed, we will show that social welfare under such ``inefficient'' product-polarization equilibria is lower than under the product-differentiation equilibrium; see Theorem~\ref{thm_welfare_DCP}~\eqref{thm_welfare_DP}.

\begin{figure}[t!]
\begin{center}
\includegraphics[width=4in]{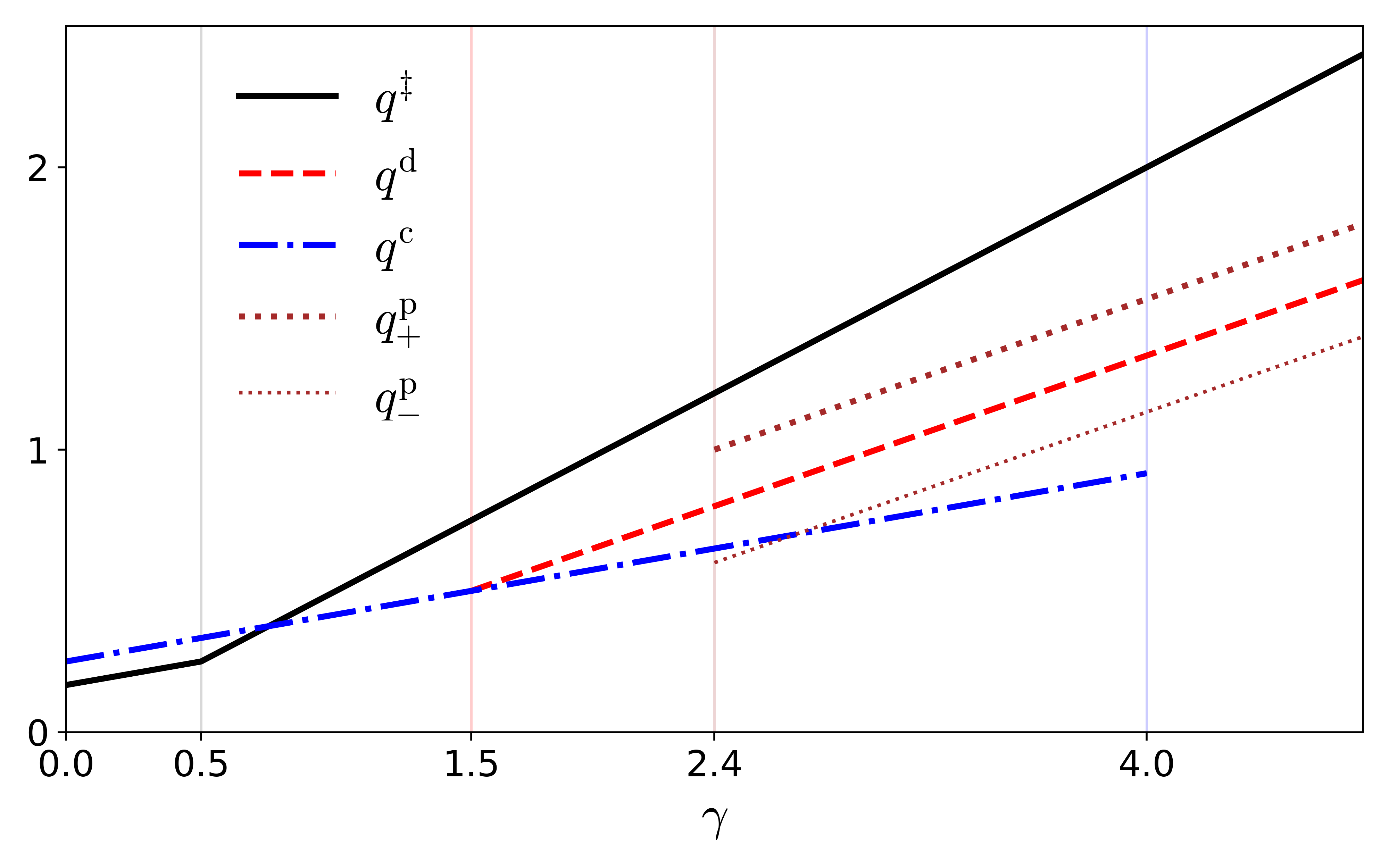}
\caption{In the symmetric case with $n=2$ and $\alpha=1$, the output levels calculated in Example~\ref{ex_sym} are plotted as functions of $\gamma$.}
\label{fig_sym_q}
\end{center}
\end{figure}

\begin{example} \label{ex_sym}
Consider the symmetric case with $\gamma_i=\gamma$ for all $i\in[n]$.
The monopolist's optimal output profile $\bm{q}^\ddagger=q^\ddagger\1$ is symmetric and given by
\[
q^\ddagger = \begin{cases}
\frac{\gamma+\alpha}{2(1+n\alpha)} &\text{if} \quad \gamma \le \frac{1}{n}, \\
\frac{\gamma}{2} &\text{if} \quad \gamma \le \frac{1}{n}.
\end{cases}
\]
Likewise, the product-differentiation equilibrium is symmetric, with its output level and existence condition given by
\[
q^{\rm d}=\frac{\gamma}{2+\alpha} \quad \text{if} \quad \gamma \ge \frac{2+\alpha}{n}.
\]
Also, the product-concentration equilibrium is given by
\[
q^{\rm c}=\frac{\gamma+\alpha}{2+(n+1)\alpha} \quad \text{if} \quad \gamma \le \frac{2(1+\alpha)}{n-1}.
\]
Lastly, assuming that $n$ is even, we consider an equilibrium exhibiting ``even'' product polarization with $N_+ = N_-$.\footnote{Note that this even splitting is the only possible form of product polarization in duopoly.}
This equilibrium is characterized by
\[
q^{\rm p}_{\pm}=\frac{\gamma}{2+\alpha} \pm \frac{\alpha(2+\alpha)}{2+(n+1)\alpha} \quad \text{if} \quad
\gamma \ge \frac{2(1+\alpha)(2+\alpha)}{2+(n+1)\alpha},
\]
where the sign $\pm$ indicates whether each firm's characteristics vector is aligned positively or negatively with $\bm{\beta}$.
Among these three equilibria, the output levels are uniformly ranked as $q^{\rm p}_+ > q^{\rm d} > q^{\rm p}_-$ and $q^{\rm d} \ge q^{\rm c}$ on the relevant parameter regions.
In addition, the monopolist's output level $q^\ddagger$ dominates these equilibrium output levels, except when $\gamma$ is small wherein only the product-concentration equilibrium exists.
Figure~\ref{fig_sym_q} depicts these relationships in the duopoly case.
\end{example}

To conclude this section, we note that the equilibrium characterizations given in Propositions~\ref{prop_eq_D} and~\ref{prop_eq_CP} are complete, in the sense that at least one equilibrium described in these propositions always exists, while no equilibria other than those described there exist.
In fact, a stronger statement holds: for any parameter values, there exists an equilibrium that exhibits one of the patterns that can arise in the social planner's problem.

\begin{theorem} \label{thm_eq}
For any parameter values, there exists at least one equilibrium that exhibits either product differentiation, product concentration, or dominant-firm polarization.
Moreover, there exists no equilibrium other than those characterized in Propositions~\ref{prop_eq_D} and~\ref{prop_eq_CP}.
\end{theorem}


\section{Welfare Implications}
\label{sec_welfare}

Throughout this section, whenever we refer to a particular equilibrium characterized in the previous section, we assume that the corresponding parameter conditions for its existence are satisfied.
Let $(\bm{A}^{\rm d}, \bm{q}^{\rm d})$ denote the equilibrium characterized in Proposition~\ref{prop_eq_D}, and let $(\bm{A}^{\bm{\sigma}}, \bm{q}^{\bm{\sigma}})$ be any equilibrium characterized in Proposition~\ref{prop_eq_CP} for some $\bm{\sigma} \in \{-1,1\}^n$.
In particular, let $(\bm{A}^{\rm c}, \bm{q}^{\rm c})$ denote the product-concentration equilibrium with $\bm{\sigma} = \1$.
In addition, let $(\bm{A}^\dagger,\bm{q}^\dagger)$ and $(\bm{A}^\ddagger,\bm{q}^\ddagger)$ denote the social planner's and the monopolist's optimal allocations, respectively.

\subsection{Welfare Comparison Between Monopoly and Oligopoly}

First, we compare oligopoly outcomes with the monopolist benchmark.
The welfare implications of this exercise hinge on the underlying parameters as well as on which equilibrium is played.
In the next two theorems, we compare social welfare under oligopoly and monopoly by focusing on two parameter regions in which both exhibit the same pattern of characteristics design; the first features product differentiation, and the second features product concentration.

\begin{theorem} \label{thm_welfare_D}
Under condition~\eqref{gamma_diff}, we have $\Omega(\bm{A}^{\rm d},\bm{q}^{\rm d}) < \Omega(\bm{A}^\ddagger,\bm{q}^\ddagger)$ if and only if $\|\bm{\gamma}\|_2 > \frac{2+\alpha}{\sqrt{4+3\alpha}}$.
\end{theorem}

This theorem shows that social welfare is higher under monopoly than under oligopoly when product differentiation is socially optimal, provided that an additional condition on $\bm{\gamma}$ is satisfied.
This additional condition is fairly weak given that condition~\eqref{gamma_diff} is already imposed.
Indeed, \eqref{gamma_diff} posits that the $\ell_1$ norm of $\bm{\gamma}$ exceeds $2+\alpha$, which is itself larger than the threshold imposed on the $\ell_2$ norm of $\bm{\gamma}$.
While this does not automatically guarantee the additional condition---since the $\ell_2$ norm of any vector is weakly smaller than its $\ell_1$ norm---it suggests that the restriction is unlikely to be binding.
To gain insights, we define the variance of standalone values as
\[
\Var\qty[\bm{\gamma}] \coloneqq \frac{1}{n} \sum_{i=1}^n \qty(\gamma_i - \frac{R(\bm{\gamma})}{n})^2 = \frac{1}{n} \qty(\|\bm{\gamma}\|_2^2 - \frac{R(\bm{\gamma})^2}{n}) .
\]
Then, simple algebra shows that the condition can be rewritten as
\[
\frac{n^2(4+3\alpha)}{(2+\alpha)^2} \cdot \Var\qty[\bm{\gamma}]
> n - (4+3\alpha)\qty(\frac{R(\bm{\gamma})}{2+\alpha})^2 .
\]
Under \eqref{gamma_diff}, the right-hand side is at most $n-(4+3\alpha)$, whereas the left-hand side is nonnegative by definition.
Therefore, the inequality is satisfied if $n<4+3\alpha$, which holds whenever the number of firms is at most $4$.

To gain economic intuition for Theorem~\ref{thm_welfare_D}, it is useful to compare the provision of common and idiosyncratic characteristics across different market structures.
Under monopoly, the characteristics profile $\bm{A}^\ddagger$ coincides with the social planner's choice $\bm{A}^\dagger$, while the output profile $\bm{q}^\ddagger$ is exactly one half of the socially optimal output $\bm{q}^\dagger$.
As a result, both common characteristics $\bm{x}^\ddagger=\bm{A}^\ddagger\bm{q}^\ddagger$ and idiosyncratic characteristics $\bm{y}^\ddagger=\bm{q}^\ddagger$ are provided at one half of their socially optimal levels, i.e., $\bm{\beta}$ and $\bm{\gamma}$, respectively.

By contrast, under oligopoly, Proposition~\ref{prop_eq_D} shows that the aggregate provision of common characteristics attains the socially optimal level, i.e., $\bm{x}^{\rm d}=\bm{A}^{\rm d}\bm{q}^{\rm d}=\bm{\beta}$.
However, this aggregate efficiency conceals a distortion in how common characteristics are produced at the firm level; recall that by Corollary~\ref{cor_cosine_diff}, the equilibrium characteristics profile $\bm{A}^{\rm d}$ is more concentrated toward $\bm{\beta}$ relative to $\bm{A}^\dagger$.
Geometrically, this means that $\bm{\beta}$ is spanned by characteristics vectors that are closer to one another rather than efficiently dispersed; see Figure~\ref{fig_duopoly_D}.
This concentration shortens the total path toward $\bm{\beta}$, meaning that $\bm{\beta}$ is achieved with smaller individual output levels.
Thus, although common characteristics are efficiently provided in aggregate, idiosyncratic characteristics $\bm{y}^{\rm d}=\bm{q}^{\rm d}$ are undersupplied in oligopoly, to a greater extent than under monopoly.

In summary, oligopoly achieves the efficient aggregate provision of common characteristics, but it does so through firms' exploitation of their standalone values, which leads to severe undersupply of idiosyncratic characteristics.
By contrast, monopoly centralizes production and thereby implements a more efficient pattern of product differentiation, despite underproducing in scale.
Theorem~\ref{thm_welfare_D} shows that, under the stated condition, the efficiency gain from centralization dominates the scale distortion, resulting in higher overall welfare under monopoly.

By contrast, the superiority of monopoly in characteristics design over oligopoly disappears when the equilibrium pattern of product differentiation coincides with the socially optimal one.
The next theorem shows that, when product concentration is socially optimal and arises in equilibrium, oligopoly can yield higher welfare than monopoly.

\begin{theorem} \label{thm_welfare_C}
Under condition~\eqref{gamma_conc}, we have $\Omega(\bm{A}^{\rm c},\bm{q}^{\rm c}) > \Omega(\bm{A}^\ddagger,\bm{q}^\ddagger)$ if either $\alpha$ is sufficiently large, $n$ is sufficiently large, or $\Var\qty[\bm{\gamma}]$ is sufficiently small.
\end{theorem}

The rationale behind this theorem is closely related to that of the standard Cournot model with homogeneous goods.
Indeed, under the product-concentration equilibrium, firms' products are highly similar, though not perfectly homogeneous due to the presence of idiosyncratic characteristics, and thus competition can benefit welfare by reducing firms' collective market power.
As such, each condition in Theorem~\ref{thm_welfare_C} describes a situation in which the environment approaches the homogeneous-products benchmark: a large $\alpha$ implies that product features are largely determined by common characteristics that are chosen identically across firms; a large $n$ intensifies competition; and a small $\Var\qty[\bm{\gamma}]$ implies that differences across products arising from idiosyncratic characteristics are relatively small.

In fact, Theorem~\ref{thm_welfare_C} corresponds to a special case of Theorem~\ref{thm_app} in Appendix~\ref{app_q}, where we conduct a welfare analysis in the baseline model of \cite{pellegrino2025} without characteristics design.
In that setting, social welfare is a quadratic function of the output profile $\bm{q}$ alone, with $\bm{b}_{\bm{A}}-\bm{c}$ and $\bm{\Sigma}_{\bm{A}}$ serving as the first- and second-order coefficients, respectively.
Economically, the former captures baseline demand for each firm's product net of its production cost, while the latter captures patterns of complementarity and substitutability across products.
It turns out that the (im)balance between these two primitives is decisive for the welfare ranking between monopoly and oligopoly.

To illustrate, suppose that firm~$i$'s product features high baseline demand, i.e., the $i$-th component of $\bm{b}_{\bm{A}}-\bm{c}$ is large.
At the same time, suppose that this product is a close substitute for firm~$j$'s product---though not necessarily a perfect substitute---so that the $(i,j)$-entry of $\bm{\Sigma}_{\bm{A}}$ is positive and relatively large.
Under oligopoly, firm~$i$ cares only about its own profit, and the high baseline demand encourages it to produce a large quantity, potentially crowding out demand for firm~$j$'s product.
The resulting allocation need not be socially desirable, since some characteristics that are specific to firm~$j$'s product may then be undersupplied.
By contrast, when a monopolist controls the provision of both products, it internalizes the substitutability between them.
Specifically, since higher output of product~$i$ cannibalizes demand for product~$j$, the monopolist has an incentive to adjust total supply so as to balance sales across the two products.
At the same time, however, the familiar welfare loss from underproduction arises due to the substantial market power granted to the monopolist.
While the overall welfare tradeoff remains ambiguous, we establish that a particular alignment between $\bm{b}_{\bm{A}}-\bm{c}$ and $\bm{\Sigma}_{\bm{A}}$ constitutes a sufficient condition under which oligopoly yields higher social welfare than monopoly.

From a theoretical perspective, this alignment is formalized through the principal-components decomposition of $\bm{\Sigma}_{\bm{A}}$ and the projection of $\bm{b}_{\bm{A}} - \bm{c}$ onto the subspace spanned by the ``major'' components.
As detailed in Appendix~\ref{app_q}, this spectral approach is closely related to recent work on optimal interventions in network games, such as \cite{galeottietal2020, galeottietal2025}, and is of independent theoretical interest.

\subsection{Welfare Comparison Across Equilibria}

Next, we analyze the welfare ranking across different types of equilibria.
The following theorem provides a sharp answer to this exercise by taking the product-differentiation equilibrium as the benchmark for comparison.

\begin{theorem} \label{thm_welfare_DCP}
Given any $\bm{\sigma} \in \{-1,1\}^n$, we have $\Omega(\bm{A}^{\rm d}, \bm{q}^{\rm d}) > \Omega(\bm{A}^{\bm{\sigma}}, \bm{q}^{\bm{\sigma}})$ if either of the following conditions is satisfied:
\[
\bm{\sigma}^\top \bm{\gamma} > 2+\alpha
\qquad \text{or} \qquad
\bm{\sigma}^\top \bm{\gamma}
<
\frac{n\alpha(2+\alpha)}{2(1+\alpha)(2+(n+1)\alpha)+n\alpha}.
\]
In particular:
\begin{enumerate}[\rm i).]
\item \label{thm_welfare_DC}
$\Omega(\bm{A}^{\rm d}, \bm{q}^{\rm d}) \ge \Omega(\bm{A}^{\rm c}, \bm{q}^{\rm c})$, with strict inequality unless $R(\bm{\gamma}) = 2+\alpha$.

\item \label{thm_welfare_DP}
$\Omega(\bm{A}^{\rm d}, \bm{q}^{\rm d}) > \Omega(\bm{A}^{\bm{\sigma}}, \bm{q}^{\bm{\sigma}})$ whenever $\bm{\sigma}^\top \bm{\gamma} \le 0$.
\end{enumerate}
\end{theorem}

This theorem provides sufficient conditions on the signed average of firms' standalone values that guarantees equilibrium social welfare is higher under product differentiation than under product concentration or polarization.

\begin{table}[t!]
\centering
\caption{In the symmetric case, the welfare levels under the social optimum $(\Omega^\dagger)$, monopoly $(\Omega^\ddagger)$, oligopoly with product differentiation $(\Omega^{\rm d})$, product concentration $(\Omega^{\rm c})$, and even product polarization $(\Omega^{\rm p})$ are reported.}
\label{table_omega_summary}
\renewcommand{\arraystretch}{1.66}
\begin{tabular}{|c|c|c|c|c|}
\hline
& $\gamma \le \tfrac{1}{n}$ 
& $\frac{1}{n} \le \gamma \le \tfrac{2+\alpha}{n}$ 
& $\frac{2+\alpha}{n} \le \gamma \le \frac{2(1+\alpha)}{n-1}$ 
& $\frac{2(1+\alpha)}{n-1} \le \gamma$ \\
\hline\hline
$\Omega^\dagger$
& $\tfrac{n(\alpha+\gamma)^2}{2(1+n\alpha)}$
& \multicolumn{3}{|c|}{$\tfrac{\alpha+n\gamma^2}{2}$} \\
\hline
$\Omega^\ddagger$
& $\tfrac{3n(\alpha+\gamma)^2}{8(1+n\alpha)}$
& \multicolumn{3}{|c|}{$\tfrac{3(\alpha+n\gamma^2)}{8}$} \\
\hline
$\Omega^{\rm d}$
& \multicolumn{2}{|c|}{NA}
& \multicolumn{2}{|c|}{$\tfrac{\alpha}{2} + \tfrac{n\gamma^2(3+2\alpha)}{2(2+\alpha)^2}$} \\
\hline
$\Omega^{\rm c}$
& \multicolumn{3}{|c|}{$\tfrac{n(\alpha+\gamma)^2(3+(n+2)\alpha)}{2(2+(n+1)\alpha)^2}$}
& NA \\
\hline
$\Omega^{\rm p}$
& \multicolumn{2}{|c|}{NA}
& \multicolumn{2}{|c|}{%
$\begin{cases}
\tfrac{n\gamma^2(3+2\alpha)}{2(2+\alpha)^2} + \tfrac{n\alpha^2(2+\alpha)(2+\alpha-n\alpha^2)}{2(2+(n+1)\alpha)^2}
&\text{if} \quad \gamma \ge \frac{2(1+\alpha)(2+\alpha)}{2+(n+1)\alpha} \\[5pt]
\text{NA}& \text{else}
\end{cases}$} \\
\hline
\end{tabular}
\end{table}

The first implication \eqref{thm_welfare_DC} is that equilibrium social welfare is always higher under product differentiation than under product concentration.
Note that this result does not contradict the welfare rankings relative to the monopolist benchmark established in Theorems~\ref{thm_welfare_D} and \ref{thm_welfare_C}, since those theorems apply to disjoint parameter regions.
Specifically, when the aggregate standalone value $R(\bm{\gamma})$ is relatively small, only the product-concentration equilibrium exists, and in this region, its welfare can exceed that under monopoly.
However, this equilibrium continues to exist even as $R(\bm{\gamma})$ increases and the socially optimal pattern of characteristics design switches to product differentiation.
When this occurs, equilibrium social welfare under product concentration becomes lower than under product differentiation, which in turn is generally lower than welfare under monopoly.

The second implication \eqref{thm_welfare_DP} pertains to the case of an ``inefficient'' product-polarization equilibrium, in which the sum of standalone values of the firms that dominate the provision of common characteristics is lower than that of the remaining firms.
There may be an intuitive sense that such an equilibrium is socially inefficient.
Indeed, firms with relatively large standalone values nevertheless choose characteristics that point in the opposite direction of the consumer's ideal vector $\bm{\beta}$ solely for the purpose of differentiating themselves from the others.
As shown in \eqref{eq_rank1}, this behavior results in reduced output levels for those firms with high standalone values.
Theorem~\ref{thm_welfare_DCP} confirms that social welfare in such polarization equilibria is necessarily lower than that under the product-differentiation equilibrium.

\setcounter{example}{1}
\begin{example}[continued]
In the symmetric case with $\bm{\gamma} = \gamma\1$, we can calculate the welfare level in each situation by substituting the output levels derived earlier into $\Omega$, albeit with tedious algebra.
Table~\ref{table_omega_summary} summarizes the results of these calculations.\footnote{The ordering of the cutoff values for product-concentration and product-polarization equilibria is ambiguous in general, so the two equilibria may or may not coexist depending on $n$ and $\alpha$.}

Focusing further on the case with $n=2$ and $\alpha=1$, Figure~\ref{fig_omega_duopoly} plots welfare levels under monopoly and duopoly relative to the social optimum.
When $\gamma$ is smaller than $0.5$---i.e., when \eqref{gamma_conc} is satisfied in the present case---only the product-concentration equilibrium exists.
Since $\Var\qty[\bm{\gamma}] = 0$, Theorem~\ref{thm_welfare_C} implies that welfare under oligopoly is strictly higher than welfare under monopoly throughout this parameter region.
Once $\gamma$ exceeds $1.5$---i.e., when \eqref{gamma_diff} becomes satisfied---the product-differentiation equilibrium emerges, and Theorem~\ref{thm_welfare_D} implies that welfare under this equilibrium is lower than under monopoly.
As $\gamma$ increases further and exceeds $2.4$, the product-polarization equilibrium emerges.
Among these equilibria, Theorem~\ref{thm_welfare_DCP} implies that the product-differentiation equilibrium yields the highest welfare, although it remains dominated by the monopolist benchmark.
\end{example}

\begin{figure}[t!]
\begin{center}
\includegraphics[width=4in]{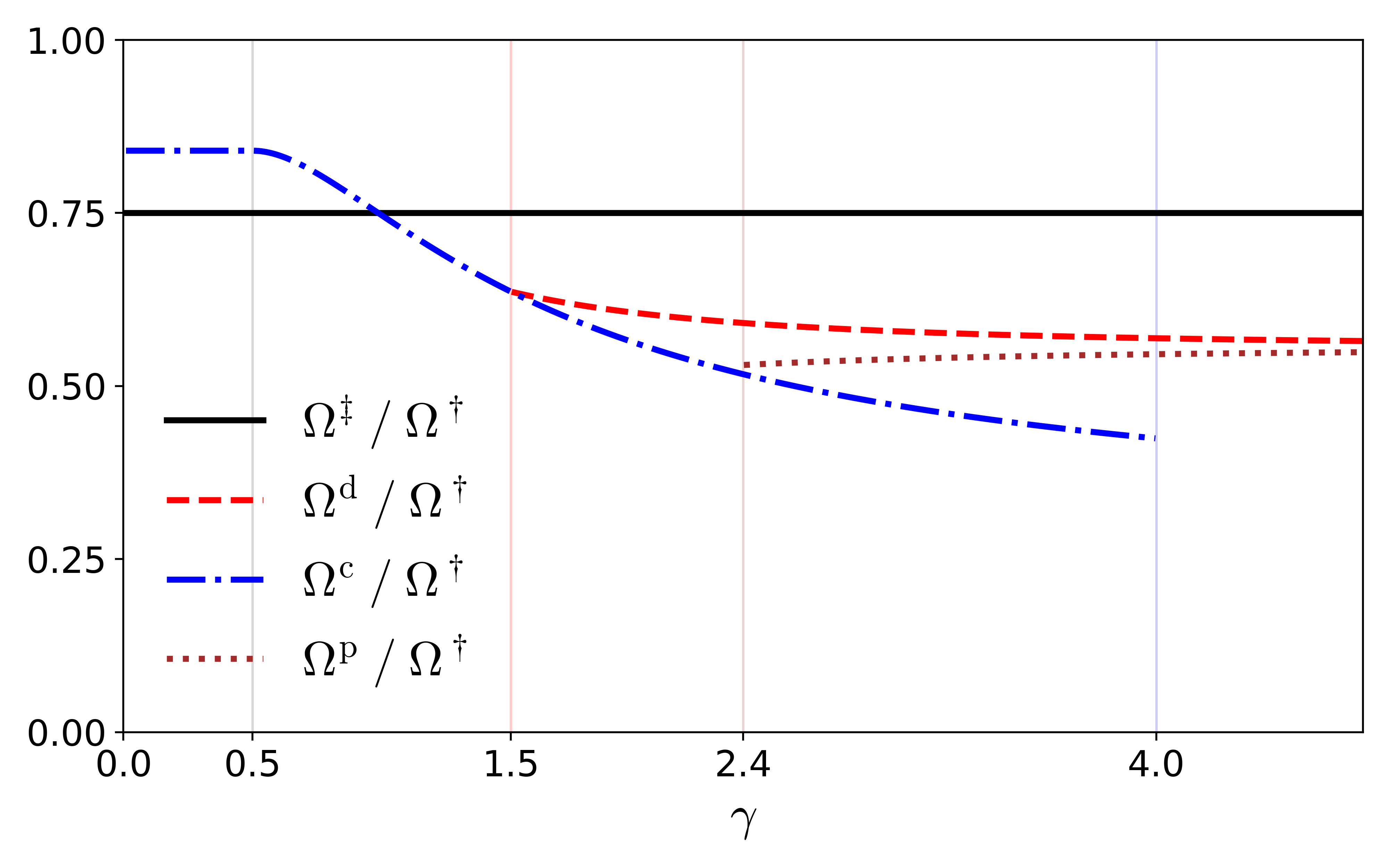}
\caption{In the symmetric case with $n=2$ and $\alpha=1$, welfare levels under monopoly and duopoly are plotted relative to the socially optimal level.}
\label{fig_omega_duopoly}
\end{center}
\end{figure}


\section{Extensions}
\label{sec_ext}

Now, we pursue two extensions of our baseline analysis.
To simplify exposition and to focus on the empirically most relevant case, throughout this section we restrict attention to environments in which product differentiation is socially optimal and arises in equilibrium.

\subsection{Network Effects}
Our baseline model implicitly assumes that there are no complementarity or substitutability across different characteristics.
This is reflected in the representative consumer's utility function, where the coefficients on all cross-quadratic terms are set to zero.

For common characteristics, this assumption is largely innocuous, since any non-orthogonality among common characteristics can be eliminated through an appropriate redefinition of the characteristics.
By contrast, idiosyncratic characteristics are exogenously given objects, and assuming orthogonality among them is potentially less innocuous.
Economically, it is natural to consider situations in which firms' brand values are interrelated, e.g., consumers may prefer to consume products from brands with similar styles in order to maintain consistency, rather than combining products from very different brands.
In such cases, the provision of a product by one firm may enhance or diminish the value of products provided by other firms.\footnote{\cite{singh1984vives} provide a canonical analysis of a differentiated duopoly in which the goods supplied by two firms deliver interrelated value to consumers.}

We capture such network effects among idiosyncratic characteristics by generalizing the consumer's utility function as follows:\footnote{
See Online Appendix~B of \cite{pellegrino2025} for an alternative extension.}
\[
U(\bm{x}, \bm{y}, H)
= \alpha \qty(\bm{x}^\top \bm{\beta} - \frac{1}{2}\bm{x}^\top \bm{x})
+ \qty(\bm{y}^\top \bm{b} - \frac{1}{2}\bm{y}^\top \qty(\bm{I}-\bm{W}) \bm{y})
- H,
\]
where $\bm{W} = [w_{ij}]_{n \times n}$ is a symmetric matrix with zero diagonal entries and spectral radius $\rho(\bm{W})$ strictly less than one.
The matrix $\bm{W}$ captures network effects across idiosyncratic characteristics: $w_{ij} > 0$ corresponds to preferences for joint consumption of products $i$ and $j$, whereas $w_{ij} < 0$ corresponds to dispreferences.

This generalization does not substantially alter either equilibrium or benchmark outputs.
Solving the consumer's problem yields the same demand structure as in the baseline model, with the matrix $\bm{\Sigma}_{\bm A}$ now given by $\bm{A}^\top \bm{A} + \bm{I}-\bm{W}$.
As a result, social welfare can be written as
\begin{equation} \label{eq_welfare_network}
\alpha \qty(\bm{\beta}^\top \bm{x} - \frac{1}{2}\bm{x}^\top \bm{x})
+ \bm{\gamma}^\top \bm{y}
- \frac{1}{2} \bm{y}^\top \qty(\bm{I}-\bm{W}) \bm{y}.
\end{equation}
Moreover, each firm's profit is given by
\begin{equation} \label{eq_profit_network}
\Pi_i \qty(\bm{a}_i, q_i; \bm{A}_{-i}, \bm{q}_{-i}) = \alpha q_i \bm{a}_i^\top \qty(\bm{\beta} - \sum_{j \neq i} q_j \bm{a}_j) - (1+\alpha) q_i^2 + \qty(\gamma_i + \sum_{j \neq i} w_{ij} q_j) q_i,
\end{equation}
which differs from \eqref{eq_profit_sep} only through the last term $\sum_{j \neq i} w_{ij} q_j$.
It follows that the best-response condition for $\bm{a}_i$ remains unchanged, and the optimal $q_i$ continues to be characterized by a first-order condition.
As a result, in any equilibrium exhibiting product differentiation, the provision of common characteristics achieves the first best, i.e., $\bm{A}^{\rm d} \bm{q}^{\rm d} = \bm{\beta}$.
Moreover, the firm's output choice aligns with the equilibrium problem in canonical network games \cite{bcz2006}, leading to a centrality-based characterization.
Following \cite{bonacich1987}, we define the vector of \emph{weighted Bonacich centralities} as
\[
{\bf b}(\delta,\bm{z}) \coloneqq (\bm{I}-\delta \bm{W})^{-1} \bm{z},
\]
where $\delta \in \R$ is a decay parameter satisfying $\delta\rho(\bm{W}) < 1$ and $\bm{z} \in \R^n$ is a given vector.

\begin{proposition} \label{prop_network}
In the presence of network effects among idiosyncratic characteristics, suppose that parameter values are such that product differentiation is socially optimal and arises in equilibrium.
Then, the social planner's, monopolist's, and oligopoly equilibrium output profiles are given by
\[
\bm{q}^\dagger = {\bf b}(1,\bm{\gamma}), \qquad
\bm{q}^\ddagger = {\bf b}\qty(\textstyle 1,\frac{\bm{\gamma}}{2}), \qquad
\bm{q}^{\rm d} = {\bf b}\qty(\textstyle \frac{1}{2+\alpha},\frac{\bm{\gamma}}{2+\alpha}).
\]
In particular, if $w_{ij} \ge 0$ for all $i \neq j$, then $\bm{q}^\ddagger \gg \bm{q}^{\rm d}$.
\end{proposition}

Each output profile admits a standard network interpretation.
Specifically, the social planner's allocation coincides with the vector of Bonacich centralities of the network $\bm{W}$ with unit decay.
Thus, each firm's output reflects not only its standalone value but also others through network effects, and firms are encouraged or discouraged to produce relative to their standalone benchmarks depending on the signs of $\bm{W}$.\footnote{Since $\bm{q}^\dagger - \bm{\gamma} = \bm{W}\bm{q}^\dagger$ by the definition of Bonacich centrraity, a firm produces more (resp.\ less) than its standalone value when its idiosyncratic characteristic is sufficiently complementary (resp.\ substitutable) with those of other firms.}
Monopoly uniformly scales down this allocation, while oligopoly outputs are both scaled down and associated with a smaller decay parameter, implying faster decay of network effects.
As in the baseline model, the scaling effect reflects underproduction due to market power, whereas the smaller decay parameter under oligopoly captures the fact that firms place relatively greater weight on direct interactions and less weight on indirect interactions that are internalized by the social planner or monopolist.
The overall implication for output levels is particularly transparent when the network encodes only nonnegative effects, in which case $\bm{q}^\ddagger \gg \bm{q}^{\rm d}$ continues to hold.

Under oligopoly, welfare derived from common characteristics continues to attain the first-best level, since the aggregate provision $\bm{A}^{\rm d}\bm{q}^{\rm d}$ remains equal to $\bm{\beta}$.
By contrast, when $\bm{W}$ is nonnegative, welfare derived from idiosyncratic characteristics tends to favor monopoly due to the above-mentioned output ranking.
More generally, the welfare comparison in the idiosyncratic dimension is nested within the welfare analysis without characteristics design developed in Appendix~\ref{app_q}.
The key takeaway is that the relative performance of monopoly and oligopoly depends on the alignment between the network structure $\bm{W}$ and the distribution of standalone values $\bm{\gamma}$.


\subsection{Common Ownership}

Our second extension concerns common ownership, which refers to an ownership structure in which large investors hold significant stakes in multiple firms, and firms act to maximize shareholder returns rather than their own standalone profits.
Such ownership structures often exhibit substantial overlap across competing firms, which creates firms' incentives to compete less aggressively, e.g., by reducing output or raising prices.\footnote{\cite{rotemberg1984} analyzes oligopolistic behavior under common ownership.
\cite{schmalz2018} provides a recent survey of the empirical literature on the anticompetitive effects of common ownership.}

When endogenous product differentiation is considered, however, the overall impact of common ownership on output levels and welfare is a priori ambiguous, since output reduction is not the only strategy firms can use to mitigate competition.
Rather, firms may instead choose a greater degree of product differentiation as a means of softening competition, which in turn raises markups and may strengthen their incentives to produce.
The net effect of common ownership on equilibrium outcomes therefore depends on the interaction between product-design and output-level channels.

We incorporate common ownership into our baseline model by adopting the modeling strategy of \cite{ederer2025pellegrino}.
Specifically, we modify each firm~$i$'s objective function as follows:\footnote{This objective function coincides with that in \cite{ederer2025pellegrino}. It is derived as the reduced form of a weighted aggregate of income accrued by investment funds holding shares in firm~$i$, with weights given by the ownership shares of each fund.}
\[
\tilde{\Pi}_i \coloneqq \Pi_i + \sum_{j \neq i} \kappa_{ij} \Pi_j,
\]
where $\Pi_i$ denotes firm~$i$'s standalone profit from \eqref{eq_profit_sep}, and $\kappa_{ij} \ge 0$ is the common-ownership weight that firm~$i$ assigns to the profit of firm~$j$ relative to its own profit.
Let $\bm{K} = [\kappa_{ij}]_{n \times n}$ be the matrix consisting of these weights, where $\kappa_{ii}=1$ for all $i \in [n]$.
We call $\bm{K}$ an \emph{ownership structure}, and assume that it is an arbitrary square matrix with nonnegative entries, unit diagonal entries, and such that $\frac{\bm{K}+\bm{K}^\top}{2}$ is positive semidefinite.\footnote{All these assumptions are satisfied if $\bm{K}$ is derived through the microfoundation in \cite{ederer2025pellegrino}.}

Substituting $\Pi_i$ from \eqref{eq_profit_sep} into $\tilde{\Pi}_i$, straightforward algebra yields
\begin{equation} \label{eq_profit_owner}
\tilde{\Pi}_i(\bm{a}_i, \bm{q}_i; \bm{A}_{-i}, \bm{q}_{-i}) = \alpha q_i \bm{a}_i^\top \qty(\bm{\beta} - \sum_{j \neq i} \qty(1+\kappa_{ij})q_j \bm{a}_j) - (1+\alpha)q_i^2 + \gamma_i q_i + F_{-i},
\end{equation}
where $F_{-i}$ collects all terms that do not depend on firm~$i$'s choice variables.
Observe that \eqref{eq_profit_owner} differs from \eqref{eq_profit_sep} in that the competitive pressure in characteristics design is now amplified by the ownership weights $\kappa_{ij}$.
The next lemma characterizes firms' best responses and the resulting equilibrium under common ownership.

\begin{lemma} \label{lem_br_owner}
Under common ownership, given any profile $(\bm{A}_{-i}, \bm{q}_{-i})$ of opponents' strategies, firm $i$'s best-response strategy $(\bm{a}_i,q_i)$ is characterized as follows:
\[
\delta_i \bm{a}_i = \bm{\beta} - \sum_{j \neq i} \qty(1+\kappa_{ij})q_j \bm{a}_j, \quad
q_i = \frac{\alpha \delta_i + \gamma_i}{2(1+\alpha)}, \quad \text{where} \quad \delta_i = \left\|\bm{\beta} - \textstyle \sum_{j \neq i} \qty(1+ \kappa_{ij})q_j \bm{a}_j \right\|_2.
\]
\end{lemma}

The key difference from Lemma~\ref{lem_br} is that the effective residual demand for common characteristics available to firm~$i$, captured by $\delta_i$, now incorporates the effects of common ownership.
As a result, equilibrium product differentiation can be closer to the social optimum, which in turn leads to higher equilibrium output levels.
However, the extent to which common ownership can affect equilibrium outcomes is limited, since each firm continues to retain market power even when it internalizes the profits of other firms.
As an application of Lemma~\ref{lem_br_owner}, the next corollary shows that attaining the first-best outcome remains impossible, even when the ownership structure can be flexibly designed.

\begin{corollary} \label{cor_owner_sp}
There exists no ownership structure under which the equilibrium outcome attains the first-best outcome $(\bm{x}^\dagger, \bm{y}^\dagger) = (\bm{\beta}, \bm{\gamma})$.
\end{corollary}

An explicit derivation of equilibrium becomes more challenging under common ownership.
This is because the right-hand side of the best-response condition for $\bm{a}_i$ is no longer identical across firms, which prevents us from deriving an equation analogous to \eqref{eq_delta} that relates the common characteristics choices of any two firms.
To overcome this difficulty, let us focus on a special case in which ownership weights are identical across all pairs of distinct firms.
Needless to say, this assumption is hardly justifiable from an empirical perspective and is imposed solely for tractability.\footnote{As documented in \cite{ederer2025pellegrino}, ownership networks typically exhibit a hub-and-spoke structure, in which a subset of firms share significant ownership links while others remain largely unconnected at the periphery, indicating a high degree of asymmetry across firms.}

\begin{proposition} \label{prop_owner}
Under common ownership with $\kappa_{ij} = \kappa$ for all $i \neq j$, any equilibrium $(\bm{A}^{\rm d}, \bm{q}^{\rm d})$ that exhibits product differentiation is characterized as follows:
\begin{equation} \label{profile_owner}
\bm{A}^{\rm d}\bm{q}^{\rm d} = \frac{\bm{\beta}}{1+\kappa}, \qquad
\bm{q}^{\rm d} = \frac{\bm{\gamma}}{2+\alpha(1-\kappa)}, \quad \text{thus} \quad \bar{s}_{\bm \gamma}(\bm{A}^{\rm d}) = \frac{\qty(\frac{2+\alpha(1-\kappa)}{1+\kappa})^2 - \|\bm{\gamma}\|_2^2}{n(n-1)}.
\end{equation}
This equilibrium exists if and only if
\begin{equation} \label{cond_owner}
r(\bm{\gamma}) \le \frac{2+\alpha(1-\kappa)}{1+\kappa} \le R(\bm{\gamma}).
\end{equation}
\end{proposition}

Proposition~\ref{prop_owner} uncovers the essential interaction between endogenous product differentiation and common ownership.
Specifically, the equilibrium degree of product differentiation, as measured by the weighted average cosine similarity $\bar{s}_{\bm{\gamma}}(\bm{A}^{\rm d})$, is monotonically decreasing in the degree of common ownership.
This result formalizes the intuitive mechanism discussed earlier in this section; namely, common ownership strengthens firms' incentives to soften competition, and firms can do so by differentiating their product characteristics.

The effect of common ownership on output is a priori ambiguous at an intuitive level, since opposing forces act on firms' output choices.
On the one hand, as in standard settings without characteristics design, firms have incentives to accommodate other firms with overlapping ownership by reducing their own output.
On the other hand, the higher degree of product differentiation induced by common ownership allows firms to sustain higher markups, which in turn strengthens their incentives to produce.
Proposition~\ref{prop_owner} shows that, in the present model, the latter effect dominates, as the equilibrium output $\bm{q}^{\rm d}$ in \eqref{profile_owner} is monotonically increasing in $\kappa$.
This stands in sharp contrast to settings with exogenous product characteristics, in which firms can soften competition only by reducing output.

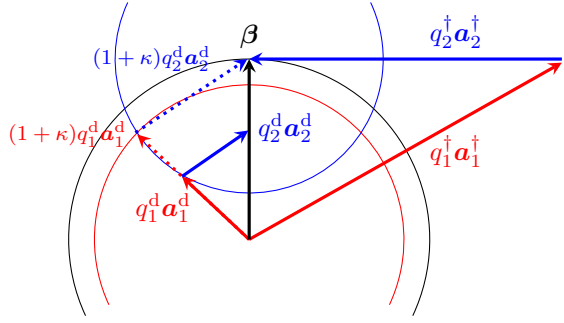
\begin{figure}[t]
\centering
\begin{tikzpicture}[>=stealth, scale=0.8]
\usetikzlibrary{calc,intersections}

\def\R{3}
\def\alpha{1}
\def\k{2/3}
\def\gOne{2}
\def\gTwo{sqrt(3)}
\pgfmathsetmacro{\Rred}{\gOne/(2+\alpha*(1-\k)) * \R}
\pgfmathsetmacro{\Rblue}{\gTwo/(2+\alpha*(1-\k)) * \R}
\def\angUpA{-25}
\def\angUpB{205}
\def\angBlueA{155}
\def\angBlueB{385}

\useasboundingbox (-4.2,-3.8) rectangle (4.2,4.2);

\coordinate (O) at (0,0);
\coordinate (B) at (0,\R);
\coordinate (P) at ({sqrt(3)*\R},\R);

\draw (O) ++(\angUpA:\R)
  arc[start angle=\angUpA, end angle=\angUpB, radius=\R];
\draw[red] (O) ++(\angUpA:\Rred)
  arc[start angle=\angUpA, end angle=\angUpB, radius=\Rred];
\draw[blue] (B) ++(\angBlueA:\Rblue)
  arc[start angle=\angBlueA, end angle=\angBlueB, radius=\Rblue];

\path[name path=rc] (O) circle (\Rred);
\path[name path=bc] (B) circle (\Rblue);
\path[name intersections={of=rc and bc, by={I1,I2}}];

\makeatletter
\newdimen\Xone \newdimen\Xtwo
\pgfextractx{\Xone}{\pgfpointanchor{I1}{center}}
\pgfextractx{\Xtwo}{\pgfpointanchor{I2}{center}}
\ifdim\Xone<\Xtwo \coordinate (I) at (I1); \else \coordinate (I) at (I2); \fi
\makeatother

\coordinate (Ired)  at ($(I)+(0,-0.03)$);
\coordinate (Iblue) at ($(I)-(0.03,0.0)$);
\coordinate (Bblue) at ($(B)-(0.03,0)$);
\coordinate (Pred)  at ($(P)-(0,0.06)$);

\draw[->, very thick, red] (O) -- (Pred) node[midway, right, xshift=5pt, font=\small] {$q_1^\dagger \bm{a}_1^\dagger$};
\draw[->, very thick, blue] (P) -- (B) node[midway, above right, xshift=5pt, font=\small] {$q_2^\dagger \bm{a}_2^\dagger$};
\draw[->, very thick, red, dotted] (O) -- (Ired) node[midway, above left, font=\scriptsize, xshift=-21pt, yshift=11pt] {$(1+\kappa)q_1^{\rm d} \bm{a}_1^{\rm d}$};
\draw[->, very thick, blue, dotted] (Iblue) -- (Bblue) node[midway, above, font=\scriptsize, xshift=-15pt, yshift=5pt] {$(1+\kappa)q_2^{\rm d} \bm{a}_2^{\rm d}$};
  
\coordinate (Ired2) at ($ (O)!{1/(1+\k)}!(Ired) $);
\draw[->, very thick, red] (O) -- (Ired2) node[midway, left, font=\small, xshift=-5pt, yshift=0pt] {$q_1^{\rm d} \bm{a}_1^{\rm d}$};

\coordinate (Iblue2) at ($ (Iblue)!{1/(1+\k)}!(B) $);
\coordinate (Iblue3) at ($ (Ired2)+(Iblue2)-(Ired)$);
\draw[->, very thick, blue] (Ired2) -- (Iblue3) node[midway, above right, font=\small, xshift=12pt, yshift=0pt] {$q_2^{\rm d} \bm{a}_2^{\rm d}$};

\draw[->, very thick]
  (O) -- (B) node[above, font=\small] {$\bm{\beta}$};

\end{tikzpicture}
\vspace{-5em}
\caption{Socially optimal and equilibrium product differentiation under common ownership with $m=n=2$, $\alpha=1$, $\kappa=\frac{2}{3}$, $\bm{\beta}=(0,1)$, $\bm{q}^\dagger = \bm{\gamma} = (2,\sqrt{3})$, and $\bm{q}^{\rm d} = \frac{\bm{\gamma}}{2+\alpha(1-\kappa)} = (\frac{6}{7},\frac{3\sqrt{3}}{7})$.}
\label{fig_duopoly_owner}
\end{figure}

\setcounter{example}{0}
\begin{example}[continued]
When symmetric common ownership is introduced into the duopoly with $\alpha = 1$ and $\bm{\gamma} = (2,\sqrt{3})$, the equilibrium is computed as
$\bm{q}^{\rm d} = \bigl(\frac{2}{3-\kappa}, \frac{\sqrt{3}}{3-\kappa}\bigr)$
and
$\bm{a}_1^\top \bm{a}_2 = \frac{1 - 10\kappa - 3\kappa^2}{2\sqrt{3}(1+\kappa)^2}$.
Figure~\ref{fig_duopoly_owner} provides a geometric illustration of the resulting structure of product differentiation for the case $\kappa = \frac{2}{3}$.
Observe that $\bm{a}_1^{\rm d}$ and $\bm{a}_2^{\rm d}$ together span the consumer's ideal vector $\bm{\beta}$ once their lengths are multiplied by $(\kappa + q_i^{\rm d})$.
Without the amplification factor $\kappa$, however, the actual equilibrium provision of common characteristics, given by $q_1^{\rm d} \bm{a}_1^{\rm d} + q_2^{\rm d} \bm{a}_2^{\rm d}$, falls short of reaching $\bm{\beta}$.
This shortfall can be interpreted economically as the standard output-reduction effect of common ownership.
At the same time, firms' choices of common characteristics become more diversified for the purpose of softening competition.
As a result, the eventual output levels are larger than their counterparts in the benchmark equilibrium without common ownership (i.e., $\kappa = 0$).
This implies that the total provision of idiosyncratic characteristics by each firm is enhanced under common ownership.
\end{example}

As illustrated in the example above, the welfare implications of common ownership are mixed.
Namely, while the increase in firms' output levels can contribute positively to social welfare through enhanced provision of idiosyncratic characteristics, the distortion in product differentiation may result in an undersupply of common characteristics relative to the ideal level that would be achieved without common ownership.
The next result provides a sufficient condition under which the overall welfare effect of common ownership is positive.

\begin{corollary} \label{cor_owner_welfare}
Social welfare under the equilibrium characterized in Proposition~\ref{prop_owner} is strictly increasing at $\kappa \in [0,1]$ if and only if
\begin{equation} \label{cond_owner_welfare}
\|\bm{\gamma}\|_2 > f(\alpha,\kappa) \cdot
\frac{2+\alpha(1-\kappa)}{1+\kappa}, \quad \text{where} \quad
f(\alpha,\kappa) = \sqrt{\frac{\kappa(2+\alpha(1-\kappa))}{(1+\kappa)(1+\alpha(1-\kappa))}} \in [0,1].
\end{equation}
\end{corollary}

This result indicates that, once endogenous characteristics design is taken into account, common ownership can have a positive welfare effect.
Specifically, for reasons similar to those in Theorem~\ref{thm_welfare_D}, we can regard condition~\eqref{cond_owner_welfare} as fairly weak.
Indeed, for a product-differentiation equilibrium to exist under common ownership, condition~\eqref{cond_owner} must already hold, which requires the $\ell_1$ norm of $\bm{\gamma}$ to exceed the threshold $\frac{2+\alpha(1-\kappa)}{1+\kappa}$.
Although condition~\eqref{cond_owner_welfare} instead imposes a restriction on the $\ell_2$ norm of $\bm{\gamma}$, it reflects a similar economic requirement that firms' standalone values be sufficiently strong in the aggregate.
Moreover, the threshold for $\|\bm{\gamma}\|_2$ is multiplicatively discounted by the factor $f(\alpha,\kappa)\in[0,1]$, making the condition even less demanding.

\begin{figure}[t!]
\begin{center}
\includegraphics[width=4in]{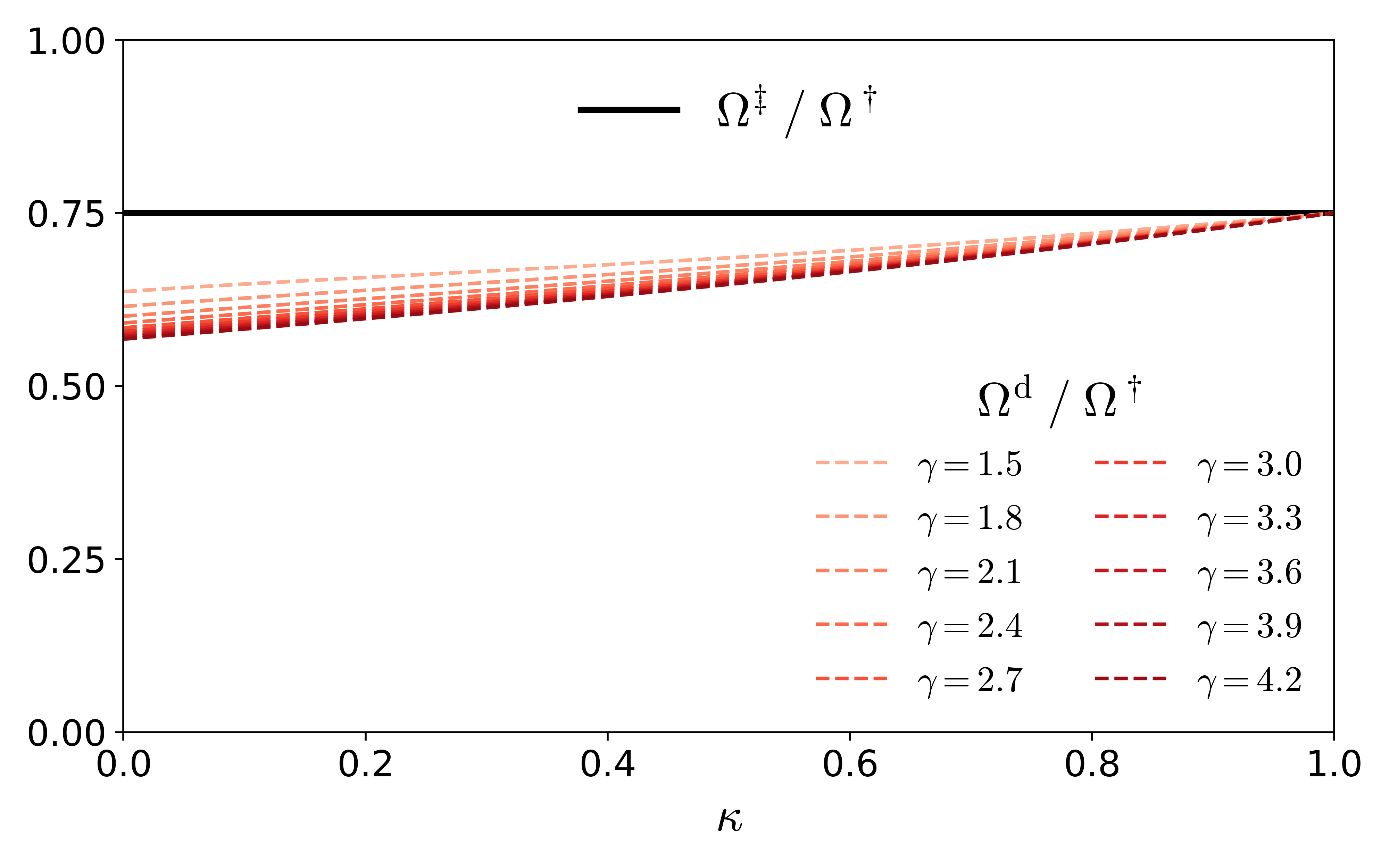}
\caption{In the symmetric case with $n=2$ and $\alpha=1$, welfare levels under common ownership are plotted relative to the socially optimal level.}
\label{fig_omega_owner}
\end{center}
\end{figure}

Figure~\ref{fig_omega_owner} plots welfare levels under the symmetric duopoly from Example~\ref{ex_sym} in the presence of common ownership.
Across different values of $\bm{\gamma}$, equilibrium welfare relative to the social planner's benchmark exhibits a similar pattern, increasing monotonically with $\kappa$.

Lastly, it should be noted that our analysis abstracts from several other channels through which common ownership may affect market outcomes and welfare, and overall policy implications may also depend on those; see Section~6.7 of \cite{ederer2025pellegrino} for related discussion.
For example, those other effects are studied theoretically by \cite{azar2021vives} in the context of labor market power and by \cite{lopez2019vives} in the context of R\&D investment with spillovers.
Moreover, in practice, the design of product characteristics may involve adjustment frictions and longer time horizons than those for quantity or price adjustments.
Recent work by \cite{hopenhayn2025okumura} develops a dynamic version of the GHL model, and \cite{okumura2025} studies the effect of common ownership on economic growth.
In summary, this paper provides the first theoretical framework that fully endogenizes product differentiation, and there remains ample scope for future research to extend and generalize our methodology in order to assess the broader interaction between product differentiation and other economic channels.


\appendix

\section{Omitted Proofs}

\subsection*{Proof of Lemma~\ref{lem_donut}}
Suppose that $\bm{x} = \bm{A}\bm{q}$ for some $\bm{A} \in \calA$.
By the triangle inequality,
\[
\|\bm{x}\|_2 = \|\bm{A}\bm{q}\|_2 = \left\| \sum_{i=1}^n q_i \bm{a}_i \right\|_2
\le \sum_{i=1}^n q_i \|\bm{a}_i\|_2 = \sum_{i=1}^n q_i = \|\bm{q}\|_1 = R(\bm{q}).
\]
Let $q_1 = \max_{i \in [n]} q_i = \|\bm{q}\|_\infty$ without loss of generality.
Again, by the triangle inequality,
\[
\|\bm{x}\|_2 = \left\|q_1 \bm{a}_1 + \sum_{i=2}^n q_i \bm{a}_i \right\|_2
\ge \qty| q_1 \|\bm{a}_1\|_2 - \left\|\sum_{i=2}^n q_i \bm{a}_i \right\|_2 |
\ge \|\bm{q}\|_{\infty} - \left\|\sum_{i=2}^n q_i \bm{a}_i \right\|_2.
\]
Moreover,
\[
\left\|\sum_{i=2}^n q_i \bm{a}_i \right\|_2 \le \sum_{i=2}^n q_i\|\bm{a}\|_2 = \sum_{i=2}^n q_i = \sum_{i=1}^n q_i - q_1 = \|\bm{q}\|_1 - \|\bm{q}\|_\infty.
\]
Combining the last two equations, it follows that
\[
\|\bm{x}\|_2 \ge \|\bm{q}\|_{\infty} - \qty(\|\bm{q}\|_1 - \|\bm{q}\|_\infty) = 2 \|\bm{q}\|_\infty - \|\bm{q}\|_1 = r(\bm{q}).
\]
Thus, we have proved the ``only if'' part.

For the ``only if'' part, we first consider the case $\bm{x}=\bm{0}$.
Note that $\bm{0}$ satisfies \eqref{feasible_x} if and only if $\|\bm{q}\|_1 \ge 2\|\bm{q}\|_\infty$, or
\begin{equation} \label{polygon0}
q_i \le \sum_{j \neq i} q_j, \quad \forall i \in [n].
\end{equation}
If $n=2$, this condition reduces to $q_1=q_2$, so $q_1\bm{a}_1 + q_2\bm{a}_2 = \bm{0}$ is achieved by simply letting $\bm{a}_1 = -\bm{a}_2$.
Henceforth, we consider the case $n \ge 3$ and assume $\bm{q} \gg \bm{0}$ without loss of generality.

\begin{lemma} \label{lem_triangle1}
If $n \ge 3$ and $\bm{q} \gg \bm{0}$ satisfies \eqref{polygon0}, then $[n]$ can be partitioned into three nonempty subsets $N_1$, $N_2$, and $N_3$ such that
\begin{equation} \label{triangle_Q}
\sum_{i \in N_k} q_i \le \sum_{i \in [n] \setminus N_k} q_i, \quad \forall k \in \{1,2,3\}.
\end{equation}
\end{lemma}

\begin{proof}
Given any partition $(N_1,N_2,N_3)$ of $[n]$, let $Q_k = \sum_{i \in N_k} q_i$ for each $k \in \{1,2,3\}$.
Assume that $Q_1 \ge Q_2 \ge Q_3 > 0$ without loss of generality.
If $Q_1 \le Q_2 + Q_3$, then \eqref{triangle_Q} holds for $k=1$, and it trivially holds for $k\in\{2,3\}$ as well.
Hence we are done.

Suppose instead that $Q_1 > Q_2 + Q_3$.
Note that $Q_1$ cannot be a singleton, since otherwise, \eqref{polygon0} would be violated.
So, we pick $i$ from $N_1$ such that $q_i = \min_{j \in N_1} q_j$---and thus, $\tilde{Q}_1 = Q_1 - q_i \ge q_i$ holds---and move it to $N_3$ to form a new partition $(\tilde{N}_1,\tilde{N}_2,\tilde{N}_3)$.
It then follows that
\begin{align*}
&\tilde{Q}_1 + \tilde{Q}_3 = \qty(Q_1 - q_i) + (Q_3 + q_i) = Q_1 + Q_3 > Q_1 > Q_2 + Q_3 > Q_2 = \tilde{Q}_2, \\
&\tilde{Q}_1 + \tilde{Q}_2 = \qty(Q_1 - q_i) + Q_2 \ge q_i + Q_2 \ge q_i + Q_3 = \tilde{Q}_3.
\end{align*}
Thus, if $\tilde{Q}_1 \le \tilde{Q}_2 + \tilde{Q}_3$, then $(\tilde{N}_1,\tilde{N}_2,\tilde{N}_3)$ satisfies \eqref{triangle_Q}, and we are done.
Otherwise, we repeat the same procedure by moving another $j$ from $\tilde{N}_1$ to either $\tilde{N}_2$ or $\tilde{N}_3$, whichever has the smaller of $\tilde{Q}_2$ and $\tilde{Q}_3$.
Iterating this process, we must eventually obtain a partition satisfying \eqref{triangle_Q}, since otherwise, the first set would be a singleton, leading to a violation of \eqref{polygon0}.
\end{proof}

\begin{lemma} \label{lem_triangle2}
For any scalars $Q_1,Q_2,Q_3 > 0$, if $Q_k \le Q_{k'} + Q_{k''}$ for any distinct $k,k',k'' \in \{1,2,3\}$, there exist unit vectors $\bm{u}_1,\bm{u}_2,\bm{u}_3$ such that $Q_1 \bm{u}_1 + Q_2 \bm{u}_2 + Q_3 \bm{u}_3 = \bm{0}$.
\end{lemma}

\begin{proof}
There exists some $\theta \in [0,2\pi]$ such that
\[
\cos\qty(\theta) = \frac{Q_1^2 + Q_2^2 - Q_3^2}{2Q_1Q_2},
\]
since the value of the right-hand side lies in $[-1,1]$ under the stated condition.
Then, define vectors $\bm{v}_1, \bm{v}_2 \in \R^m$ by
\[
\bm{v}_1 = \mqty[Q_1 \\ 0 \\ \bm{0}], \quad
\bm{v}_2 = \mqty[Q_2 \cos\qty(\theta) \\ Q_2 \sin\qty(\theta) \\ \bm{0}].
\]
By construction, we have
\[
\|\bm{v}_1\|_2 = Q_1, \qquad
\|\bm{v}_2\|_2 = \sqrt{Q_2^2 \qty(\cos^2\qty(\theta) + \sin^2\qty(\theta))} = Q_2.
\]
Moreover, it holds that
\begin{align*}
\|\bm{v}_2 - \bm{v}_1\|_2 &= \sqrt{\qty(Q_2\cos\qty(\theta) - Q_1)^2 + Q_2^2 \sin^2\qty(\theta)} \\
&= \sqrt{Q_1^2 + Q_2^2 \qty(\cos^2\qty(\theta) + \sin^2\qty(\theta)) - 2Q_1Q_2 \cos\qty(\theta)}
= Q_3.
\end{align*}
Therefore, by letting $\bm{u}_1 = \frac{\bm{v}_1}{\|\bm{v}_1\|_2}$, $\bm{u}_2 = -\frac{\bm{v}_2}{\|\bm{v}_2\|_2}$, and $\bm{u}_3 = \frac{\bm{v}_2-\bm{v}_1}{\|\bm{v}_2-\bm{v}_1\|_2}$, we obtain
\[
Q_1 \bm{u}_1 + Q_2 \bm{u}_2 + Q_3 \bm{u}_3
= \bm{v}_1 - \bm{v}_2 + \qty(\bm{v}_2 - \bm{v}_1) = \bm{0}. \qedhere
\]
\end{proof}

Hence, by Lemmas~\ref{lem_triangle1} and~\ref{lem_triangle2}, if $n \ge 3$ and $\bm{q} \gg \bm{0}$ satisfies~\eqref{polygon0}, there exist a partition $(N_1,N_2,N_3)$ of $[n]$ and associated unit vectors $(\bm{u}_1,\bm{u}_2,\bm{u}_3)$ such that
\[
\sum_{i=1}^n q_i \bm{a}_i = Q_1 \bm{u}_1 + Q_2 \bm{u}_2 + Q_3 \bm{u}_3 = \bm{0}, 
\]
where we let $\bm{a}_i = \bm{u}_k$ for each $i \in N_k$ and $k \in \{1,2,3\}$.
This concludes the proof for the case $\bm{x} = \bm{0}$.

Now, take any $\bm{x} \in \R^m \setminus \{\bm{0}\}$ that satisfies \eqref{feasible_x}.
Consider a positive vector of length $n+1$ defined by
$\tilde{\bm{q}} = [q_1,\, \ldots,\, q_n,\, \|\bm{x}\|_2]^\top$.
By construction and \eqref{feasible_x}, we have
\begin{align*}
\|\tilde{\bm{q}}\|_1 - 2 \|\tilde{\bm{q}}\|_\infty
&= \|\bm{q}\|_1 + \|\bm{x}\|_2 - 2\max\qty{\|\bm{q}\|_\infty,\, \|\bm{x}\|_2} \\
&= \min\qty{\|\bm{q}\|_1 + \|\bm{x}\|_2 - 2\|\bm{q}\|_\infty,\, \|\bm{q}\|_1 - \|\bm{x}\|_2} \ge 0,
\end{align*}
meaning that $\tilde{\bm{q}}$ satisfies \eqref{polygon0}.
Hence, by the previous conclusion, there exist unit vectors $\tilde{\bm{a}}_1, \ldots, \tilde{\bm{a}}_n$, and $\tilde{\bm{a}}_{n+1}$ such that
\[
\bm{0}
= \sum_{i=1}^{n+1} \tilde{q}_i \tilde{\bm{a}}_i
= \sum_{i=1}^n q_i \tilde{\bm{a}}_i + \|\bm{x}\|_2 \cdot \tilde{\bm{a}}_{n+1}.
\]
By letting $\tilde{\bm{x}} = - \|\bm{x}\|_2 \cdot \tilde{\bm{a}}_{n+1}$,
\begin{equation} \label{eq_donut_conc}
\sum_{i=1}^n q_i \tilde{\bm{a}}_i = \tilde{\bm{x}}.
\end{equation}
We are done if $\tilde{\bm{x}} = \bm{x}$.
Otherwise, let $\bm{z} = \tilde{\bm{x}} - \bm{x}$, and define the associated Householder matrix
$\bm{H} = \bm{I} - 2\frac{\bm{z}\bm{z}^\top}{\bm{z}^\top \bm{z}}$.
Noting that $\|\tilde{\bm{x}}\|_2 = \|\bm{x}\|_2$, the matrix $\bm{H}$ satisfies $\bm{H}\tilde{\bm{x}} = \bm{x}$, and $\bm{H}$ is orthogonal so that it preserves norms, i.e., $\|\bm{H}\bm{v}\|_2 = \|\bm{v}\|_2$ for any $\bm{v} \in \R^m$; see Chapter~2.1 of \cite{horn}.
Hence, by multiplying both sides of \eqref{eq_donut_conc} by $\bm{H}$ from the left, we obtain unit vectors $\bm{a}_i = \bm{H}\tilde{\bm{a}}_i$ such that $\sum_{i=1}^n q_i \bm{a}_i = \bm{x}$.
\hfill {\it Q.E.D.}


\subsection*{Proof of Lemma~\ref{lem_opt_A}}

This lemma follows immediately from Lemma~\ref{lem_donut}.
\hfill {\it Q.E.D.}


\subsection*{Proof of Lemma~\ref{lem_R}}

The ``if'' part is straightforward.
To prove the ``only if'' part, suppose that $\bm{A}\bm{q}=\bm{x}$ and $\|\bm{x}\|_2 = R(\bm{q})$.
This implies that $\sum_{i=1}^n q_i \bm{a}_i = \sum_{i=1}^n q_i \cdot \frac{\bm{x}}{\|\bm{x}\|_2}$.
Calculating the squared $\ell_2$ norm of each side, we get $\|\sum_{i=1}^n q_i \bm{a}_i \|_2^2 = \sum_{i,j \in [n]} q_i q_j \bm{a}_i^\top \bm{a}_j$ and $\|\sum_{i=1}^n q_i \cdot \frac{\bm{x}}{\|\bm{x}\|_2} \|_2^2 = \sum_{i,j \in [n]} q_i q_j$.
Since these must be equal,
\[
\sum_{i,j \in [n]} q_iq_j \qty(1-\bm{a}_i^\top \bm{a}_j) = 0.
\]
Since $\bm{a}_i^\top \bm{a}_j \le 1$ and $q_i q_j > 0$ for all $i,j \in [n]$, we must have $\bm{a}_i^\top \bm{a}_j = 1$ for all $i,j \in [n]$, which means that $\bm{a}_i$ is identical across all $i \in [n]$.
Consequently, we have $\bm{A}\bm{q} = R(\bm{q}) \bm{a}_i = \bm{x}$, which in turn implies $\bm{a}_i=\frac{\bm{x}}{\|\bm{x}\|_2}$ for all $i \in [n]$.
\hfill {\it Q.E.D.}


\subsection*{Proof of Lemma~\ref{lem_r}}

The ``if'' part is straightforward.
To prove the ``only if'' part, suppose that $\bm{A}\bm{q}=\bm{x}$ and $\|\bm{x}\|_2 = r(\bm{q}) > 0$.
Note that $r(\bm{q}) > 0$ implies that there exists one and only one $i$ such that $q_i > \sum_{j \neq i} q_j$.
Assume $i=1$ without loss of generality.
By assumption, we have $q_1 \bm{a}_1 + \sum_{i=2}^nq_i \bm{a}_i = (q_1 - \sum_{i=2}^n q_j) \frac{\bm{x}}{\|\bm{x}\|_2}$.
Calculating the $\ell_2$ norm of each side,
\begin{equation} \label{lem_r1}
\left\|q_1 \bm{a}_1 + \sum_{i=2}^nq_i \bm{a}_i \right\|_2 = q_1 - \sum_{i=2}^n q_i.
\end{equation}
By the triangle inequality,
\begin{equation} \label{lem_r2}
\left\|\sum_{i=2}^nq_i \bm{a}_i \right\|_2 \le \sum_{i=2}^n q_i.
\end{equation}
Also, by the reverse triangle inequality,
\begin{equation} \label{lem_r3}
\left\|q_1 \bm{a}_1 + \sum_{i=2}^n q_i \bm{a}_i \right\|_2
\ge q_1 - \left\|\sum_{i=2}^n q_i \bm{a}_i \right\|_2.
\end{equation}
Combining \eqref{lem_r1}, \eqref{lem_r2}, and \eqref{lem_r3}, we have $\|\sum_{i=2}^n q_i \bm{a}_i\|_2 = \sum_{i=2}^n q_i$, which holds true if and only if $\bm{a}_2 =  \cdots = \bm{a}_n \equiv \bm{u}$ for some unit vector $\bm{u}$.
Plugging this into \eqref{lem_r1}, we get $\|q_1\bm{a}_1 + \sum_{i=2}^n q_i \bm{u}\|_2 = q_1 - \sum_{i=2}^n q_i$, which holds true if and only if $\bm{a}_1 = -\bm{u}$.
It is now clear that we must have $\bm{u} = - \frac{\bm{x}}{\|\bm{x}\|_2}$ to satisfy $\sum_{i=1}^n q_i \bm{a}_i = - q_1 \bm{u} + \sum_{i=2}^n q_i \bm{u} = -r(\bm{q})\bm{u} = \frac{r(\bm{q})}{\|\bm{x}\|_2} \cdot \bm{x}$.
\hfill {\it Q.E.D.}


\subsection*{Proof of Lemma~\ref{lem_b}}
Suppose that $\bm{A}\bm{q} = \bm{x}$.
To prove the contrapositive, assume that $\rank(\bm{A})=1$.
This implies that there are some unit vector $\bm{u}$ and $\bm{\sigma} \in \{-1,1\}^n$ such that $\bm{a}_i = \sigma_i\bm{u}$ for all $i \in [n]$.
Then, since $\bm{A}\bm{q} = (\sum_{i=1}^n q_i \sigma_i) \cdot \bm{u} = \bm{x}$, we must have $\bm{u} = \frac{\bm{x}}{\|\bm{x}\|_2}$ and $|\sum_{i=1}^n q_i \sigma_i| = \|\bm{x}\|_2$.
We are done if $\sum_{i=1}^n q_i \sigma_i = \|\bm{x}\|_2$.
Otherwise, we have $(-\bm{\sigma})^\top \bm{q} = \|\bm{x}\|_2$, which establishes the claim.
\hfill {\it Q.E.D.}


\subsection*{Proof of Proposition~\ref{prop_sp_mono}}

Clearly, $r(\lambda \bm{q})=\lambda r(\bm{q})$ and $R(\lambda \bm{q})=\lambda R(\bm{q})$ hold for any $\lambda > 0$.
Thus, it follows that
\[
\rho^\ddagger \qty(\frac{\bm{q}}{2}) = 
\max \qty{r\qty(\frac{\bm{q}}{2}),\, \min\qty{R\qty(\frac{\bm{q}}{2}),\, \frac{1}{2}}}
= \frac{1}{2} \max \qty{r(\bm{q}),\, \min\qty{R(\bm{q}),\, 1}}
= \frac{\rho^\dagger(\bm{q})}{2}.
\]
Hence, we have
\begin{align*}
\Pi\qty(\bm{A}^\ddagger\qty(\frac{\bm{q}}{2}), \frac{\bm{q}}{2})
&=  \alpha \qty(\rho^\ddagger \qty(\frac{\bm{q}}{2}) - \rho^\ddagger \qty(\frac{\bm{q}}{2})^2) 
+ \frac{\bm{q}^\top \bm{\gamma}}{2}-\frac{\bm{q}^\top \bm{q}}{4} \\
&= \frac{1}{2} \qty(\alpha \qty(\rho^\dagger(\bm{q}) - \frac{\rho^\dagger(\bm{q})^2}{2}) 
+ \bm{q}^\top \bm{\gamma} - \frac{\bm{q}^\top \bm{q}}{2})
= \frac{\Omega(\bm{A}^\dagger\qty(\bm{q}), \bm{q})}{2},
\end{align*}
whence $\bm{q}^\dagger$ maximizes $\Omega$ if and only if $\bm{q}^\dagger/2$ maximizes $\Pi$.
Moreover, Lemma~\ref{lem_opt_A} implies that the monopolist's optimal $\bm{A}^\ddagger$ can be the same as the social planner's optimal $\bm{A}^\dagger$.
\hfill {\it Q.E.D.}


\subsection*{Proof of Proposition~\ref{prop_sp}}

To simplify the notation, we write
\[
\Omega(\bm{A}^\dagger(\bm{q}), \bm{q}) = \Omega(\bm{q}) = \alpha \underbrace{\qty(\rho(\bm{q}^\dagger)-\frac{\rho^\dagger(\bm{q})^2}{2})}_{\eqqcolon \; \Omega_1(\bm{q})} + \underbrace{\qty(\bm{q}^\top \bm{\gamma} - \frac{\bm{q}^\top \bm{q}}{2})}_{\eqqcolon \; \Omega_2(\bm{q})}.
\]
Notice that $\rho^\dagger$ is a piecewise linear continuous function, with its right partial derivative at $q_i = 0$ being calculated as follows:
\[
\frac{\6_+ \rho^\dagger(0, \bm{q}_{-i})}{\6 q_i} = \lim_{\epsilon \downarrow 0} \frac{\rho^\dagger(\epsilon, \bm{q}_{-i}) - \rho^\dagger(0, \bm{q}_{-i})}{\epsilon} =
\begin{cases}
-1 &\text{if} \quad 1 < r(0, \bm{q}_{-i}), \\
0 &\text{if} \quad r(0, \bm{q}_{-i}) \le 1 \le R(0, \bm{q}_{-i}), \\
1 &\text{if} \quad R(0, \bm{q}_{-i}) < 1.
\end{cases}
\]
Hence, by the chain rule, it follows that
\[
\frac{\6_+\Omega\qty(0,\bm{q}_{-i})}{\6q_i}
= \underbrace{\qty(1-\rho^\dagger(0,\bm{q}_{-i})) \cdot \frac{\6_+ \rho^\dagger(0, \bm{q}_{-i})}{\6 q_i}}_{\ge \; 0} \; + \; \gamma_i > 0.
\]
This implies that the social planner's optimal output profile must be strictly positive.

\medskip

\eqref{prop_sp1}.
Observe that $\Omega_1$ is maximized when $\rho^\dagger(\bm{q})=1$, while $\Omega_2$ is maximized when $\bm{q}=\bm{\gamma}$.
This upper bound is attained precisely when $\rho^\dagger(\bm{\gamma}) = 1$.
By the definition of $\rho^\dagger$, therefore, $\bm{q}^\dagger = \bm{\gamma}$ is optimal if and only if $r(\bm{\gamma}) \le 1 \le R(\bm{\gamma})$.

\medskip

\eqref{prop_sp2}.
Suppose that $R(\bm{\gamma}) < 1$.
We claim that any $\bm{q}$ with $R(\bm{q}) > 1$ cannot be optimal.
To this end, consider $\bm{q}^\lambda = (1-\lambda)\bm{q} + \lambda \bm{\gamma}$, with $\lambda \in (0,1)$.
Since $\|\bm{q}^\lambda - \bm{\gamma}\|_2 < \|\bm{q} - \bm{\gamma}\|_2$, it follows that $\Omega_2(\bm{q}^\lambda) > \Omega_2(\bm{q})$.
Moreover, since $R(\bm{\gamma}) < 1 < R(\bm{q})$ by assumption, and since $R$ is continuous, the intermediate value theorem implies that there exists some $\lambda \in (0,1)$ such that $R(\bm{q}^\lambda) = 1$.
For this choice of $\lambda$, we have $\rho^\dagger(\bm{q}^\lambda)=1$, and hence $\Omega_1(\bm{q}^\lambda)$ attains its maximal possible value.
Combining these observations yields $\Omega(\bm{q}^\lambda) > \Omega(\bm{q})$.

Therefore, the optimal output profile $\bm{q}^\dagger$ must lie in the set $\{\bm{q} \in \R^n_{++} : R(\bm{q}) \le 1\}$.
On this set, the objective function $\Omega$ is differentiable, with $\rho^\dagger(\bm{q}) = R(\bm{q})$ and thus $\frac{\6 \rho^\dagger(\bm{q})}{\6 \bm{q}} = \1$.\footnote{Strictly speaking, $\rho^\dagger(\bm{q})$ is not differentiable when $R(\bm{q})=1$. However, this does not pose any difficulty, since the left derivative $\frac{\6_- \rho^\dagger(\bm{q})}{\6 q_i}$ is well defined and equal to $1$ at $R(\bm{q})=1$, and any $\bm{q}$ with $R(\bm{q})>1$ is suboptimal.}
Hence, $\bm{q}^\dagger$ is characterized by the first-order condition,
\begin{equation} \label{sp_foc1}
\alpha \qty(1 - R(\bm{q})) + \gamma_i - q_i = 0,
\quad \forall i \in [n].
\end{equation}
Summing \eqref{sp_foc1} over all $i \in [n]$ yields
\[
n\alpha \qty(1-R(\bm{q})) + R(\bm{\gamma}) - R(\bm{q}) =  0
\quad \iff \quad
R(\bm{q}) = \frac{n\alpha + R(\bm{\gamma})}{n\alpha+1}.
\]
Note that $R(\bm{q}) < 1$ holds by $R(\bm{\gamma}) < 1$.
Substituting $R(\bm{q})$ back into \eqref{sp_foc1}, we obtain
\[
q_i = \gamma_i + \frac{\alpha \qty(1-R(\bm{\gamma}))}{n\alpha + 1},
\quad \forall i \in [n].
\]

\medskip

\eqref{prop_sp3}.
Suppose that $r(\bm{\gamma}) > 1$.
First, we show that any $\bm{q}$ with $r(\bm{q}) < 1$ cannot be optimal.
Again, consider $\bm{q}^\lambda = (1-\lambda)\bm{q} + \lambda \bm{\gamma}$, with $\lambda \in (0,1)$, for which $\Omega_2(\bm{q}^\lambda) > \Omega_2(\bm{q})$ holds.
Since $r(\bm{q}) < 1 < r(\bm{\gamma})$ by assumption, and since $r$ is continuous, the intermediate value theorem implies that there exists some $\lambda \in (0,1)$ such that $r(\bm{q}^\lambda) = 1$.
For this choice of $\lambda$, we have $\rho^\dagger(\bm{q}^\lambda)=1$, and hence $\Omega_1(\bm{q}^\lambda)$ attains its maximal possible value.
Combining these observations yields $\Omega(\bm{q}^\lambda) > \Omega(\bm{q})$.

In the following, assume that $\gamma_1 \ge \sum_{i=2}^n \gamma_i+1-\alpha$ without loss of generality.
Let us show that the optimal output profile must be such that the firm $1$ entails the highest output level.
To this end, consider any $\bm{q}$ such that there exists $i \neq 1$ with $q_i > q_1$.
Let $\tilde{\bm{q}}$ be an alternative output profile obtained from $\bm{q}$ by swapping the $1$st and $i$-th coordinates, that is, $\tilde{q}_1 = q_i$, $\tilde{q}_i = q_1$, and $\tilde{q}_j = q_j$ for all $j \neq \{1,i\}$.
Then $R(\tilde{\bm{q}}) = R(\bm{q})$ and $\|\tilde{\bm{q}}\|_\infty = \|\bm{q}\|$, and hence $r(\tilde{\bm{q}}) = r(\bm{q})$.
It follows that $\rho^\dagger(\tilde{\bm{q}})=\rho^\dagger(\bm{q})$ and thus $\Omega_1(\tilde{\bm{q}})=\Omega_1(\bm{q})$.
Moreover, $\tilde{\bm{q}}^\top \tilde{\bm{q}} = \bm{q}^\top \bm{q}$, so the quadratic term in $\Omega_2$ is unchanged.
The only change comes from the linear term:
\[
\tilde{\bm{q}}^\top \bm{\gamma} - \bm{q}^\top \bm{\gamma}
= \qty(\gamma_1 q_i + \gamma_i q_1) - \qty(\gamma_1 q_1 + \gamma_i q_i)
= (\gamma_1 - \gamma_i)(q_i - q_1),
\]
which is strictly positive since $\gamma_1 > \gamma_i$ and $q_i > q_1$.
Hence $\Omega_2(\tilde{\bm{q}})>\Omega_2(\bm{q})$, and thus $\Omega(\tilde{\bm{q}})>\Omega(\bm{q})$.

Therefore, the optimal output profile $\bm{q}^\dagger$ must lie in the set $\{\bm{q} \in \R^n_{++} : r(\bm{q}) \ge 1,\, q_1 = \|\bm{q}\|_\infty\}$.
On this set, we have $\rho^\dagger(\bm{q}) = q_1 - \sum_{j=2}^n q_j$ and thus $\frac{\6 \rho^\dagger(\bm{q})}{\6 \bm{q}} = [1, -1, \ldots, -1]^\top$.\footnote{Again, while $\rho^{\dagger}(\bm{q})$ is not differentiable when $r(\bm{q})=1$, this does not affect our discussion since any $\bm{q}$ with $r(\bm{q}) < 1$ is suboptimal and the right derivatives are given as $\frac{\6_+ \rho^\dagger(\bm{q})}{\6q_1} = 1$ and $\frac{\6_+ \rho^\dagger(\bm{q})}{\6q_i} = -1$ for $i \neq 1$ at $r(\bm{q})=1$.}
Hence, $\bm{q}^\dagger$ is characterized by the first-order condition,
\begin{alignat}{4}
\alpha \qty(1- r(\bm{q})) &+ \gamma_1 & &-q_1 & &= 0, \label{sp_foc2} \\
-\alpha \qty(1 - r(\bm{q})) &+ \gamma_i & &-q_i & &= 0, \quad \forall i \ge 2. \label{sp_foc3}
\end{alignat}
Subtracting from \eqref{sp_foc2} the sum of \eqref{sp_foc3} across all $i \ge 2$, we have
\[
n\alpha \qty(1- r(\bm{q})) + \qty(\gamma_1 - \sum_{i=2}^n \gamma_i) - \qty(q_1 - \sum_{i=2}^n q_i) = 0
\quad \iff \quad
r(\bm{q}) = \frac{n\alpha + r(\bm{\gamma})}{n\alpha + 1}.
\]
Note that $r(\bm{q}) > 1$ holds by $r(\bm{\gamma}) > 1-\alpha$.
Substituting $r(\bm{q})$ back into \eqref{sp_foc2} and \eqref{sp_foc3}, we obtain
\[
q_1 = \gamma_1 - \frac{\alpha(r(\bm{\gamma})-1)}{n\alpha+1}, \qquad
q_i = \gamma_1 + \frac{\alpha(r(\bm{\gamma})-1)}{n\alpha+1}, \quad \forall i \ge 2.
\]

In each case, having established the optimal output profile $\bm{q}^\dagger$, we can identify the associated characteristics profile $\bm{A}^\dagger(\bm{q}^\dagger)$ and $\bm{A}^\dagger(\bm{q}^\dagger)\bm{q}^\dagger$ via Lemmas~\ref{lem_R}--\ref{lem_b}.
Therefore, we have completed the proof of Proposition~\ref{prop_sp}.
\hfill {\it Q.E.D.}


\subsection*{Proof of Lemma~\ref{lem_br}}
Fix any $(\bm{A}_{-i}, \bm{q}_{-i})$, and let $\delta_i = \|\bm{\beta} - \sum_{j \neq i} q_j \bm{a}_j\|$.
As discussed in the main text, the firm $i$'s optimal $\bm{a}_i$ is arbitrary if $\delta_i = 0$, while it should be proportional to $\bm{\beta} - \sum_{j \neq i} q_j \bm{a}_j$ if $\delta_i \neq 0$.
In the latter case, $\bm{a}_i = \frac{1}{\delta_i} \cdot (\bm{\beta} - \sum_{j \neq i} q_j \bm{a}_j)$ holds since $\bm{a}_i$ must have unit length.
Moreover, substituting \eqref{opt_value_pi} into \eqref{eq_profit_sep}, firm $i$'s problem reduces to the following optimization over $q_i$ alone:
\[
\max_{q_i \ge 0} \; \qty(\alpha \delta_i + \gamma_i) q_i - (1+\alpha)q_i^2.
\]
The objective is concave, and the first-order pins down the optimal output level $q_i = \frac{\alpha \delta_i + \gamma_i}{2(1+\alpha)}$, which is strictly positive since $\gamma_i > 0$.

Suppose that $(\bm{A}^*,\bm{q}^*)$ is an equilibrium.
By the best-response condition for $q_i^*$, we have $\alpha \delta_i^* = 2(1+\alpha) q_i^* -\gamma_i$, whereas $\delta_i^* = \delta(\sum_{j \neq i} q_j^* \bm{a}_j^*) \ge 0$ holds by construction.
This implies \eqref{eq_cond1}.
Moreover, we can easily verify \eqref{eq_cond2} by multiplying both sides of the best-response condition for $\bm{a}_i^*$ by $\alpha$ and substituting $\alpha \delta_i^* = 2(1+\alpha)-\gamma_i$ into it.

Conversely, suppose that $(\bm{A}^*,\bm{q}^*)$ satisfies \eqref{eq_cond1} and \eqref{eq_cond2}.
Set $\delta_i = \frac{2(1+\alpha)q_i^*-\gamma_i}{\alpha}$, and notice that $\delta_i \ge 0$ by \eqref{eq_cond1}.
Since $q_i^* = \frac{\alpha \delta_i + \gamma_i}{2(1+\alpha)}$ by construction, the best-response condition for $q_i^*$ is trivially satisfied.
Moreover, since $\bm{A}^*\bm{q}^* = q_i^* \bm{a}_i^* + \sum_{j \neq i}q_j^* \bm{a}_j^*$, we can equivalently write \eqref{eq_cond2} as
\[
\qty(2(1+\alpha)q_i^* - \gamma_i) \bm{a}_i^* = \alpha \bm{\beta} - \alpha \sum_{j \neq i} q_j^* \bm{a}_j^*
\]
Then, substituting $q_i^* = \frac{\alpha \delta_i + \gamma_i}{2(1+\alpha)}$ into this equation, we have $\delta_i \bm{a}_i^* = \bm{\beta} - \sum_{j \neq i} q_j^* \bm{a}_j^*$, which verifies the best-response condition for $\bm{a}_i^*$.
In addition, calculating the $\ell_2$ norm of each side, we obtain $|\delta_i| = \delta_i = \|\bm{\beta} - \sum_{j \neq i} q_j^* \bm{a}_j^*\|_2$.
Therefore, $(\bm{A}^*,\bm{q}^*)$ constitutes an equilibrium.

Finally, consider any equilibrium $(\bm{A}^*, \bm{q}^*)$.
Multiplying condition~\eqref{eq_cond2} by $(\bm{a}_i^*)^\top$ from left and arranging terms, we have
\[
\alpha \qty(\bm{a}_i^*)^\top \qty(\bm{\beta} - \bm{A}^*\bm{q}^*) + \gamma_i - q_i^* = (1+\alpha) q_i^*.
\]
Noticing that the left-hand side represents the firm $i$'s markup, this implies the firm $i$'s markup and profit are, respectively, given by $(1+\alpha)q_i^*$ and $(1+\alpha)q_i^{*2}$.
\hfill {\it Q.E.D.}


\subsection*{Proof of Proposition~\ref{prop_eq_D}}

Note that by Lemma~\ref{lem_donut}, there exists $\bm{A}^{\rm d} \in \calA$ such that $\bm{A}^{\rm d} \bm{q}^{\rm d} = \bm{\beta}$ if and only if $r(\bm{q}^{\rm d}) \le 1 \le R(\bm{q}^{\rm d})$.
By the definition of $\bm{q}^{\rm d}$ in Proposition~\ref{prop_eq_D}, this condition is stated as $r(\bm{\gamma}) \le 2+\alpha \le R(\bm{\gamma})$.
Moreover, evaluated at the same $\bm{q}^{\rm d}$, we can easily see that \eqref{eq_cond1} is always satisfied.
It is also evident that \eqref{eq_cond2} is satisfied for $(\bm{A}^{\rm d}, \bm{q}^{\rm d})$ described in the proposition.
Therefore, by Lemma~\ref{lem_br}, we conclude that $(\bm{A}^{\rm d}, \bm{q}^{\rm d})$ constitutes an equilibrium if and only if $r(\bm{\gamma}) \le 2+\alpha \le R(\bm{\gamma})$.
\hfill {\it Q.E.D.}


\subsection*{Proof of Corollary~\ref{cor_cosine_diff}}

This corollary is derived as a consequence of the following lemma.

\begin{lemma} \label{lem_cosine}
For any $(\bm{A},\bm{q})$ such that $\bm{A}\bm{q} = c\bm{\beta}$ and $\bm{q} = d\bm{\gamma}$, where $c,d > 0$, we have
\[
\bar{s}_{\bm \gamma}(\bm{A}) = \frac{c^2/d^2 - \|\bm{\gamma}\|_2^2}{n(n-1)}.
\]
\end{lemma}

\begin{proof}
By assumption, we have
\begin{equation} \label{lem_cosine1}
d \sum_{j=1}^n \gamma_j \bm{a}_j = c \bm{\beta}
\end{equation}
Taking the inner product with $\bm{a}_i$ on both sides yields
\[
d \sum_{j \neq i} \gamma_j \bm{a}_i^\top \bm{a}_j + d \gamma_i = c \bm{a}_i^\top \bm{\beta}, \quad \forall i \in [n].
\]
Multiplying both sides by $\gamma_i$ and summing over all $i \in [n]$, we obtain
\begin{equation} \label{cosine_diff1}
d \sum_{i=1}^n \sum_{j \neq i} \gamma_i \gamma_j \bm{a}_i^\top \bm{a}_j + d \|\bm{\gamma}\|_2^2 = c \sum_{i=1}^n \gamma_i \bm{a}_i^\top \bm{\beta}.
\end{equation}
At the same time, taking the inner product with $\bm{\beta}$ on both sides of \eqref{lem_cosine1} yields
\begin{equation} \label{cosine_diff2}
d \sum_{i=1}^n \gamma_i \bm{a}_i^\top \bm{\beta} = c.
\end{equation}
Substituting \eqref{cosine_diff2} into \eqref{cosine_diff1}, it follows that
\[
d \sum_{i=1}^n \sum_{j \neq i} \gamma_i \gamma_j \bm{a}_i^\top \bm{a}_j + d \|\bm{\gamma}\|_2^2 = \frac{c^2}{d},
\]
from which
\[
\bar{s}_{\bm{\gamma}}(\bm{A}) = \frac{1}{2\binom{n}{2}}\sum_{i=1}^n \sum_{j \neq i} \gamma_i \gamma_j \bm{a}_i^\top \bm{a}_j
= \frac{c^2/d^2 - \|\bm{\gamma}\|_2^2}{n(n-1)}. \qedhere
\]
\end{proof}

Noticing that $c=d=1$ for $(\bm{A}^\dagger, \bm{q}^\dagger)$, Lemma~\ref{lem_cosine} implies $\bar{s}_{\bm{\gamma}}(\bm{A}^\dagger) = \frac{1 - \|\bm{\gamma}\|_2^2}{n(n-1)}$, and this also implies $\bar{s}_{\bm{\gamma}}(\bm{A}^\ddagger) = \frac{1 - \|\bm{\gamma}\|_2^2}{n(n-1)}$ since $\bm{A}^\dagger = \bm{A}^\ddagger$ by Proposition~\ref{prop_sp_mono}.
Moreover, since $c=1$ and $d=\frac{1}{2+\alpha}$ for $(\bm{A}^{\rm d}, \bm{q}^{\rm d})$, we have $\bar{s}_{\bm{\gamma}}(\bm{A}^*) = \frac{(2+\alpha)^2 - \|\bm{\gamma}\|_2^2}{n(n-1)}$.
\hfill {\it Q.E.D.}


\subsection*{Proof of Proposition~\ref{prop_eq_CP}}
Fix any $\bm{\sigma}\in\{-1,1\}^n$, and let $(\bm{A}^*,\bm{q}^*)$ be a candidate equilibrium profile with $\bm{a}_i = \sigma_i \bm{\beta}$ for each $i \in [n]$.
Since $\bm{A}^*=\bm{\beta}\bm{\sigma}^\top$, the equilibrium condition \eqref{eq_cond2} in Lemma~\ref{lem_br} is written as follows:
\[
\qty((2+\alpha)q_i^* - \gamma_i) \sigma_i \bm{\beta} = \alpha \qty(1 - \bm{\sigma}^\top \bm{q}^*) \bm{\beta},
\]
which holds if and only if
\[
(2+\alpha) \sigma_i q_i^* +\alpha \bm{\sigma}^\top \bm{q}^* = \sigma_i \gamma_i + \alpha.
\]
Multiplying both sides by $\sigma_i$, since $\sigma_i^2=1$, this is equivalently expressed in vector form as
\begin{equation} \label{br_rank1}
(2+\alpha)\bm{q}^* + \alpha \qty(\bm{\sigma}^\top \bm{q}^*)\bm{\sigma}
= \bm{\gamma} + \alpha\bm{\sigma}.
\end{equation}
Thus, $(\bm{A}^*,\bm{q}^*)$ forms an equilibrium if and only if it satisfies \eqref{eq_cond1} and \eqref{br_rank1}.

Multiplying both sides of \eqref{br_rank1} by $\bm{\sigma}^\top$ from left, since $\bm{\sigma}^\top \bm{\sigma} = n$, it follows that
\[
(2+\alpha) \qty(\bm{\sigma}^\top \bm{q}^*) + n\alpha \qty(\bm{\sigma}^\top \bm{q}^*) = \bm{\sigma}^\top \bm{\gamma} + n\alpha,
\]
from which
\[
\bm{\sigma}^\top \bm{q}^* = \frac{\bm{\sigma}^\top \bm{\gamma} + n\alpha}{2+(n+1)\alpha}.
\]
Plugging this back into \eqref{br_rank1}, it follows that
\[
\bm{q}^* = \frac{1}{2+\alpha} \qty(\bm{\gamma} + \frac{\alpha(2+\alpha-\bm{\sigma}^\top \bm{\gamma})}{2+(n+1)\alpha} \cdot \bm{\sigma}).
\]
Thus, condition~\eqref{br_rank1} uniquely pins down $\bm{q}^\dagger$ as in Proposition~\ref{prop_eq_CP}.
Moreover, evaluated at the same $\bm{q}^*$, condition~\eqref{eq_cond1} is expressed as follows:
\begin{align*}
\eqref{eq_cond1} &\quad \iff \quad
\frac{2(1+\alpha)}{2+\alpha} \qty(\bm{\gamma} + \frac{\alpha (2+\alpha - \bm{\sigma^\top \bm{\gamma}})}{2+(n+1)\alpha} \cdot \bm{\sigma}) \ge \bm{\gamma} \\
&\quad \iff \quad \frac{\alpha}{2+\alpha} \bm{\gamma} \ge \frac{2\alpha (1+\alpha) (\bm{\sigma^\top \bm{\gamma}} - (2+\alpha))}{(2+\alpha)(2+(n+1)\alpha)} \cdot \bm{\sigma} \\
&\quad \iff \quad \frac{2+(n+1)\alpha}{2(1+\alpha)} \bm{\gamma} \ge \qty(\bm{\sigma^\top \bm{\gamma}} - (2+\alpha)) \bm{\sigma} \\
&\quad \iff \quad \bm{\sigma^\top \bm{\gamma}} \bm{\sigma} \le (2+\alpha) \bm{\sigma} + \frac{2+(n+1)\alpha}{2(1+\alpha)} \bm{\gamma}.
\end{align*}
Recalling the definitions of $N_+$ and $N_-$, the last line is expanded as
\begin{align}
\bm{\sigma^\top \bm{\gamma}} \le (2+\alpha) + \frac{2+(n+1)\alpha}{2(1+\alpha)} \gamma_i, \quad \forall i \in N_+, \label{eq_rank1_cond1} \\
\bm{\sigma^\top \bm{\gamma}} \ge (2+\alpha) - \frac{2+(n+1)\alpha}{2(1+\alpha)} \gamma_i, \quad \forall i \in N_-. \label{eq_rank1_cond2}
\end{align}
The tightest requirements of \eqref{eq_rank1_cond1} and \eqref{eq_rank1_cond2} arise when $\gamma_i = \min_{j \in N_+} \gamma_j$ and $\gamma_i = \min_{j \in N_-} \gamma_j$, respectively, which together yield condition~\eqref{cond_rank1} in Proposition~\ref{prop_eq_CP}.
\hfill {\it Q.E.D.}


\subsection*{Proof of Theorem~\ref{thm_eq}}

The existence conditions for the product-differentiation, product-concentration, and dominant-firm polarization equilibria are given, respectively, by $r(\bm{\gamma}) \le 2+\alpha \le R(\bm{\gamma})$, \eqref{cond_conc}, and \eqref{cond_pola} with $\tilde{r}_i(\bm{\gamma}) = r(\bm{\gamma})$.
It is immediate that the union of these conditions exhausts all possible parameter values, hence at least one of these three equilibria always exists.

Consider any equilibrium $(\bm{A}^*, \bm{q}^*)$.
If $q_i^* = \frac{\gamma_i}{2+\alpha}$ for some $i$ (and hence for all $i$), then $(\bm{A}^*, \bm{q}^*)$ must coincide with the equilibrium characterized in Proposition~\ref{prop_eq_D}.
Otherwise, if $q_i^* \neq \frac{\gamma_i}{2+\alpha}$ for some $i$, then \eqref{eq_delta} implies that there exist a unit vector $\bm{u}$ and a sign vector $\bm{\sigma} \in \{-1,1\}^n$ such that $\bm{a}_i^* = \sigma_i \bm{u}$ for all $i$.
We want to show that $\bm{u} = \bm{\beta}$ or $\bm{u} = -\bm{\beta}$.
Using $\bm{a}_i = \sigma_i \bm{u} = \sigma_i\qty(\bm{\beta}+(\bm{u}-\bm{\beta}))$, we can rewrite the best-response condition for $\bm{a}_i^*$ as
\[
\sigma_i \delta_i \qty(\bm{\beta} + (\bm{u}-\bm{\beta}))
= \qty(1-\sum_{j \neq i} \sigma_j q_j^*) \bm{\beta}
- \qty(\sum_{j \neq i} \sigma_j q_j^*) \qty(\bm{u}-\bm{\beta}).
\]
Since $\qty(\bm{u}+\bm{\beta})^\top \qty(\bm{u}-\bm{\beta})
= \|\bm{u}\|_2^2 - \|\bm{\beta}\|_2^2 = 0$, multiplying both sides by $\qty(\bm{u}+\bm{\beta})^\top$ from the left yields
\[
\qty(1-\sum_{j \neq i} \sigma_j q_j^* - \sigma_i \delta_i)
\qty(\bm{u}+\bm{\beta})^\top \bm{\beta} = 0.
\]
If $1-\sum_{j \neq i} \sigma_j q_j^* - \sigma_i \delta_i = 0$, then substituting this equality into the best-response condition for $\bm{a}_i^*$ implies that $\bm{u} = \bm{\beta}$.
Otherwise, if $1-\sum_{j \neq i} \sigma_j q_j^* - \sigma_i \delta_i \neq 0$, the above equation implies that $(\bm{u}+\bm{\beta})^\top \bm{\beta} = 0$, or equivalently $\bm{u}^\top \bm{\beta} = -1$.
Since $\|\bm{u}\|_2 = \|\bm{\beta}\|_2 = 1$, the Cauchy-Schwartz equality condition dictates $\bm{u} = -\bm{\beta}$.
\hfill {\it Q.E.D.}


\subsection*{Proof of Theorem~\ref{thm_welfare_D}}

For any allocation $(\bm{A},\bm{q})$ with $\bm{x}=\bm{A}\bm{q}$ and $\bm{y}=\bm{q}$, social welfare is given by
\begin{equation} \label{welfare_pf}
\Omega(\bm{A},\bm{q}) = \alpha \qty(\bm{x}^\top \bm{\beta} - \frac{\bm{x}^\top \bm{x}}{2}) + \bm{y}^\top \bm{\gamma} - \frac{\bm{y}^\top \bm{y}}{2}.
\end{equation}
Evaluating \eqref{welfare_pf} at the monopoly allocation $(\bm{x}^\ddagger,\bm{y}^\ddagger)=(\frac{\bm{\beta}}{2},\frac{\bm{\gamma}}{2})$ under condition~\eqref{gamma_diff}, we obtain
\[
\Omega(\bm{A}^\ddagger,\bm{q}^\ddagger)
= \frac{3\alpha}{8} + \frac{3\|\bm{\gamma}\|_2^2}{8}.
\]
Also, evaluating \eqref{welfare_pf} at the product-differentiation allocation $(\bm{x}^{\rm d},\bm{y}^{\rm d})=(\bm{\beta},\frac{\bm{\gamma}}{2+\alpha})$ yields
\[
\Omega(\bm{A}^{\rm d},\bm{q}^{\rm d})
= \frac{\alpha}{2}
+ \frac{(3+2\alpha)\|\bm{\gamma}\|_2^2}{2(2+\alpha)^2}.
\]
It follows from a direct comparison of these expressions that $\Omega(\bm{A}^\ddagger,\bm{q}^\ddagger) > \Omega(\bm{A}^{\rm d},\bm{q}^{\rm d})$ holds if and only if $\|\bm{\gamma}\|_2^2 > \frac{3\alpha+4}{(2+\alpha)^2}$.
\hfill {\it Q.E.D.}


\subsection*{Proof of Theorem~\ref{thm_welfare_C}}

This theorem corresponds to a special case of Theorem~\ref{thm_app}, which characterizes the welfare ranking between oligopoly and monopoly when $\bm{A}$ is fixed; see Example~\ref{example_app} there.
\hfill {\it Q.E.D.}


\subsection*{Proof of Theorem~\ref{thm_welfare_DCP}}

Observe that \eqref{welfare_pf} is globally maximized at $\bm{x} = \bm{\beta}$ and $\bm{y} = \bm{\gamma}$.
Hence, between any two allocations $(\bm{A},\bm{q})$ and $(\tilde{\bm{A}},\tilde{\bm{q}})$, we have $\Omega(\bm{A},\bm{q}) > \Omega(\tilde{\bm{A}},\tilde{\bm{q}})$ if
\begin{equation} \label{allocation_welfare_l2}
\|\bm{\beta} - \bm{A}\bm{q}\|_2 \le \|\bm{\beta} - \tilde{\bm{A}}\tilde{\bm{q}}\|_2,
\qquad
\|\bm{\gamma} - \bm{q}\|_2 \le \|\bm{\gamma} - \tilde{\bm{q}}\|_2,
\end{equation}
with at least one inequality holding strictly.

Fix any $\bm{\sigma} \in \{-1,1\}^n$, and compare welfare under the two equilibrium allocations $(\bm{A}^{\rm d},\bm{q}^{\rm d})$ and $(\bm{A}^{\bm{\sigma}},\bm{q}^{\bm{\sigma}})$.
Since $\bm{A}^{\rm d}\bm{q}^{\rm d} = \bm{\beta}$ by Proposition~\ref{prop_eq_D}, the first inequality in \eqref{allocation_welfare_l2} is always satisfied.
To evaluate the second inequality, recall that $\bm{q}^{\rm d} = \frac{\bm{\gamma}}{2+\alpha}$ by Proposition~\ref{prop_eq_D}$,$ and that $\bm{q}^{\bm{\sigma}} = \bm{q}^{\rm d} - \phi(\bm{\sigma})\bm{\sigma}$ by Proposition~\ref{prop_eq_CP}.
It follows that
\begin{align*}
\|\bm{\gamma} - \bm{q}^{\bm{\sigma}}\|_2^2
&= \left\|\qty(\frac{1+\alpha}{2+\alpha})\bm{\gamma} + \phi(\bm{\sigma})\bm{\sigma}\right\|_2^2 \\
&= \qty(\frac{1+\alpha}{2+\alpha})^2 \|\bm{\gamma}\|_2^2
+ \frac{2(1+\alpha)\phi(\bm{\sigma})}{2+\alpha}\bm{\sigma}^\top \bm{\gamma}
+ \phi(\bm{\sigma})^2 \|\bm{\sigma}\|_2^2 \\
&= \|\bm{\gamma} - \bm{q}^{\rm d}\|_2^2
+ \phi(\bm{\sigma}) \underbrace{\qty(\frac{2(1+\alpha)\bm{\sigma}^\top \bm{\gamma}}{2+\alpha}
+ n\phi(\bm{\sigma}))}_{\eqqcolon\;\psi(\bm{\sigma})}.
\end{align*}
Therefore, $\|\bm{\gamma} - \bm{q}^{\rm d}\|_2 < \|\bm{\gamma} - \bm{q}^{\bm{\sigma}}\|_2$ if and only if $\phi(\bm{\sigma})\psi(\bm{\sigma}) > 0$.

Consider first the case $\phi(\bm{\sigma}) > 0$.
By construction, this occurs if and only if $\bm{\sigma}^\top \bm{\gamma} > 2+\alpha$.
In this case, we have $\psi(\bm{\sigma}) > 2(1+\alpha) > 0$, and hence $\phi(\bm{\sigma})\psi(\bm{\sigma}) > 0$ whenever $\bm{\sigma}^\top \bm{\gamma} > 2+\alpha$.
This yields the first sufficient condition for $\Omega(\bm{A}^{\rm d},\bm{q}^{\rm d}) > \Omega(\bm{A}^{\bm{\sigma}},\bm{q}^{\bm{\sigma}})$ in Theorem~\ref{thm_welfare_DCP}.
In particular, when $\bm{\sigma} = \bm{1}$, this condition is satisfied under $R(\bm{\gamma}) > 2+\alpha$, which is a necessary condition for the product-differentiation equilibrium $(\bm{A}^{\rm d},\bm{q}^{\rm d})$ to exist.

Next, consider the case $\phi(\bm{\sigma}) < 0$.
In this case, $\phi(\bm{\sigma})\psi(\bm{\sigma}) > 0$ holds if and only if $\psi(\bm{\sigma}) < 0$, which is equivalent to
\begin{align*}
&\frac{2(1+\alpha)\bm{\sigma}^\top \bm{\gamma}}{2+\alpha} + n\phi(\bm{\sigma}) < 0 \\
&\quad \iff \quad
\frac{2(1+\alpha)\bm{\sigma}^\top \bm{\gamma}}{2+\alpha}
+ \frac{n\alpha\qty(\bm{\sigma}^\top \bm{\gamma}-(2+\alpha))}{(2+\alpha)(2+(n+1)\alpha)} < 0 \\
&\quad \iff \quad
2(1+\alpha)\qty(2+(n+1)\alpha)\bm{\sigma}^\top \bm{\gamma}
+ n\alpha\qty(\bm{\sigma}^\top \bm{\gamma}-(2+\alpha)) < 0 \\
&\quad \iff \quad
\bm{\sigma}^\top \bm{\gamma}
< \frac{n\alpha(2+\alpha)}{2(1+\alpha)\qty(2+(n+1)\alpha)+n\alpha}.
\end{align*}
This yields the second sufficient condition.
In particular, since the right-hand side is strictly positive, this condition is always satisfied whenever $\bm{\sigma}^\top \bm{\gamma} \le 0$.
\hfill {\it Q.E.D.}


\subsection*{Proof of Proposition~\ref{prop_network}}

The first-order condition implies that \eqref{eq_welfare_network} is globally maximized at $\bm{x} = \bm{\beta}$ and $\bm{y} = (\bm{I}-\bm{W})^{-1} \bm{\gamma}$.
Thus, the social planner's optimum $(\bm{A}^\dagger, \bm{q}^\dagger)$ is given by $\bm{A}^\dagger \bm{q}^\dagger = \bm{\beta}$ and $\bm{q}^\dagger = {\bf b}(1,\bm{\gamma})$, provided that $r({\bf b}(1,\bm{\gamma})) \le 1 \le R({\bf b}(1,\bm{\gamma}))$ by Lemma~\ref{lem_donut}.
Moreover, summing \eqref{eq_profit_network} across all firms, the monopolist's objective is written as
\[
\alpha \qty(\bm{\beta}^\top \bm{x} - \bm{x}^\top \bm{x}) + \bm{\gamma}^\top \bm{y} - \bm{y}^\top (\bm{I}-\bm{W}) \bm{y},
\]
which is globally maximized at $\bm{x} = \frac{\bm{\beta}}{2}$ and $\bm{y} = \frac{1}{2} (\bm{I}-\bm{W})^{-1} \bm{\gamma}$.
Hence, the monopolist's optimum is given by $\bm{A}^\ddagger = \bm{A}^\dagger$ and $\bm{q}^\ddagger = \frac{\bm{q}^\dagger}{2} = {\bf b}(1,\frac{\bm{\gamma}}{2})$.

Next, from \eqref{eq_profit_network}, each firm's best-response conditions are given by
\[
\delta_i \bm{a}_i = \bm{\beta} - \sum_{j \neq i} q_j \bm{a}_j, \qquad
q_i = \frac{\alpha \delta_i + \gamma_i + \sum_{j \neq i} w_{ij} q_j}{2(1+\alpha)}, \qquad \text{where} \quad \delta_i = \left\|\bm{\beta} - \textstyle \sum_{j \neq i} q_j \bm{a}_j \right\|_2.
\]
The first condition implies that $(\delta_i - q_i) \bm{a}_i = (\delta_j - q_j) \bm{a}_j$ holds for all $i \neq j$ in any equilibrium, whence $\delta_i = q_i$ for all $i$ is necessary for an equilibrium exhibiting product differentiation.
Consequently, the equilibrium output profile in such an equilibrium is characterized by the linear system
\[
\bm{q} = \frac{\alpha \bm{q} + \bm{\gamma} + \bm{W} \bm{q}}{2(1+\alpha)}
\quad \iff \quad
\qty(\bm{I} - \frac{\bm{W}}{2+\alpha}) \bm{q} = \frac{\bm{\gamma}}{2+\alpha},
\]
from which we obtain $\bm{q}^{\rm d} = {\bf b}(\frac{1}{2+\alpha}, \frac{\bm{\gamma}}{2+\alpha})$.
Note that since $\bm{A}^{\rm d}\bm{q}^{\rm d}$ follows from the best-response conditions for $\bm{a}_i$, this equilibrium exists if and only if $r({\bf b}(\frac{1}{2+\alpha}, \frac{\bm{\gamma}}{2+\alpha})) \le 1 \le R({\bf b}(\frac{1}{2+\alpha}, \frac{\bm{\gamma}}{2+\alpha}))$.

Lastly, since $\rho(\bm{W}) < 1$ by assumption, the standard Neumann series expansion for matrices implies that
\[
\bm{q}^\ddagger - \bm{q}^{\rm d}
= \sum_{t=0}^\infty \bm{W}^t \cdot \frac{\bm{\gamma}}{2} - \sum_{t=0}^\infty \qty(\frac{\bm{W}}{2+\alpha})^t \cdot \frac{\bm{\gamma}}{2+\alpha} 
= \qty(\frac{1}{2} - \frac{1}{2+\alpha}) \bm{\gamma} + \sum_{t=1}^{\infty}
\underbrace{\qty(\frac{1}{2} - \qty(\frac{1}{2+\alpha})^{t+1})}_{\ge \; 0} \cdot \; \bm{W}^t \bm{\gamma}.
\]
The first term is strictly positive.
Moreover, if all entries of $\bm{W}$ are nonnegative, then $\bm{W}^t \bm{\gamma}$ is a nonnegative vector for any $t \ge 1$, which implies $\bm{q}^\ddagger \gg \bm{q}^{\rm d}$.
\hfill {\it Q.E.D.}


\subsection*{Proof of Lemma~\ref{lem_br_owner}}

The proof follows analogously to that of Lemma~\ref{lem_br}, with the same arguments applied to each firm's modified objective function \eqref{eq_profit_owner}.
\hfill {\it Q.E.D.}


\subsection*{Proof of Corollary~\ref{cor_owner_sp}}

Consider any ownership structure $\bm{K}$ and equilibrium $(\bm{A}^*, \bm{q}^*)$.
By Lemma~\ref{lem_br_owner}, the best-response condition for $q_i^*$ implies that
$\alpha \delta_i^* = 2(1+\alpha)q_i^* - \gamma_i$,
where $\delta_i^* = \|\bm{\beta} - \sum_{j \neq i} (1+\kappa_{ij})q_j^* \bm{a}_j^*\| \ge 0$.
Since $\kappa_{ii}=1$ by normalization, substituting this expression into the best-response condition for $\bm{a}_i^*$ in Lemma~\ref{lem_br_owner} yields
\begin{equation} \label{equiv_owner}
(2q_i^* - \gamma_i^*) \bm{a}_i^* = \alpha \qty(\bm{\beta} - \sum_{j\in[n]}(1+\kappa_{ij})q_j^* \bm{a}_j^*).
\end{equation}

Suppose not, $(\bm{A}^*, \bm{q}^*)$ attains the first-best, i.e., $\bm{A}^*\bm{q}^* = \bm{\beta}$ and $\bm{q}^* = \bm{\gamma}$.
Substituting these into \eqref{equiv_owner}, the equation reduces to
\[
\gamma_i \bm{a}_i^* + \alpha \sum_{j\in[n]} \kappa_{ij} \gamma_j \bm{a}_j^* = 0.
\]
Multiplying this equation by $(\gamma_i\bm{a}_i^*)^\top$ from the left and summing across all $i$, we obtain
\begin{equation} \label{eq_owner_sp_cont}
\|\bm{\gamma}\|_2^2 + \alpha \sum_{i,j\in[n]} \kappa_{ij} \bigl(\gamma_i\bm{a}_i^* \bigr)^\top \bigl(\gamma_j \bm{a}_j^*\bigr) = 0.
\end{equation}
Denote by $\bm{A} \bullet \bm{B} = \sum_{i,j \in [n]} a_{ij} b_{ij}$ the Frobenius inner product between two matrices $\bm{A} = [a_{ij}]_{n \times n}$ and $\bm{B}=[b_{ij}]_{n \times n}$ of the same size.
Clearly, $\bm{A} \bullet \bm{B} = \frac{\bm{A}+\bm{A}^\top}{2} \bullet \bm{B}$ if $\bm{A}$ is symmetric.
Also, it is a well-known fact that $\bm{A} \bullet \bm{B} \ge 0$ if both $\bm{A}$ and $\bm{B}$ are positive semidefinite.
Therefore, writing $\bm{\Gamma} = \Diag (\gamma_1,\ldots,\gamma_n)$, the second term in \eqref{eq_owner_sp_cont} can be evaluated as
\[
\sum_{i,j\in[n]} \kappa_{ij} \bigl(\gamma_i\bm{a}_i^* \bigr)^\top \bigl(\gamma_j \bm{a}_j^*\bigr)
= \bm{K} \bullet \qty[\qty(\bm{A}\bm{\Gamma})^\top \qty(\bm{A}\bm{\Gamma})]
= \qty[\frac{\bm{K}+\bm{K}^\top}{2}] \bullet \qty[\qty(\bm{A}\bm{\Gamma})^\top \qty(\bm{A}\bm{\Gamma})] \ge 0.
\]
However, since $\|\bm{\gamma}\|_2 > 0$, equation~\eqref{eq_owner_sp_cont} cannot be satisfied, yielding a contradiction.
\hfill {\it Q.E.D.}


\subsection*{Proof of Proposition~\ref{prop_owner}}

Suppose that $\kappa_{ij} = \kappa$ for all $i \neq j$.
Taking the difference between the first-order conditions for $\bm{a}_i$ and $\bm{a}_j$ in Lemma~\ref{lem_br_owner}, we obtain $(\delta_i - (1+\kappa)q_i)\bm{a}_i = (\delta_j - (1+\kappa)q_j)\bm{a}_j$ for all $i \neq j$ in any equilibrium, whence $\delta_i = (1+\kappa)q_i$ for all $i$ is necessary for an equilibrium exhibiting product differentiation.
In this equilibrium, we have $(1+\kappa) \sum_{i\in[n]} q_i^{\rm d} \bm{a}_i^{\rm d} = \bm{\beta}$, or equivalently, $\bm{A}^{\rm d} \bm{q}^{\rm d} = \frac{\bm{\beta}}{1+\kappa}$.
Moreover, substituting $\delta_i = (1+\kappa)q_i$ into the best-response condition for $q_i$, the equilibrium output profile is determined as $\bm{q}^{\rm d} = \frac{\bm{\gamma}}{2+\alpha(1-\kappa)}$.
By Lemma~\ref{lem_donut}, this equilibrium exists if and only if $r(\bm{q}^{\rm d}) \le \|\frac{\bm{\beta}}{1+\kappa}\|_2 \le R(\bm{q}^{\rm d})$, which is equivalently expressed as condition~\eqref{cond_owner} in Proposition~\ref{prop_owner}.
Lastly, by taking $c=\frac{1}{1+\kappa}$ and $d=\frac{1}{2+\alpha(1-\kappa)}$, Lemma~\ref{lem_cosine} implies that $\bar{s}_{\bm{\gamma}}(\bm{A}^{\rm d})$ is calculated as in the proposition.
\hfill {\it Q.E.D.}


\subsection*{Proof of Corollary~\ref{cor_owner_welfare}}

Let $(\bm{A}^{\rm d}, \bm{q}^{\rm d})$ be any equilibrium characterized in Proposition~\ref{prop_owner}.
Evaluating \eqref{welfare_pf} at the resulting allocation $(\bm{x}^{\rm d},\bm{y}^{\rm d})=(\frac{\bm{\beta}}{1+\kappa}, \frac{\bm{\gamma}}{2+\alpha(1-\kappa)})$, we obatin
\[
\Omega(\bm{A}^{\rm d}, \bm{q}^{\rm d}) = \frac{\alpha(1+2\kappa)}{2(1+\kappa)^2} + \frac{\|\bm{\gamma}\|_2^2 (3+2\alpha(1-\kappa))}{2(2+\alpha(1-\kappa))^2}.
\]
Taking the derivative of this expression with respect to $\kappa$, tedious yet straightforward algebra yields
\[
\frac{\6\Omega(\bm{A}^{\rm d}, \bm{q}^{\rm d})}{\6 \kappa} = \alpha \cdot \qty[\frac{1+\alpha(1-\kappa)}{(2+\alpha(1-\kappa))^3} \cdot \|\bm{\gamma}\|_2^2 - \frac{\kappa}{(1+\kappa)^3}].
\]
Hence, we have $\frac{\6\Omega(\bm{A}^{\rm d}, \bm{q}^{\rm d})}{\6 \kappa} > 0$ if and only if
\[
\|\bm{\gamma}\|_2^2 > \frac{\kappa}{(1+\kappa)^3} \cdot \frac{(2+\alpha(1-\kappa))^3}{1+\alpha(1-\kappa)} = \underbrace{\qty(\frac{\kappa(2+\alpha(1-\kappa))}{(1+\kappa)(1+\alpha(1-\kappa))})}_{=\; f(\alpha,\kappa)^2} \; \cdot \; \qty( \frac{2+\alpha(1-\kappa)}{1+\kappa})^2,
\]
which is equivalent to condition~\eqref{cond_owner_welfare} in Corollary~\ref{cor_owner_welfare}.
It is clear that $f(\alpha, \kappa) \ge 0$ for any $\alpha > 0$ and $\kappa \in [0,1]$.
Moreover, since
\[
(1+\kappa)(1+\alpha(1-\kappa)) - \kappa(2+\alpha(1-\kappa)) = (1+\alpha)(1-\kappa) \ge 0,
\]
we have $f(\alpha, \kappa) \le 1$ for any $\alpha > 0$ and $\kappa \in [0,1]$.
\hfill {\it Q.E.D.}


\section{Welfare Analysis Without Characteristics Design}
\label{app_q}

This appendix is devoted to a welfare analysis in the baseline model of \cite{pellegrino2025}, where the characteristics profile $\bm{A}$ is exogenously fixed.
In this setting, as noted in Online Appendix~C of \cite{pellegrino2025}, the presence of idiosyncratic characteristics is redundant, since $\bm{y}$ can be subsumed into $\bm{x}$ by adding $n$ extra rows to the matrix $\bm{A}$.
Given this observation, social welfare can be written as a function of the output profile $\bm{q}$ as
\[
\Omega(\bm{q}) = \bm{\psi}^\top \bm{q} - \frac{1}{2} \bm{q}^\top \bm{\Sigma} \bm{q},
\]
where $\bm{\psi}$ and $\bm{\Sigma}$ correspond to $\bm{b}_A - \bm{c} = \alpha\bm{A}^\top\bm{\beta} + \bm{\gamma}$ and $\bm{\Sigma}_{\bm{A}} = \alpha \bm{A}^\top \bm{A} + \bm{I}$, respectively, using the notation of Section~\ref{sec_model}.
Throughout this appendix, we take $(\bm{\psi}, \bm{\Sigma})$ as primitives.
We assume that $\bm{\psi} \gg \bm{0}$ and $\bm{\Sigma}$ is positive definite.
In addition, assume that the diagonal entries of $\bm{\Sigma}$ are equal to some common constant $\bar{\lambda} \coloneqq \sigma_{11} = \cdots = \sigma_{nn} > 1$, since it corresponds to $1+\alpha$ in the main text.
Note that $\bar{\lambda}$ represents the average of the $n$ eigenvalues of $\bm{\Sigma}$, since $\tr(\bm{\Sigma}) = n\bar{\lambda}$.

For any output profile $\bm{q} \ge \bm{0}$, firm~$i$'s profit is given by
\[
\Pi_i(\bm{q}) = \qty(\gamma_i - \sum_{j=1}^n \sigma_{ij} q_j) q_i .
\]
The Cournot game in which firms simultaneously choose output levels is a special case of canonical linear-quadratic games \citep{bcz2006}, and its unique equilibrium can be derived by a routine calculation.
Specifically, the social planner's optimal profile $\bm{q}^\dagger$, which maximizes $\Omega$, the monopolist's optimal profile $\bm{q}^\ddagger$, which maximizes $\sum_{i=1}^n \Pi_i$, and the equilibrium profile $\bm{q}^*$ are given, respectively, by
\begin{equation} \label{output_app_q}
\bm{q}^\dagger = \bm{\Sigma}^{-1}\bm{\psi}, \qquad
\bm{q}^\ddagger = \frac{1}{2}\bm{\Sigma}^{-1}\bm{\psi}, \qquad
\bm{q}^* = \qty(\bm{\Sigma}+\bar{\lambda}\bm{I})^{-1}\bm{\psi}.
\end{equation}
Moreover, evaluating social welfare at each of these profiles yields
\begin{equation} \label{welfare_app_q}
\Omega(\bm{q}^\dagger) = \bm{\psi}^\top S^\dagger \bm{\psi}, \qquad
\Omega(\bm{q}^\ddagger) = \bm{\psi}^\top S^\ddagger \bm{\psi}, \qquad
\Omega(\bm{q}^*) = \bm{\psi}^\top S^* \bm{\psi}, \qquad
\end{equation}
where
\[
\bm{S}^\dagger = \frac{1}{2}\bm{\Sigma}^{-1}, \qquad
\bm{S}^\ddagger = \frac{3}{8}\bm{\Sigma}^{-1}, \qquad
\bm{S}^* = \qty(\bm{\Sigma}+\bar{\lambda}\bm{I})^{-1}
- \frac{1}{2}\qty(\bm{\Sigma}+\bar{\lambda}\bm{I})^{-1} \bm{\Sigma} \qty(\bm{\Sigma}+\bar{\lambda}\bm{I})^{-1}.
\]
It then follows that $\frac{\Omega(\bm{q}^\ddagger)}{\Omega(\bm{q}^\dagger)} = \frac{3}{4}$, meaning that the monopolist's profile always attains $75\%$ of the social optimum.
Also, $\frac{\Omega(\bm{q}^*)}{\Omega(\bm{q}^\dagger)} < 1$ holds by the definition of $\bm{q}^\dagger$.

To uncover the ranking between $\Omega(\bm{q}^\ddagger)$ and $\Omega(\bm{q}^*)$, we study the spectral properties of the relevant matrices.\footnote{This exercise is closely related to the spectral approach to welfare analysis in network games developed by \cite{bko2015} and surveyed in \cite{golub2025}, but differs in its objective. These papers evaluate the aggregate payoffs of players, which in the present setting corresponds to the monopolist's profit $\sum_{i=1}^n \Pi_i$.
By contrast, our analysis takes social welfare as the objective and uses the monopolist outcome as the benchmark.}
Allowing for multiplicity, let $\lambda_1,\ldots,\lambda_n$ denote the $n$ eigenvalues of $\bm{\Sigma}$, and let $\bm{u}_1,\ldots,\bm{u}_n$ denote the corresponding orthonormal eigenvectors.
We order these eigenvalues in descending order, and let $k$ be the last index such that the corresponding eigenvalue exceeds the average $\bar{\lambda}$, i.e.,
\[
\lambda_1 \ge \lambda_2 \ge \cdots \ge \lambda_k > \bar{\lambda} \ge \lambda_{k+1} \ge \cdots \ge \lambda_n > 0.
\]
Since both $\bm{S}^\ddagger$ and $\bm{S}^*$ are rational functions of $\bm{\Sigma}$, they share the same set of orthonormal eigenvectors $\{\bm{u}_i\}_{i=1}^n$, with the associated $i$-th eigenvalues given by $\frac{3}{8\lambda_i}$ and $\frac{1}{\lambda_i+\bar{\lambda}} - \frac{\lambda_i}{2(\lambda_i+\bar{\lambda})^2}$, respectively.
In particular, this implies that these matrices are simultaneously diagonalized as
\[
\bm{S}^\ddagger = \sum_{i=1}^n \frac{3}{8\lambda_i} \bm{u}_i \bm{u}_i^\top, \qquad 
\bm{S}^* = \sum_{i=1}^n \frac{\lambda_i + 2\bar{\lambda}}{2(\lambda_i+\bar{\lambda})^2} \bm{u}_i \bm{u}_i^\top.
\]
Substituting these expressions into \eqref{welfare_app_q}, it follows that
\begin{align}
&\Omega(\bm{q}^*) - \Omega(\bm{q}^\ddagger)
= \bm{\psi}^\top \qty(\bm{S}^* - \bm{S}^\ddagger) \bm{\psi} \nonumber \\
&\qquad= \sum_{i=1}^n \qty(\frac{\lambda_i + 2\bar{\lambda}}{2(\lambda_i+\bar{\lambda})^2} - \frac{3}{8\lambda_i}) \cdot \qty(\bm{\psi}^\top \bm{u}_i)^2
= \frac{1}{8}\sum_{i=1}^n \underbrace{\qty(\frac{(\lambda_i + 3\bar{\lambda})(\lambda_i - \bar{\lambda})}{\lambda_i(\lambda_i+\bar{\lambda})^2})}_{\eqqcolon \; w_i} \; \cdot \; \qty(\bm{\psi}^\top \bm{u}_i)^2. \label{welfare_ranking_app}
\end{align}
By the definition of $k$, we have $w_i \ge 0$ for all $i \le k$, while $w_i < 0$ for all $i > k$.
This observation leads to the following result.

\begin{theorem} \label{thm_app}
We have $\Omega(\bm{q}^*) > \Omega(\bm{q}^\ddagger)$ if and only if
\begin{equation} \label{cond_app}
\sum_{i=1}^k \qty|w_i| \cdot \qty(\bm{\psi}^\top \bm{u}_i)^2 > \sum_{i=k+1}^n \qty|w_i| \cdot \qty(\bm{\psi}^\top \bm{u}_i)^2.
\end{equation}
In particular, if $\bm{\psi} \in \Span \{\bm{u}_1,\ldots,\bm{u}_k \}$, then $\Omega(\bm{q}^*) > \Omega(\bm{q}^\ddagger)$.
\end{theorem}

This theorem characterizes the welfare ranking between oligopoly and monopoly through the key condition~\eqref{cond_app}.
Roughly speaking, this condition is satisfied when the vector $\bm{\psi}$ is well aligned with the \emph{major components} of $\bm{\Sigma}$, i.e., those eigenvectors $\bm{u}_1,\ldots,\bm{u}_k$, whose associated eigenvalues exceed the average $\bar{\lambda}$.

The condition can be interpreted through the principal-component decomposition of $\bm{\Sigma}$ and the projection of $\bm{\psi}$ onto the subspace spanned by the major components.
This interpretation parallels the argument used in \cite{galeottietal2020} to explain the structure of optimal targeting interventions in network games.
Specifically, since $\{\bm{u}_i\}_{i=1}^n$ forms an orthonormal basis of $\R^n$, we can decompose $\bm{\psi}$ as
\[
\bm{\psi} = \sum_{i=1}^n \qty(\bm{\psi}^\top \bm{u}_i) \cdot \bm{u}_i,
\]
that is, as a linear combination of the eigenvectors of $\bm{\Sigma}$, where the coefficient on each $\bm{u}_i$ is given by the inner product $\bm{\psi}^\top \bm{u}_i$.
This inner product is proportional to the cosine similarity, thereby capturing the extent to which $\bm{\psi}$ is correlated with $\bm{u}_i$.

Since the eigenvectors $\{\bm{u}_i\}_{i=1}^n$ are ordered according to descending eigenvalues, the early components correspond to the principal directions of the linear transformation induced by $\bm{\Sigma}$.
Condition~\eqref{cond_app} then holds---so that oligopoly yields higher welfare than monopoly---when $\bm{\psi}$ is more strongly correlated with the major components of $\bm{\Sigma}$ whose associated eigenvalues exceed the average level $\bar{\lambda}$ than with the minor components whose eigenvalues lie below $\bar{\lambda}$.
Formally, the condition requires that the weighted sum of squared inner products with the former exceeds that with the latter, where the weights are determined by the corresponding eigenvalues.
In particular, when $\bm{\psi}$ lies in the span of the major components $\{\bm{u}_1,\ldots,\bm{u}_k\}$, it has zero correlation with the minor components, and condition~\eqref{cond_app} holds trivially.

As an important special case, the present analysis nests the case of product concentration considered in Theorem~\ref{thm_welfare_C}, where all common characteristics vectors align with $\bm{\beta}$.
The following example formally analyzes this case and thereby establishes the proof of Theorem~\ref{thm_welfare_C}.

\begin{example} \label{example_app}
Recall that $\bm{\Sigma}$ and $\bm{\psi}$ correspond to $\alpha \bm{A}^\top \bm{A} + \bm{I}$ and $\alpha \bm{A}^\top \bm{\beta} + \bm{\gamma}$ in the main text.
Focusing on the case of product concentration, that is, $\bm{A} = \bm{\beta}\1^\top$, these expressions reduce to $\bm{\Sigma} = \alpha \bm{J} + \bm{I}$ and $\bm{\psi} = \alpha \1 + \bm{\gamma}$.
Observe that the eigenvalues of $\bm{\Sigma}$ are given by $\lambda_1 = 1+n\alpha$ and $\lambda_2 = \cdots = \lambda_n = 1$.
It then follows that $\bar{\lambda} = 1+\alpha$ and $k=1$.

Consequently, the weights $\{w_i\}_{i=1}^n$ in Theorem~\ref{thm_app} are given by
\[
w_1 = \frac{\alpha(n-1)(4+(n+3)\alpha)}{(1+n\alpha)(2+(n+1)\alpha)^2}, \qquad
w_2 = \cdots = w_n = - \frac{\alpha(4+3\alpha)}{(2+\alpha)^2}.
\]
Moreover, the eigenvector associated with $\lambda_1$ is the uniform unit vector $\bm{u}_1=\1/\sqrt{n}$.
It therefore follows that
\[
\bm{\psi}^\top \bm{u}_1 = \sqrt{n}\alpha + \frac{R(\bm{\gamma})}{\sqrt{n}}.
\]
In addition, using Parseval's identity $\|\bm{\psi}\|_2^2 = \sum_{i=1}^n (\bm{\psi}^\top \bm{u}_i)^2$, we obtain
\begin{align*}
\sum_{i=2}^n \qty(\bm{\psi}^\top \bm{u}_i)^2
&= \|\bm{\psi}\|_2^2 - \qty(\bm{\psi}^\top \bm{u}_1)^2 \\
&= \qty(n\alpha^2 + 2\alpha R(\bm{\gamma}) + \|\bm{\gamma}\|_2^2)
 - \qty(n\alpha^2 + 2\alpha R(\bm{\gamma}) + \frac{R(\bm{\gamma})^2}{n}) \\
&= \|\bm{\gamma}\|_2^2 - \frac{R(\bm{\gamma})^2}{n}
= n \Var\qty[\bm{\gamma}].
\end{align*}

Putting these expressions into Theorem~\ref{thm_app}, we obtain
\begin{align}
&\Omega(\bm{q}^*) > \Omega(\bm{q}^\ddagger)
\quad \iff \quad
w_1 \cdot \qty(\bm{\psi}^\top \bm{u}_1)^2
> |w_2| \cdot \sum_{i=2}^n \qty(\bm{\psi}^\top \bm{u}_i)^2
\nonumber \\
&\quad \iff \quad
\qty(\frac{\alpha(n-1)(4+(n+3)\alpha)}{(1+n\alpha)(2+(n+1)\alpha)^2})
\cdot \qty(\sqrt{n}\alpha + \frac{R(\bm{\gamma})}{\sqrt{n}})^2
> \qty(\frac{\alpha(4+3\alpha)}{(2+\alpha)^2})
\cdot n \Var\qty[\bm{\gamma}]
\nonumber \\
&\quad \iff \quad
\frac{(n-1)(4+(n+3)\alpha)\qty(\alpha + \frac{R(\bm{\gamma})}{n})^2}
{(1+n\alpha)(2+(n+1)\alpha)^2}
>
\frac{(4+3\alpha)\Var\qty[\bm{\gamma}]}{(2+\alpha)^2}.
\label{app_sym_conc}
\end{align}
Observe that \eqref{app_sym_conc} is satisfied when $\Var\qty[\bm{\gamma}]$ is sufficiently small.

To examine how \eqref{app_sym_conc} depends on $\alpha$ and $n$, note that \eqref{gamma_conc} implies the following bounds on $R(\bm{\gamma})$ and $\Var\qty[\bm{\gamma}]$:
\[
0 \le R(\bm{\gamma}) \le 1, \qquad
0 \le \Var\qty[\bm{\gamma}] \le \frac{n-1}{n^2}.
\]
A sufficient condition for \eqref{app_sym_conc} is therefore
\[
\frac{\alpha^2(n-1)(4+(n+3)\alpha)}{(1+n\alpha)(2+(n+1)\alpha)^2}
>
\frac{4+3\alpha}{(2+\alpha)^2} \cdot \frac{n-1}{n^2},
\]
or equivalently,
\[
f \coloneqq \alpha^2 n^2 (2+\alpha)^2 (4+(n+3)\alpha)
>
(1+n\alpha)(4+3\alpha)(2+(n+1)\alpha)^2
\eqqcolon g.
\]
Observe that the degree of $f$ in $\alpha$ is $5$, which exceeds the degree of $g$ in $\alpha$, equal to $4$, and that the leading coefficients of both polynomials are strictly positive.
It therefore follows that \eqref{app_sym_conc} holds for sufficiently large $\alpha$.
Moreover, both $f$ and $g$ have degree $4$ in $n$, but the leading coefficient of $n$ in $f$ is $\alpha^3(2+\alpha)^2$, which is strictly larger than that in $g$, given by $\alpha^3(4+3\alpha)$.
Hence, \eqref{app_sym_conc} also holds for sufficiently large $n$.
\end{example}

\bibliography{reference}

\end{document}